\documentclass[reqno]{amsart}
\usepackage{fullpage,amsfonts,amssymb,amsthm,mathrsfs,graphicx,color,framed,esint,stackrel,amsmath}
\usepackage{cancel,bm}
\usepackage{tikz}
\usepackage{bbm}
\usetikzlibrary{decorations.markings}
\usetikzlibrary{decorations.markings,decorations.pathmorphing}
\def\xr{4}
\def\yr{1}

\usepackage{booktabs,longtable}
\usepackage[initials]{amsrefs}

\DeclareMathAlphabet\mathbfcal{OMS}{cmsy}{b}{n}

\usepackage[colorlinks=true, pdfstartview=FitV, linkcolor=blue, citecolor=blue, urlcolor=blue]{hyperref}
\newcommand{\be}{\begin{equation}}
\newcommand{\ee}{\end{equation}}
\newcommand{\bea}{\begin{eqnarray}}
\newcommand{\eea}{\end{eqnarray}}
\newcommand{\bes}{\begin{equation*}}
\newcommand{\ees}{\end{equation*}}
\newcommand{\beas}{\begin{eqnarray*}}
\newcommand{\eeas}{\end{eqnarray*}}

\renewcommand{\d}{{\mathrm d}}

\newcommand{\im}{\mathrm{i}}
\newcommand{\e}{\mathrm{e}}

\def\tr{\mathop{\mathrm{tr}}\limits}

\newtheorem{theo}{Theorem}[section]
\newtheorem{thm}{Theorem}[section]

\newtheorem{quest}[theo]{Question}
\newtheorem{lem}[theo]{Lemma}
\newtheorem{rem}[theo]{Remark}
\newtheorem{problem}[theo]{Riemann-Hilbert Problem}
\newtheorem{remark}[theo]{Remark}
\newtheorem{prop}[theo]{Proposition}
\newtheorem{cor}[theo]{Corollary}

\newtheorem{defn}[thm]{Definition}

\DeclareMathOperator{\sgn}{sgn}

\newcommand{\E}{\mathbb{E}}

\DeclareFontFamily{U}{mathx}{}
\DeclareFontShape{U}{mathx}{m}{n}{<-> mathx10}{}
\DeclareSymbolFont{mathx}{U}{mathx}{m}{n}
\DeclareMathAccent{\widecheck}{0}{mathx}{"71}

\begin{document}
\title{Joint moments of characteristic polynomials in the circular Jacobi ensemble and Painlev\'e equations}
\author{Thomas Bothner}
\address{School of Mathematics, University of Bristol, Fry Building, Woodland Road, Bristol, BS8 1UG, United Kingdom}
\email{thomas.bothner@bristol.ac.uk}

\author{Fei Wei}
\address{School of Mathematical Sciences, Beijing University of Posts and Telecommunications, Beijing, China}
\email{weif0831@gmail.com}

\date{\today}


\begin{abstract}

In this paper, we establish a connection between joint moments of characteristic
polynomials and their derivatives in the Circular Jacobi Ensemble, a
generalisation of the Circular Unitary Ensemble, and solutions of nonlinear
Painlev\'e equations. For finite $N$, we show that these joint moments are
characterised by a solution of the $\sigma$-Painlev\'e V equation for all real
moment exponents in their admissible range. Under an appropriate large-$N$ scaling limit, we further prove that the
limiting joint moments admit a representation in terms of a solution of the
$\sigma$-Painlev\'e {\rm III}$'$ equation for a certain range of moment
exponents. As applications, we answer a question posed by Assiotis et al.\
\cite{AGS} by showing that the characteristic function of a distinguished
random variable is connected with the $\sigma$-Painlev\'e {\rm III}$'$
equation for complex parameters. Furthermore, we investigate joint moments involving higher-order derivatives
of characteristic polynomials in the Circular Jacobi Ensemble. As a
consequence, we extend a result of Assiotis et al.\ \cite{AGKW} concerning
the joint moments of a sequence of random variables arising from the
ergodic decomposition of Hua--Pickrell measures, from real parameters to
complex parameters.
\end{abstract}

\maketitle

\tableofcontents

\section{Introduction}

The joint moments of characteristic polynomials and their derivatives for Haar-distributed random unitary matrices have attracted significant attention over the past two decades; see, for example, \cite{conreyetal,Dehaye2008,Forrester_2006,Hughes,HKO,keating2000random,W}, and more recently \cite{AAKP,ACRS,ABGS,AGKW,AGKW2,AKW,Bailey_2019,B,Basor_2019,KW,keating-fei,SW}. One major reason for this interest originates from the close connection of such polynomials with the corresponding moments of the Riemann zeta function. Another important reason is the rich interplay with other research areas, such as random variables arising from the Hua–Pickrell determinantal point process \cite{AGKW,AKW}, and integrable systems of Painlev\'e type \cite{ABGS,AGKW,Bailey_2019,Basor_2019,keating-fei}.\smallskip

In this paper, we investigate the joint moments of the characteristic polynomial and its derivative for a generalisation of the Circular Unitary Ensemble, which is the unitary group $\mathbb{U}(N)$ endowed with the probability measure
\bea\label{generalized measure}
\mu_{N}^{(\delta)}
:=
\mathrm{const}\cdot
\det(I-U)^{\overline{\delta}}
\det(I-U^{-1})^{\delta}
\times
(\text{the Haar measure on } \mathbb{U}(N)),
\eea
where $\delta\in\mathbb{C}$ with $\textnormal{Re}(2\delta)>-1$. When $\delta=0$, this setup reduces to the Circular Unitary Ensemble. The measure \eqref{generalized measure} is referred to as the
finite-dimensional Hua--Pickrell measure on the unitary group by Borodin and
Olshanski \cite{Borodin_2001}. They studied the corresponding Hua--Pickrell
measures on infinite-dimensional Hermitian matrices and their ergodic
properties. We will gather more background for \eqref{generalized measure} in Subsection \ref{background}. In the work of Forrester and Witte \cite{WF}, the measure \eqref{generalized measure} is referred to as the cJUE distribution. In Bourgade, Nikeghbali and Rouault \cite{BNR2}, the measure \eqref{generalized measure} is referred to as the analogue of the Ewens measure, and the unitary group endowed with this measure is referred to as the Circular Jacobi ensemble. Following \cite{WF,BNR2}, we shall refer to $\mathbb{U}(N)$ endowed with the measure \eqref{generalized measure} as the Circular Jacobi Ensemble (CJE), and there is a substantial body of literature related to \eqref{generalized measure}; see, for example, \cite{B,BNR1,BNR2,H,N,P1,P2}. In particular, the joint eigenvalue probability density function of the CJE is given by the following probability measure, 
\bea\label{definitionofJacobimeasure}
\nu_{N}^{(\delta)}(\mathrm{d}\theta_{1},\ldots,\mathrm{d}\theta_{N})=\frac{1}{\mathrm{D}_{N,\delta}}\prod_{1\leq j<k\leq N}\left|\mathrm{e}^{\im\theta_{j}}-e^{\im\theta_{k}}\right|^2\prod_{j=1}^{N}\big(1-\mathrm{e}^{-\im\theta_{j}}\big)^{\delta}\big(1-\mathrm{e}^{\im\theta_{j}}\big)^{\overline{\delta}}\mathrm{d}\theta_j,\ \ \ \theta_j\in[0,2\pi),
\eea
where $\mathrm{D}_{N,\delta}>0$ is a normalisation constant given by
$
\mathrm{D}_{N,\delta}
=
2^{N^2+N(\delta+\overline{\delta})}
N!\,\mathrm{C}_{N,\overline{\delta}},
$
with $\mathrm{C}_{N,\overline{\delta}}$ defined below in
$(\ref{definitionoftheconstant0407})$. We are interested in a particular joint moment integral, taken against the measure \eqref{generalized measure}. Namely, for $A\in\mathbb{U}(N)$ and $(s,h)\in\mathbb{R}\times\mathbb{C}$ with $-1<\textnormal{Re}(2h)<2s+1$, the integral
\bea\label{object}
F_{N,\delta}(s,h):=\int_{\mathbb{U}(N)}|Z_A'(0)|^{2h}|Z_A(0)|^{2s-2h}\mathrm{d}\mu_{N}^{(\delta)}(A).
\eea
Here, $Z_A:[0,2\pi)\rightarrow\mathbb{R}$ relates to the ordinary characteristic polynomial $V_A(\theta)$ of $A\in\mathbb{U}(N)$ via
\begin{equation}\label{polydef}
	Z_A(\theta):=\exp\left(\frac{\im N}{2}(\theta+\pi)-\frac{\im}{2}\sum_{j=1}^N \theta_j\right) V_A(\theta);\ \ \ \ V_A(\theta):=\det\big(I-\e^{-\im\theta}A\big)=\prod_{j=1}^N\Big(1-\e^{\im(\theta_j-\theta)}\Big),
\end{equation}
and $\{\theta_j\}_{j=1}^N\subset[0,2\pi)$ are the eigenangles of $A$. In particular, $|Z_A(\theta)|=|V_A(\theta)|$ on $[0,2\pi)$. In recent works by Forrester \cite{F} and Assiotis, Gunes, Soor \cite{AGS} an analogue of \eqref{object}, denoted $F_{N,\beta,\delta}(s,h)$, was investigated, wherein \eqref{generalized measure} in the right hand side of \eqref{object} is replaced by its analogue in the Circular Jacobi $\beta$-Ensemble. For $\beta=2$, this is simply \eqref{object}. When $\delta = 0$, an explicit formula for $F_{N,\beta,0}(s,h)$ with $2h\in\mathbb{N}$ is given in \cite{F} in terms of combinatorial sums. For general $\delta$ and admissible real-valued $(s,h)$, the convergence of $F_{N,\beta,\delta}(s,h)$ (after suitable scaling) as $N \to \infty$ has been investigated in \cite[Theorem $1.1$]{AGS}.\smallskip

A central theme of this paper is to show that \eqref{object} admits
a representation in terms of Painlev\'e special functions
and to
identify the corresponding $\sigma$-Painlev\'e structures. To the best of our knowledge, for $\delta \neq 0$, it is unknown whether explicit, of Painlev\'e-type or otherwise, formulae\,exist for $F_{N,\beta,\delta}(s,h)$ when $\beta = 2$, either for integer exponents $(s,h)$ or, more generally, for non-integer ones. Addressing this question is one of the main goals of this paper. 

\subsection{$F_{N,\delta}(s,h)$ through the lens of integrable systems}
We begin with our first main result, which expresses \( F_{N,\delta}(s,h) \) in terms of a solution to the \( \sigma \)-Painlev\'e V equation.

\begin{prop}\label{mainthreorem}
Let $\delta \in \mathbb{C}$ with $\textnormal{Re}(2\delta) > -1$. Let $s \in \mathbb{R}$ and  with $2s + \textnormal{Re}(2\delta) > -1$. Set $w = s + \delta$, and let $N \in\mathbb{N}$. Let $\sigma_{N}(z;w)$ be the solution to the $\sigma$-Painlev\'e V equation
\begin{align}\label{transformpainleve}
	\left( 2z{\frac {{\rm d}^{2}\sigma_{N}}{{\rm d}{z}^{2}}}   \right) ^{2} 
= &\left(2\sigma_{N}-N(w+\overline{w})-N^2 +4\, \left( {\frac {{\rm d}\sigma_{N}}{{\rm d}z}} 
 \right) ^{2}-\left( 2z +2w -2\,\overline{w}   \right) {\frac {{\rm d}\sigma_{N}}{
{\rm d}z}}  \right) ^{2} \nonumber\\
&-\left( N-2\,{\frac {{\rm d}\sigma_{N}}{{\rm d}z}} 
 \right)  \left( N+ 2\,{\frac {{\rm d}\sigma_{N}}{{\rm d}z}}  
 \right)  \left( N+2\,\overline{w}+2\,{\frac {{\rm d}\sigma_{N}}{{\rm d}z}}  \right)  \left( N+2w-2\,{\frac {{\rm d}\sigma_{N}}{{\rm d}z}} \right),
\end{align}
given explicitly by
\bea\label{0403formula}
\sigma_{N}(z;w)=\frac{Nz}{2}+z\,n_1^{11}\!\left(-z,\frac{1}{2}(N+w),\frac{1}{2}(N+\overline{w}),N\right),
\eea
where $n_1^{11}(z,\mu,\nu,n)$ is the $(1,1)$-entry of $n_1(z,\mu,\nu,n)$ appearing in the unique solution to RHP \ref{RHP2}. Then 
for any $h\in \mathbb{R}$ with $0\leq 2h<2s+\textnormal{Re}(2\delta)+1$, or for any $h\in\mathbb{C}\setminus\mathbb{R}$
with $0<\textnormal{Re}(2h)<2s+\textnormal{Re}(2\delta)+1$ and $\textnormal{Re}(h)\notin \mathbb{Z}_{\geq 1}$,
\begin{align*}
F_{N,\delta}(s,h)
 =F_{N,\delta}(s,0)\times
 \begin{cases} \vspace{.2cm}
\dfrac{(-1)^{h}}{2}\displaystyle \dfrac{\d^{2h}}{\d\lambda^{2h}}\big(g_{1}(\lambda;w)+g_{2}(\lambda;w)\big)\bigg|_{\lambda=0}, \quad h \in \mathbb{Z}_{\geq 0}\bigskip, \\[1.0ex]
-\dfrac{\sin(\pi h)}{\pi} \Gamma(2h-2M)\displaystyle\int_0^\infty \dfrac{\d^{2M+1}}{\d\lambda^{2M+1}}(g_{1}\big(\lambda;w)+g_{2}(\lambda;w)\big)\frac{\d\lambda}{\lambda^{2h-2M}}, \quad  h\notin \mathbb{Z}_{\geq 0},
 \end{cases}
\end{align*}
where $M$ in the last case is a non-negative integer such that $\textnormal{Re}(h) \in (M, M+1)$. Here, for $\lambda > 0$, 
\begin{equation*}
g_1(\lambda;w) = \exp\left[ \int_0^\lambda \sigma_{N}(t;w)\frac{\d t}{t} \right],\ \ \ \ \textnormal{resp.}\ \ \ \ g_2(\lambda;w) = \exp\left[ \int_0^\lambda \sigma_{N}(t;\overline{w}) \frac{\d t}{t} \right],
\end{equation*}
where the integration path is taken in $\mathbb{H}_+ \setminus E_{N,w}$, resp. in $\mathbb{H}_+ \setminus E_{N,\overline{w}}$, with $E_{N,w}$ a discrete subset of $\mathbb{H}_+=\{z\in\mathbb{C}:\,\textnormal{Re}(z)>0\}$, avoiding the poles of $\sigma_{N}(z;w)$.
\end{prop}

When $w \in \mathbb{R}$, the large-$N$ limit version of Proposition \ref{mainthreorem} has been studied in \cite[Theorem 1.1]{BW}.
When $w \in \mathbb{C} \setminus \mathbb{R}$, the poles of the function $z \mapsto \sigma_N(z;w)$ at finite $N$ may accumulate inside the region
$\mathrm{Re}(z)>0$ as $N$ grows. But the locations of such accumulation points
are unclear, and so it becomes difficult to choose pole-avoiding integration paths in Proposition
\ref{mainthreorem}, as $N\rightarrow\infty$. For this reason, we do not employ the
Riemann-Hilbert asymptotic analysis method used in \cite{BW}
to study the large-$N$ behaviour of $F_{N,\delta}(s,h)$. Instead, we use the
probabilistic approach introduced in \cite{AGS,AKW}, based on a family of random variables, to analyse the
large-$N$ limit of $F_{N,\delta}(s,h)$. This allows us to obtain a limiting version
in terms of a $\sigma$-Painlev\'e III$^{\prime}$ solution for positive integers
$h$ and parameters $(s,\delta)$ under the same assumptions as in Proposition
\ref{mainthreorem}; see Proposition \ref{representationforintegerh} below.\smallskip

We now discuss more details used in the probabilistic approach of \cite{AGS,AKW} for $F_{N,\delta}(s,h)$. In \cite{AKW}, Assiotis, Keating, and Warren successfully employed certain random variables arising from a determinantal point process on $\mathbb{R}\setminus\{0\}$ to study the convergence properties of $F_{N,0}(s,h)/N^{s^2+2h}$ for real parameters $s$ and $h$ with $2s > -1$ and $0 < 2h < 2s + 1$. The key reason that this approach works is, for real $s$, the quantity $F_{N,0}(s,h)$ can be related to the moments of the trace of a Hermitian matrix with respect to the finite-dimensional Hua–Pickrell measure. Furthermore, based on results in \cite{AGS}, the convergence of $F_{N,\delta}(s,h)$ (under suitable scaling) holds in the range $s, h \in \mathbb{R}$, $\delta \in \mathbb{C}$ with $\textnormal{Re}(2\delta)>-1$, provided that $s + \textnormal{Re}(\delta) > 0$ and $0 < 2h < 2s + \textnormal{Re}(2\delta) + 1$. In this paper, we extend the convergence to the full range required for the statement of Proposition \ref{mainthreorem}.

\begin{prop}\label{convergencetheorem}
Let $\delta \in \mathbb{C}$ with $\textnormal{Re}(2\delta) > -1$. Let $s \in \mathbb{R}$ with $2s + \textnormal{Re}(2\delta) > -1$. Let $h \in \mathbb{C}$ with $0<\textnormal{Re}(2h)<2s+\textnormal{Re}(2\delta)+1$. Then \footnote{
There is a typo in \cite[Theorem 1.1]{AGS}, namely, a missing term $\frac{2s}{\beta}(\delta+\overline{\delta})$ in the exponent of $N$. This arises because, in the expression of $F_{\beta,\delta}(s,0)$ given in \cite{AGS}, the authors relied on \cite[Lemma 4.14]{DHR}, where there is a missing factor of $\log n$ multiplying $\frac{z}{\beta}(\delta+\overline{\delta})$. Here, we have corrected this typo for the case $\beta=2$.}\footnote{As stated in \cite[Theorem~1.1]{AGS}, for general $\beta$, the validity range of $(s,\delta)$ for the analogue of \eqref{convergenceformula} is subject to certain restrictions, while for $\beta=2$, the validity range of \eqref{convergenceformula} can be stated as in Proposition~\ref{convergencetheorem}, due to the corresponding validity range of $\mathrm{d}(s,\delta)$.}

\bea\label{convergenceformula}
 \lim_{N\to \infty} \frac{1}{N^{s^2+s(\delta+\overline{\delta})+2h }}F_{N,\delta}(s,h)  = \mathrm{d}(s,\delta)2^{-2h } \mathbb{E}\left[\left|\mathsf{X}_{1}(s+\overline{\delta})\right|^{2h}\right],
\eea
where 
\bea\label{defofds}
\mathrm{d}(s,\delta)=\frac{G(1+\delta+\overline{\delta})G(1+\delta+s)G(1+\overline{\delta}+s)}{G(1+\delta+\overline{\delta}+2s)G(1+\delta)G(1+\overline{\delta})},
\eea
and $G(z)$ is the Barnes G-function, cf. \cite[$\S 5.17$]{NIST}.
\end{prop}
\subsection{Background for the random variable $\mathsf{X}_{1}(w)$}\label{background}

In the following, we discuss some background about the real-valued random variable $\mathsf{X}_{1}(w)$ that appears in \eqref{convergenceformula}. Let $\mathbf{H}$ be the space of infinite Hermitian matrices, i.e.,
\[
\mathbf{H}=\big\{ X=[X_{ij}]_{i,j=1}^\infty:\ X_{ij}\in \mathbb{C},\, \overline{X_{ij}} = X_{ji} \big\}.
\]
Let $\mathbf{U}(\infty)$ be the group of infinite unitary matrices $U=[U_{ij}]_{i,j=1}^\infty$ with finitely many entries $U_{ij} \neq \delta_{ij}$. The same can be realised as the inductive limit of the unitary groups $\mathbb{U}(N)$ with respect to the natural inclusions $\mathbb{U}(N) \hookrightarrow \mathbb{U}(N+1)$. Now $\mathbf{U}(\infty)$ acts on $\mathbf{H}$ via conjugation and $\mathbf{U}(\infty)$-invariant measures on the space $\mathbf{H}$ have been studied in many fields, such as in random matrix theory and in harmonic analysis on infinite dimensional groups. The $\mathbf{U}(\infty)$-invariant measures relevant to our study are the \textit{Hua-Pickrell measures} on $\mathbf{H}$, denoted by 
\begin{equation*}
	\Big\{m^{(w)}:\ w\in \mathbb{C},\,\textnormal{Re}(2w)>-1\Big\},
\end{equation*}
and which are defined as the projective limit of a sequence of probability measures $m^{(w)}_N$ on $\mathbf{H}(N)$, the space of complex Hermitian $N\times N$ matrices with $N\in\mathbb{N}$. The existence of this limit is due to the consistency property of the family of measures 
\begin{equation*}
	\Big\{m^{(w)}_N:\ w\in\mathbb{C},\ \textnormal{Re}(2w)>-1\Big\}_{N\in\mathbb{N}}
\end{equation*}
with each $m^{(w)}_N$ of the form
\begin{align}\label{defofmNw}
\text{constant} \cdot \det\big((I+\im X)^{-w-N}\big) \det\big((I-\im X)^{-\overline{w}-N}\big)\times (\text{the Lebesgue measure on~} \mathbf{H}(N)).
\end{align}
We assume that $z^{\alpha}$ takes its principal value and the pair $(\mathbf{H}(N), m^{(w)}_N)$ is called the $N$-th pseudo-Jacobi ensemble, cf. \cite{Borodin_2001}. The joint eigenvalue density function in this ensemble is given by the following probability measure on $\mathbb{R}^N/\sigma(N)$, cf. \cite[formula (1.2)]{Borodin_2001}, where $\sigma(N)$ is the $N$-th symmetric group,
\bea\label{defofmathfrakmN}
\mathfrak{m}_N^{(w)} (\d x_1,\ldots,\d x_N) = \frac{1}{\mathrm{C}_{N,w}}\prod_{1\leq j\leq k\leq N}|x_{k}-x_{j}|^2
\prod_{j=1}^N (1+\im x_j)^{-w-N} (1-\im x_j)^{-\overline{w}-N}\mathrm{d}x_{j}.
\eea
Here
\bea\label{definitionoftheconstant0407}
\mathrm{C}_{N,w}=\frac{(2\pi)^{N}}{2^{N(N+w+\overline{w})}}\prod_{j=0}^{N-1}\frac{\Gamma(1+j)\Gamma(N+w+\overline{w}-j)}{\Gamma(N+w-j)\Gamma(N+\overline{w}-j)}.
\eea
as follows from the Selberg integral formula; see, for example, \cite[Eq.~(1.19)]{ForresterWarnaar_2008} and \cite{Mehta_2004}. In order to describe $\mathsf{X}_1(w)$, we let $X_N$ be a random matrix in $\mathbf{H}(N)$ with 
\begin{equation*}
	\textnormal{Law}(X_N) = m^{(w)}_N.
\end{equation*}
As proven by Borodin and Olshanski \cite[Theorem 5.1]{Borodin_2001}, $\textnormal{Tr}(X_N)/N$ converges, as $N\rightarrow\infty$, almost surely to a random variable which is precisely $\mathsf{X}_{1}(w)$, for any $w$ with $\textnormal{Re}(2w)>-1$. Throughout, for any fixed $w\in \mathbb{C}$ with $\textnormal{Re}(2w)>-1$, we denote expectations with respect to the measures $\mathfrak{m}_N^{(w)}$ and $m^{(w)}$ by $\E_N^{(w)}[\cdot]$ and $\E[\cdot]$, respectively. In particular,
\begin{equation*}
	\E\left[\mathrm{e}^{\frac{\im t}{2N}\mathrm {Tr}(X_{N})}\right]=\E_{N}^{(w)}\left[\mathrm{e}^{\frac{\im t}{2N}\mathrm {Tr}(X_{N})}\right],\ \ \ t\in\mathbb{R},
\end{equation*}
and by the above discussion we also have that
\bea\label{the convergence for the expectation}
\lim_{N\rightarrow \infty}\E_{N}^{(w)}\left[\mathrm{e}^{\frac{\im t}{2N}\mathrm {Tr}(X_{N})}\right]=\E\Big[\mathrm{e}^{\frac{\im t}{2}\mathsf{X}_{1}(w)}\Big]
\eea
uniformly on compact subsets of $\mathbb{R}$.\smallskip

At this point we briefly review the connection between $\mathsf{X}_{1}(w)$ and the ergodic decomposition of the Hua-Pickrell measures $m^{(w)}$ on $\mathbf{H}$. Any ergodic $\mathbf{U}(\infty)$-invariant probability measure on $\mathbf{H}$ can be parameterised by the points $\omega\in \Omega$ \cite[Proposition 4.1]{Borodin_2001}, where the space $\Omega$ is given by
\bea\label{definitionofOmega}
\Omega &=& \Big\{
\omega=(\alpha^+, \alpha^-, \gamma_{1}, \delta)\in \mathbb{R}_{\geq 0}^\infty \times \mathbb{R}_{\geq 0}^\infty \times \mathbb{R} \times \mathbb{R}_{\geq 0}:
\alpha^+ = (\alpha^+_1 \geq\alpha^+_2 \geq \cdots \geq 0), \nonumber\\
&& 
\,\, \,\, \ \ \alpha^- = (\alpha^-_1 \geq \alpha^-_2 \geq \cdots \geq 0),\ \ \ 
\sum (\alpha^+_i)^2 + \sum (\alpha^-_i)^2 \leq \delta,\ \gamma_{1} \in \mathbb{R}
\Big\}.
\eea
It is known by \cite[Proposition 4.4]{Borodin_2001} that there exists a unique probability measure $\mathcal{M}^{(w)}$ on $\Omega$ such that
\[
m^{(w)}=\int_\Omega \mathcal{E}_{\omega} \mathcal{M}^{(w)} (\mathrm{d}\omega),
\]
where
each ergodic measure $\mathcal{E}_{\omega}$ is determined by its Fourier transform. 
Moreover, for the parameters $\alpha^+$ and $\alpha^-$, Borodin and Olshanski showed in \cite[Theorem 6.1]{Borodin_2001} that, under $\mathcal{M}^{(w)}$,
\begin{align}\label{defOfmathcalC}
\{\alpha_{i}^+\} \sqcup \{-\alpha_{i}^-\} \stackrel{\mathsf{d}}{=} \mathcal{C}^{(w)},
\end{align} where $\mathcal{C}^{(w)}$ is a determinantal point process. For the parameter $\gamma_2:=\delta-\sum(\alpha_{i}^{+})^2-\sum(\alpha_{i}^{-})^2$, under $\mathcal{M}^{(w)}$, it was proven in \cite[Section 7]{Borodin_2001} that it is almost surely zero for $w=0$ and by Qiu \cite[Theorem 1.1]{Q} that it is almost surely zero for general $w$ with $\textnormal{Re}(2w)>-1$. For the parameter $\gamma_1$, with respect to the measure $\mathcal{M}^{(w)}$,
\begin{equation*}
	\gamma_1 \stackrel{\mathsf{d}}{=} \mathsf{X}_{1}(w).
\end{equation*}
Specifically for $w\in \mathbb{R}$ with $2w>-1$, Qiu \cite[Theorem 1.2]{Q} gave an identification for $\mathsf{X}_{1}(w)$, which is almost surely
\bea\label{describtionofX1w}
\lim_{m\rightarrow \infty}
\left[
\sum_{i=1}^\infty \alpha_i^+ \mathbbm{1}_{\{i\in \mathbb{N}: \alpha_i^+ > m^{-2}\}}
-
\sum_{i=1}^\infty \alpha_i^- \mathbbm{1}_{\{i\in \mathbb{N}:\alpha_i^- > m^{-2}\}}
\right].
\eea
To the best of our knowledge, it is unknown whether (\ref{describtionofX1w}) holds for $w\in \mathbb{C}\setminus\mathbb{R}$.

\subsection{Painlev\'e equation for the characteristic function of $\mathsf{X}_1(w)$}

We now state our third main result, namely that the characteristic function 
\begin{equation}\label{chara}
	t\mapsto \E\left[\mathrm{e}^{-\frac{\im t}{2} \mathsf{X}_{1}(w)}\right]
\end{equation}
is a tau-function associated with the $\sigma$-Painlev\'e III$'$ equation, for any $w\in \mathbb{C}$ with $\textnormal{Re}(2w)>-1$, see Theorem \ref{maintheoreminnote2}. This gives an affirmative answer to the following question proposed in \cite[Remark 1.7]{AGS}: 

\begin{quest}\label{que0}
Is the characteristic function \eqref{chara} connected to the 
$\sigma$-Painlev\'e~\textnormal{III$'$} equation for admissible $w\in\mathbb{C}$?
\end{quest}
Here, in our notation, the random variable $\mathsf{X}_{1}(w)$ is identical to the random variable $\mathsf{X}_{\beta}(w)$ with $\beta=2$ in \cite{AGS}. Note that, for real $w$, Question \ref{que0} was answered previously in \cite{ABGS}. In detail,
\begin{theo}\label{citetheorem}\cite[Theorem 1.2]{ABGS}
Let $w\in \mathbb{R}$ with $w>-\frac{1}{2}$ and let $t\in \mathbb{R}\setminus\{0\}$.
Define
\[
\tau_w(t) := t \frac{\d}{\d t} \ln \E\left[\mathrm{e}^{-\frac{\im t}{2}\mathsf{X}_{1}(w)}\right].
\]
Then $\tau_w=\tau_w(t)$ satisfies the $\sigma$-Painlev\'e {\rm III$'$} equation 
\bea\label{theequationsatisfiedbytauw}
\left( t \frac{\d^2 \tau_w}{\d t^2} \right)^2
= -4 t \left(  \frac{\d \tau_w}{\d t} \right)^3
+4\big(w^2+\tau^{(w)}\big) \left(  \frac{\d \tau_w}{\d t} \right)^2
+t  \frac{\d \tau_w}{\d t} -\tau_w
\eea
with boundary conditions
\[
\begin{cases} \vspace{.2cm}
\tau_w(0) = 0, & \textnormal{for } w>0, \\
\frac{\d}{\d t} \tau_w(t) \Big|_{t=0} = 0, & \textnormal{for } w>\frac{1}{2}.
\end{cases}
\]
\end{theo}
The following result extends this connection between \eqref{chara} and the
$\sigma$-Painlev\'e equation to general $w\in\mathbb{C}$ with
$\operatorname{Re}(2w)>-1$.

\begin{theo}\label{maintheoreminnote2}
Let $w\in \mathbb{C}$ with $\textnormal{Re}(2w)>-1$. Then there is a discrete set
$\mathcal{Z}_{w}\subset \mathbb{R}_{+}$ consisting of the zeros of
\begin{equation*}
	\mathbb{R}_+\ni t\mapsto
	\mathbb{E}\Big[\e^{-\frac{\im t}{2}\mathsf{X}_1(w)}\Big],
\end{equation*}
such that the function
\begin{equation}\label{defofthelimitsolution0701}
\tau_w(t)
:=
t \frac{\d}{\d t}
\ln \E\Big[\mathrm{e}^{-\frac{\im t}{2}\mathsf{X}_{1}(w)}\Big],
\end{equation}
defined for
\[
t\in \mathbb{R}\setminus
\left(
\{0\}\cup\mathcal{Z}_{w}\cup(-\mathcal{Z}_{\overline{w}})
\right),
\]
satisfies the following $\sigma$-Painlev\'e {\rm III$'$} equation
\begin{equation}\label{limitequation}
\left(t\frac{\d^{2}\tau_w}{\d t^{2}}\right)^{2}
=
-4 t \left(\frac{\d\tau_w}{\d t}\right)^{3}
+4\bigg((\textnormal{Re}(w))^2+\tau_w\bigg)
\left(\frac{\d\tau_w}{\d t}\right)^{2}
+\Big(t+2\im\,\textnormal{Im}(w^2)\Big)
\frac{\d\tau_w}{\d t}
-\tau_w-(\textnormal{Im}(w))^2,
\end{equation}
with boundary conditions
\bea\label{boundaryconditionsforthelinitcaseatt=0}
\begin{cases}\vspace{.2cm}
\tau_w(0) = 0, & \textnormal{for } \textnormal{Re}(w)>0, \\
\frac{\d}{\d t} \tau_w(t) \Big|_{t=0} =-\frac{\im}{2}\frac{\textnormal{Im}(w)}{\textnormal{Re}(w)}, & \textnormal{for }\textnormal{Re}(w)>\frac{1}{2}.
\end{cases}
\eea
\end{theo}
We remark that the restriction
$
t\notin \{0\}\cup\mathcal{Z}_w\cup(-\mathcal{Z}_{\overline{w}})
$
is natural. Indeed, for $w\in\mathbb{C}\setminus \mathbb{R}$, the characteristic function of
$\mathsf{X}_1(w)$ may vanish for some values of $t$. This is in contrast
with the real case considered in Theorem \ref{citetheorem}, where it was
proved in \cite{ABGS} that the characteristic function is strictly
positive.\smallskip

As an application of Theorem \ref{maintheoreminnote2}, we now compute the second moment of $\mathsf{X}_{1}(w)$.

\begin{cor}\label{almostsurely}
Let $w\in \mathbb{C}$ with $\textnormal{Re}(2w)>1$. Then
\bea\label{thesecondmomentofXw}
\mathbb{E}\left[(\mathsf{X}_{1}(w))^2\right]=\frac{1+4\,\textnormal{Im}^2(w)}{(2\,\textnormal{Re}(w)+1)(2\,\textnormal{Re}(w)-1)}>0.
\eea
In particular, $\mathsf{X}_{1}(w)$ is not almost surely zero.
\end{cor}

The above result \eqref{thesecondmomentofXw} provides a partial answer to the question posed in \cite[Remark 1.6]{AGS}, which asks whether $\mathsf{X}_{1}(w)$ is not almost surely zero in the full range $w \in \mathbb{C}$ with $\textnormal{Re}(2w) > -1$. Corollary \ref{almostsurely} confirms this for $\textnormal{Re}(2w) > 1$. Prior to this, it was known that the statement holds for real values $w$ with $2w > -1$, as shown in \cite{AKW}.\smallskip

As a consequence of Proposition \ref{convergencetheorem} and Corollary \ref{almostsurely}, we obtain the following result concerning the leading coefficients of $F_{N,\delta}(s,h)$ in \eqref{object} for positive integers $h$.

\begin{prop}\label{representationforintegerh}
Let $\delta \in \mathbb{C}$ with $\textnormal{Re}(2\delta) > -1$. Let $s \in \mathbb{R}$ and  with $2s + \textnormal{Re}(2\delta) > -1$. Let $\mathrm{d}(s,\delta)$ be as in
(\ref{defofds}). Then for any $h\in\mathbb{N}$ with $2h<2s+\textnormal{Re}(2\delta)+1$,
\bea\label{expressionfortheabsolutemomentsofXw}
 \lim_{N\to \infty} \frac{F_{N,\delta}(s,h)}
 {N^{s^2+s(\delta+\overline{\delta})+2h }}
 =
 (-1)^{h}\mathrm{d}(s,\delta)
 \frac{\mathrm{d}^{2h}}{\mathrm{d}t^{2h}}
 \exp\bigg[\int_{0}^{t}
 \tau_w(u)\frac{\mathrm{d}u}{u}\bigg]\Bigg|_{t=0},\ \ \ w=s+\overline{\delta},
\eea
with $\tau_w(t)$ exactly as in Theorem
\ref{maintheoreminnote2}. The integration path in \eqref{expressionfortheabsolutemomentsofXw} lies in a neighbourhood of $0$ that is disjoint from $\mathcal{Z}_{w} \cup (-\mathcal{Z}_{\overline{w}})$, for $w=s+\overline{\delta}$. In particular, when $h=1$,
\bea\label{convergenceformula2}
 \lim_{N\to \infty} \frac{F_{N,\delta}(s,1)}
 {N^{s^2+s(\delta+\overline{\delta})+2}}
 =
 \frac{\mathrm{d}(s,\delta)}{4}
 \frac{1+4\,\textnormal{Im}^2(\delta)}
 {4(s+\textnormal{Re}(\delta))^2-1}
\eea
holds for $s\in \mathbb{R}$ with $2s+\textnormal{Re}(2\delta)>1$.
\end{prop}

\subsection{A connection between joint moments in the Circular Jacobi ensemble and $\mathsf{X}_n(w)$}

For $w\in \mathbb{R}$, a sequence of random variables $(\mathsf{X}_{n}(w))_{n\in\mathbb{N}}$ has been used in \cite{AGKW} to connect the joint moments of CUE characteristic polynomials of any order of derivatives. We now show that, for general $w \in \mathbb{C}$, this sequence can similarly be employed to connect the joint moments of an arbitrary number of derivatives of characteristic polynomials from the Circular Jacobi Ensemble with arbitrary positive real exponents. In this study, the power sums of the random points of $\mathcal{C}^{(w)}$ in \eqref{defOfmathcalC} play an important role. We begin with the following definition.

\begin{defn}\label{DefinitionRandomVariables0406}
Let $w\in \mathbb{C}$ with $\textnormal{Re}(2w)>-1$.  Associated to the determinantal point process $\mathcal{C}^{(w)}$ in \eqref{defOfmathcalC}, define random variables $(\mathsf{X}_n(w))_{n=2}^{\infty}$, where the random point configuration $\mathsf{P}$ follows the distribution of $\mathcal{C}^{(w)}$,
\begin{align}\label{defofxn}
    \mathsf{X}_n(w)&:= \sum_{x\in \mathsf{P}} x^n, \ \ \textnormal{for } n\in\mathbb{Z}_{\geq 2}.
\end{align}
Moreover, define random variables $(\mathsf{R}_n(w))_{n=0}^{\infty}$ by $\mathsf{R}_0(w)\equiv 1$, $\mathsf{R}_1(w)=-\mathsf{X}_1(w)$ and inductively, for $n\in\mathbb{Z}_{\geq 2}$,
\begin{equation}\label{PowerSumsToElemSymFunctions0408}
    \mathsf{R}_n(w)=-\sum_{j=1}^{n}\frac{(n-1)!}{(n-j)!}\mathsf{R}_{n-j}(w)\mathsf{X}_j(w).
\end{equation}
\end{defn}
We remark that \eqref{defofxn} is well defined due to \cite[$(13)$]{AN}.
\begin{prop}\label{representation0407}
Let $\delta\in \mathbb{C}$ with $\textnormal{Re}(2\delta)>-1$ and $k\in\mathbb{N}$. Let $n_1>\ldots> n_k \in \mathbb{Z}_{\geq 0}$, and $\{h_j\}_{j=1}^k\subset\mathbb{R}_+$ and put $s=\sum_{j=1}^{k}h_{j}$. Then we have
\begin{equation}
       \lim_{N\to \infty} \frac{ \int_{\mathbb{U}(N)}\prod_{j=1}^{k}\big|Z_{A}^{(n_{j})}(0)\big|^{2h_{j}}\mathrm{d}\mu_{N}^{(\delta)}(A)}{N^{s^2+s(\delta+\overline{\delta})+2\sum_{j=1}^k n_j h_j} }  = \mathrm{d}(s,\delta)2^{-2\sum_{j=1}^k n_j h_j  } \mathbb{E}\left[\prod_{j=1}^k\left|\mathsf{R}_{n_j}(s+\bar{\delta})\right|^{2h_j}\right],
        \nonumber
\end{equation}
    \begin{align*}
\lim_{N\to \infty} \frac{\int_{\mathbb{U}(N)}\prod_{j=1}^{k}\big|V_{A}^{(n_{j})}(0)\big|^{2h_{j}}\mathrm{d}\mu_{N}^{(\delta)}(A)}{N^{s^2+s(\delta+\overline{\delta})+2\sum_{j=1}^k n_j h_j} }  =&\ \mathrm{d}(s,\delta)2^{-2\sum_{j=1}^k n_j h_j  }\\
       &\ \ \ \ \ \times\mathbb{E}\left[\prod_{j=1}^k \Bigg|\sum_{m_{j}=0}^{n_{j}}(-\textnormal{i})^{m_j} \binom{n_{j}}{m_{j}} \mathsf{R}_{n_j-m_j}(s+\bar{\delta})\Bigg|^{2h_j}\right]
       \nonumber,
\end{align*}
where $\mathrm{d}(s,\delta)$ is given in \eqref{defofds}, and $Z_A(\theta),V_A(\theta)$ in \eqref{polydef}.
\end{prop}
As an application of Proposition \ref{representation0407}, we compute some moments of the random variables $\mathsf{R}_{n}(w)$ using the reproducing kernel associated with the probability measure defined in (\ref{definitionofJacobimeasure}). As a consequence of the result, we conclude that the random variables $\mathsf{R}_{n}(w)$ are almost surely non-zero, a fact that seems hard to prove directly via the point process representations of $\mathsf{R}_{n}(w)$. 

\begin{cor}\label{non-zeroproperty}
Let $w\in \mathbb{C}$ with $\textnormal{Re}(2w)>1$. Let $n\in \mathbb{N}$. Then we have
\begin{align}\label{2ndmomentexplicit}
 \E\left[\left|\mathsf{R}_{n}(w)\right|^2\right] =  \Gamma(w+\overline{w}) \Gamma(w+\overline{w}-1)
&\sum_{i,j=0}^{n}\binom{n}{i}\binom{n}{j}(-2)^{i+j}\sum_{\ell_1=0}^i \sum_{\ell_2=0}^j (-1)^{\ell_1+\ell_2}\nonumber\\
\times \, &\frac{(w-1)^{(\ell_{1})}(\overline{w}-1)^{(\ell_{2})}}{\Gamma(w+\overline{w}-1+\ell_1) \Gamma(w+\overline{w}-1+\ell_2)}
\frac{1}{i+j+w+\overline{w}-1}, \nonumber \\
\mathbb{E}\left[\Bigg|\sum_{m=0}^{n}(-\textnormal{i})^{m} \binom{n}{m} \mathsf{R}_{n-m}(w)\Bigg|^2\right]=2^{2n} &\Gamma(w+\overline{w}) \Gamma(w+\overline{w}-1)\sum_{i,j=0}^{n}\binom{n}{i}\binom{n}{j}(-1)^{i+j}\nonumber\\
\times \, &\frac{(w-1)^{(i)}(\overline{w}-1)^{(j)}}{\Gamma(w+\overline{w}-1+i) \Gamma(w+\overline{w}-1+j)}
\frac{1}{2n+w+\overline{w}-1},
\nonumber
\end{align}
where $(a)^{(n)}:=\prod_{\ell=0}^{n-1}(a+\ell)$ is the rising factorial with $(a)^{(0)}=1$. 
Moreover, we have that
\begin{align}
\mathbb{P}\left(\left|\sum_{m=0}^{n}(-\textnormal{i})^{m} \binom{n}{m} \mathsf{R}_{n-m}(w)\right|>0\right)>0 \;\; & \textnormal{ and } \;\; \mathbb{P}\left(\left|\mathsf{R}_n(w)\right|>0\right)>0.\nonumber
\end{align}
\end{cor}

\begin{rem}
We remark that when $w\in \mathbb{R}$, Corollary \ref{non-zeroproperty} coincides with \cite[Corollary 1.6]{AGKW}. In fact, by the definition of the Gaussian hypergeometric function,
\begin{align*}
&\sum_{\ell=0}^{n}
(-1)^{\ell}
\frac{(w-1)^{(\ell)}}{\Gamma(2w-1+\ell)}
\\
&=
\frac{1}{\Gamma(2w-1)}
\sum_{\ell=0}^{n}
(-1)^{\ell}
\frac{(w-1)^{(\ell)}}{(2w-1)^{(\ell)}}
=
\frac{1}{\Gamma(2w-1)}
\,{}_2F_1(-n,w-1;2w-1;1).
\end{align*}
Using then the explicit evaluation of the Gaussian hypergeometric function at $1$, i.e.
\beas
{}_2F_1(-n,a;c;1)=\frac{(c-a)^{(n)}}{(c)^{(n)}},
\eeas
we indeed recover the result in \cite[Corollary 1.6]{AGKW}.
\end{rem}

\subsection{Joint moments of $\mathsf{X}_n(w)$ in terms of Painlev\'e transcendents}
The following is one of our main results. It represents the joint moments of
$\{\mathsf{X}_{j}(w)\}_{j=1}^k$ in terms of a solution to the
$\sigma$-Painlev\'e {\rm III}$'$ equation, extending
\cite[Theorem 1.11]{AGKW} from real parameters to complex parameters.

\begin{theo}\label{painlevethm0328}
Let $k\in\mathbb{Z}_{\geq 2}$, $n_{2},\ldots,n_{k}\geq 0$ be non-negative integers and $\sum_{j=2}^{k}n_j >0$. Suppose $w \in \mathbb{C}$ has $\textnormal{Re}(2w)>{\sum_{\ell=2}^{k}\ell n_{\ell}-1}$. Then, for $t_{1}\geq 0$,

\begin{equation}\label{formula10328}
 \mathbb{E}\bigg[\mathrm{e}^{-\textnormal{i} t_{1}\mathsf{X}_{1}(w)}\prod_{j=2}^{k}\mathsf{X}_{j}(w)^{n_{j}}\bigg]=\frac{(-2\textnormal{i})^{\sum_{\ell=2}^{k}\ell n_{\ell}} }{t_{1}^{\sum_{\ell=2}^{k}\ell n_{\ell}-1}}
\sum_{m=0}^{\sum_{\ell=2}^{k}(\ell-1) n_{\ell}} t_1^{m-1} \mathcal{A}_{m}^{(w,\overline{w})}(t_{1})\frac{\mathrm{d}^{m}}{\mathrm{d}t_{1}^{m}}\mathbb{E}\Big[\mathrm{e}^{-\textnormal{i} t_{1}\mathsf{X}_{1}(w)}\Big],   
\end{equation}
where $t_1^{m-1}\mathcal{A}_{m}^{(w,\overline{w})}(t_{1})$ are polynomials in $t_{1}$ of degree at most $\sum_{\ell=2}^{k}(\ell-1)n_{\ell}-1$ and whose coefficients are polynomials themselves in $w$ and $\overline{w}$, with degree in $w$ or $\overline{w}$ at most $\sum_{\ell=2}^{k}(\ell-1)n_{\ell}$. 
\end{theo}

When $w\in \mathbb{R}$, the coefficient $\mathcal{A}_{m}^{(w,\overline{w})}$ in Theorem \ref{painlevethm0328} coincides with $\mathcal{A}_{m}^{(w)}$ as given in \cite[Theorem 1.11]{AGKW}. Even though we state Theorem \ref{painlevethm0328} for $t_1\geq 0$, using
\begin{equation*}
\mathbb{E}\Big[\mathrm{e}^{-\im t_{1}\mathsf{X}_{1}(w)}\prod_{j=2}^{k}\mathsf{X}_{j}(w)^{n_{j}}\Big]
=(-1)^{\sum_{q=2}^{k}qn_{q}}\mathbb{E}\Big[\mathrm{e}^{-\im (-t_{1})\mathsf{X}_{1}(\overline{w})}\prod_{j=2}^{k}\mathsf{X}_{j}(\overline{w})^{n_{j}}\Big],
\end{equation*}
one immediately has an analogous result for $t_1<0$. We omit its obvious statement.\smallskip

From Theorem \ref{maintheoreminnote2}, we know that the function 
\begin{equation*}
	t_1\mapsto t_{1}\frac{\mathrm{d}}{\mathrm{d}t_{1}}\log\mathbb{E}\Big[\mathrm{e}^{-\im t_{1}\mathsf{X}_{1}(w)}\Big]
\end{equation*}
solves the $\sigma$-Painlev\'e III$'$ equation \eqref{limitequation}. So, Theorem \ref{painlevethm0328} actually gives a representation of the 
the left hand-side of \eqref{formula10328} in terms of a solution of the $\sigma$-Painlev\'{e}-III$'$ equation. Moreover, the polynomial coefficients appearing in the right hand-side of  (\ref{formula10328}) can be computed recursively. Hence,  the closed formula (\ref{formula10328}) provides an efficient and a recursive way to compute all the integer joint moments of $\mathsf{X}_{j}(w)$. A few examples of these formulae are listed in the following corollary. In particular, when $w\in \mathbb{R}$, these results agree with \cite[Corollary 1.12]{AGKW}.  
\begin{cor}\label{painlevecor0324}
Let $t_1\geq 0$. For all $w\in \mathbb{C}$ with $\textnormal{Re}(2w)>1$, we have
\begin{equation*}
    \E\Big[
\mathrm{e}^{-\textnormal{i} t_{1}\mathsf{X}_{1}(w)}\mathsf{X}_{2}(w)
\Big]
= - \frac{w+\overline{w}}{t_{1}} \frac{\mathrm{d}}{\mathrm{d}t_{1}} 
\E \Big[
\mathrm{e}^{-\textnormal{i} t_{1}\mathsf{X}_{1}(w)}
\Big]+\frac{\overline{w}-w}{t_{1}}\E \Big[
\mathrm{e}^{-\textnormal{i} t_{1}\mathsf{X}_{1}(w)}
\Big],
\end{equation*}
and for all $w\in \mathbb{C}$ with $\textnormal{Re}(2w)>3$,
\beas
\E\Big[
\mathrm{e}^{-\textnormal{i} t_{1}\mathsf{X}_{1}(w)}\mathsf{X}_{2}(w)^2
\Big]&=&
\frac{(w+\overline{w})^2+2}{t_{1}^2} \frac{\mathrm{d}^2}{\mathrm{d}t_{1}^2} 
\E \Big[
\mathrm{e}^{-\textnormal{i} t_{1}\mathsf{X}_{1}(w)} 
\Big]\\
&&-\frac{3(w+\overline{w})^2+2t_{1}\left((\overline{w})^2-w^2\right)}{t_{1}^3}\frac{\mathrm{d}}{\mathrm{d}t_{1}} 
\E\Big[
\mathrm{e}^{-\textnormal{i} t_{1}\mathsf{X}_{1}(w)}  
\Big]\\
&&+\frac{(\overline{w}-w)^2t-2t+3\overline{w}^2-3w^2}{t_{1}^3}\E\Big[
\mathrm{e}^{-\textnormal{i} t_{1}\mathsf{X}_{1}(w)}
\Big]. 
\eeas
\end{cor}
The above formulae may appear singular  when $t_1=0$, but this is not the case:  the moments of $\mathsf{X}_1(w)$ and $\mathsf{X}_2(w)$ satisfy certain relations that ensure the cancellation of singular $t_1$ powers in the denominator at $t_1=0$. In particular, evaluating at $t_1 = 0$ and using Corollary \ref{almostsurely}, allows one to obtain, for instance, the following formula:
\begin{align*}
 \mathbb{E}\big[\mathsf{X}_1(w)\big]
=&\,\,\frac{\textnormal{Im}(w)}{\textnormal{Re}(w)}\\
    \mathbb{E}\big[\mathsf{X}_2(w)\big]=&\,\,2\,\textnormal{Re}(w)\, \mathbb{E}\Big[\left(\mathsf{X}_1(w)\right)^2\Big]-2\,\textnormal{Im}(w)\,\mathbb{E}\big[\mathsf{X}_1(w)\big]\stackrel{\eqref{thesecondmomentofXw}}{=}\frac{2|w|^2}{(2\,\textnormal{Re}(w)+1)(2\,\textnormal{Re}(w)-1)\textnormal{Re}(w)}.
\end{align*}

\subsection{Outline of the paper}

In Section \ref{ahanklerepresentation0402}, we represent
$F_{N,\delta}(s,h)$ in terms of a Hankel determinant. In Section
\ref{RHmethod}, using the Riemann--Hilbert method, we show that this
Hankel determinant is a tau-function associated with the Painlev\'e V
equation. In Section \ref{representation}, we derive the local
asymptotic behaviour of this tau-function near the origin and obtain a
representation of $F_{N,\delta}(s,h)$ in terms of the corresponding
$\sigma$-Painlev\'e V transcendent. The proof of Proposition
\ref{mainthreorem} is also given in this section. In Section \ref{proofofproposition14}, we analyse the large-$N$
scaling limit of the $\sigma$-Painlev\'e V equation and establish the
associated $\sigma$-Painlev\'e III$'$ description. We also obtain a
probabilistic representation of the limiting quantity in terms of the
characteristic function of a random variable arising from the ergodic
decomposition of Hua--Pickrell measures. The proof of Proposition
\ref{convergencetheorem} is given in this section. Passing to the large-$N$ limit under the appropriate scaling of the $\sigma$-Painlev\'e V equation leads to the limiting $\sigma$-Painlev\'e III$'$ equation stated in
Theorem \ref{maintheoreminnote2}, whose proof is given in Section
\ref{proofoftheorem16}. Corollary \ref{almostsurely} and Proposition
\ref{representationforintegerh} are also proved there. In Section \ref{proofsofProposition111and112}, we relate the joint moments of higher-order derivatives of characteristic polynomials from the Circular Jacobi Ensemble to the joint statistics of a sequence of random variables defined via power sums of random points from a determinantal point process. Proposition
\ref{representation0407} and Corollary \ref{non-zeroproperty} are also
proved in this section. Finally, in Section \ref{sectiononjointmofasequenceofrandomvariables},
we establish further connections between these random matrix statistics
and integrable systems, proving Theorem \ref{painlevethm0328} and
Corollary \ref{painlevecor0324}.

\section{$F_{N,\delta}(s,h)$ in terms of a Hankel determinant}\label{ahanklerepresentation0402}
Recall that there is a natural map from $\mathbf{H}(N)$ to
$\mathbb{R}^N/\sigma(N)$, sending a matrix
$X_N\in\mathbf{H}(N)$ to the equivalence class of its eigenvalues. Under the
natural identification of $\mathbb{R}^N/\sigma(N)$ with the closed Weyl chamber
\[
\mathbb{W}^N
:=
\left\{
(x_1,\ldots,x_N)\in\mathbb{R}^N:
x_1\geq\ldots\geq x_N\right\},
\]
this map sends $X_N$ to its ordered eigenvalues
\[
\lambda_1^{(N)}(X_N)
\geq\ldots\geq
\lambda_N^{(N)}(X_N).
\]
The pushforward of $m_N^{(w)}$ under this eigenvalue map is precisely
$\mathfrak{m}_N^{(w)}$. Equivalently, if
$X\in\mathbf{H}$ is distributed according to $m^{(w)}$, then the ordered
eigenvalues of its top-left $N\times N$ corner $X_N$ are distributed
according to $\mathfrak{m}_N^{(w)}$. Therefore, if
\[
\bigl(x_1^{(N)},\ldots,x_N^{(N)}\bigr)\in\mathbb{W}^N
\]
is distributed according to $\mathfrak{m}_N^{(w)}$, then
\[
\bigl(x_1^{(N)},\ldots,x_N^{(N)}\bigr)
\overset{\mathsf{d}}{=}
\bigl(
\lambda_1^{(N)}(X_N),\ldots,
\lambda_N^{(N)}(X_N)
\bigr),
\]
where $X_N$ is distributed according to $m_N^{(w)}$.
Using the joint eigenvalue density \eqref{definitionofJacobimeasure} of the
circular Jacobi ensemble together with the change of variables
$
\cot\left(\theta_j/2\right)=x_j,
$
see, for example, \cites{keating2000random,W}, the quantity
$F_{N,\delta}(s,h)$ can be represented in terms of a moment of the eigenvalues distributed according to
$\mathfrak{m}_N^{(s+\overline{\delta})}$. More precisely, we have the
following result.
\begin{prop}\label{arepresentationofasanexpectation}
Let $s \in \mathbb{R}$ and $\delta \in \mathbb{C}$ with $2s + \textnormal{Re}(2\delta) > -1$, and let $h \in \mathbb{C}$ satisfy $-1< \textnormal{Re}(2h) < 2s + \textnormal{Re}(2\delta) + 1$. Then 
\beas
F_{N,\delta}(s,h)=F_{N,\delta}(s,0)2^{-2h}\mathbb{E}_{N}^{(s+\overline{\delta})}\Bigg[\Bigg|\sum_{j=1}^{N}x_{j}^{(N)}\Bigg|^{2h}\Bigg],
\eeas
where $\mathbb{E}_{N}^{(w)}[\cdot]$ denotes the expectation with respect to $\mathfrak{m}^{(w)}_N$ in \eqref{defofmathfrakmN}.
\end{prop}
Then, by an argument similar to that in \cite[Lemma 2.4]{BW}, but using the Fourier transform
\begin{align*}
&\int_{-\infty}^{\infty}
(1+\im x)^{-N-w}
(1-\im x)^{-N-\overline{w}}
\mathrm{e}^{-\im \xi x}\,\mathrm{d}x \\
&\qquad =
\frac{2^{2-2N-w-\overline{w}}\pi}
{\Gamma(N+\overline{w})}
\mathrm{e}^{-\xi}
U\!\left(
1-N-\overline{w},
2-2N-w-\overline{w},
2\xi
\right),
\qquad \xi\geq 0,
\end{align*}
instead of the special case $w=\overline{w}$ considered in \cite{BW},
we obtain the following Hankel determinant representation for the
characteristic function of $\operatorname{Tr}(X_N)$. Here,
$U(a,b,z)$ denotes the confluent hypergeometric function of the second kind.
\begin{prop}\label{Hankelrepresentation}
For any $w\in\mathbb{C}$ with $\textnormal{Re}(2w)>-1$ and any $\xi\geq 0$,
\begin{equation}\label{i13}
\mathbb{E}_{N}^{(w)}\bigg[\mathrm{e}^{-\im\frac{\xi}{2}\sum_{j=1}^{N}x_{j}^{(N)}}\bigg]=\frac{1}{\mathrm{C}_{N,w}}(-1)^{N(N-1)/2}\bigg(\frac{2\pi\mathrm{e}^{-\frac{\xi}{2}}}{\Gamma(N+\overline{w})}\bigg)^N2^{-N^2-N(w+\overline{w})}\,\mathcal{J}_N(\mu,\nu,\xi)
\end{equation}
in terms of the Hankel determinant
\begin{equation}\label{i14}	
	\mathcal{J}_N(\mu,\nu,\xi)=\det\Big[U(1-2\nu,2-2\nu-2\mu+j+k,\xi)\Big]_{j,k=0}^{N-1},
\end{equation}
with parameters $\mu=\frac{1}{2}(w+N)$, $\nu=\frac{1}{2}(\overline{w}+N)$ and with $\mathrm{C}_{N,w}$ as in \eqref{definitionoftheconstant0407}.
\end{prop}

\begin{rem}
In Proposition \ref{Hankelrepresentation}, we require $\zeta\geq 0$. However, by symmetry, for $\xi<0$,
\beas
\mathbb{E}_{N}^{(w)}\bigg[\mathrm{e}^{-\im\frac{\xi}{2}\sum_{j=1}^{N}x_{j}^{(N)}}\bigg]=\mathbb{E}_{N}^{(\overline{w})}\bigg[\mathrm{e}^{\im\frac{\xi}{2}\sum_{j=1}^{N}x_{j}^{(N)}}\bigg].
\eeas
\end{rem}
\section{The Riemann-Hilbert method for the Hankel determinant $\mathcal{J}_N$}\label{RHmethod}
We now begin to investigate the Hankel determinant \eqref{i14} by Riemann-Hilbert techniques. To that end, recall that $U(a,b,z)$ solves Kummer's equation \cite[$\S 13.2.1$]{NIST} being uniquely determined by the asymptotics
\begin{equation}\label{largezasymptotics}
	U(a,b,z)\sim z^{-a}\sum_{m=0}^{\infty}(-1)^m\frac{(a)_m(a-b+1)_m}{m!z^m},\ \ \ \ z\rightarrow\infty,
\end{equation}
in the sector $-\frac{3\pi}{2}<\textnormal{arg}\,z<\frac{3\pi}{2}$ with the principal value of $z^{-a}$. Moreover, $U(a,b,z)$ admits the following integral representation, cf. \cite[$\S 13.4.14$]{NIST},
\begin{equation}\label{e2}
	U(a,b,z)=\mathrm{e}^{-\im \pi a}\frac{\Gamma(1-a)}{2\pi\im}\int_{\infty}^{(0+)}\mathrm{e}^{-zt}t^{a-1}(1+t)^{b-a-1}\mathrm{d} t,\ \ \ \ \ \ \ -\frac{\pi}{2}<\textnormal{arg}\,z<\frac{\pi}{2}
\end{equation}
where $a\notin\mathbb{Z}_{\geq 1}$ is assumed, the branches of $t^{a-1}$ and $(1+t)^{b-a-1}$ are defined on the complex $t$-plane cut along $[0,\infty)$ and $(-\infty,-1]$ such that $0<\textnormal{arg}\,t<2\pi$ and $-\pi<\textnormal{arg}\,(1+t)<\pi$, and the contour of integration in \eqref{e2} is shown in Figure \ref{fig1}.
\begin{figure}[tbh]
\centering
	\begin{tikzpicture}[xscale=0.9,yscale=0.9]
	\draw [thick,decorate,decoration={zigzag,segment length=4,amplitude=1,post=lineto,post length=15}] (-1.05*\xr,0) -- (0,0);
        \draw [thick,->,decorate,decoration={zigzag,segment length=4,amplitude=1,pre=lineto,pre length=0,post=lineto,post length=3}] (0,0) -- (1.05*\xr,0) node [right] {$\footnotesize{\textnormal{Re}(t)}$};
  \draw [thick,->] (0,-\yr) -- (0,\yr) node [above] {$\footnotesize{\textnormal{Im}(t)}$};
  \draw[thick, xshift=5,red,decoration={markings,mark=between positions 0.125 and 0.875 step 0.25 with \arrow{<}},postaction={decorate}] (\xr,-\yr/4) -- (-\yr/4,-\yr/4) arc (-90:-270:\yr/4) (-\yr/4,\yr/4) -- (\xr,\yr/4) node[above left] {$$}; 
\end{tikzpicture}
\caption{The oriented integration contour used in \eqref{e2} in the complex $t$-plane in \textcolor{red}{red} colour with branch cuts on $(-\infty,-1]$ for $t\mapsto (1+t)^{\beta}$ and on $[0,\infty)$ for $t\mapsto t^{\beta}$.}
\label{fig1}
\end{figure}
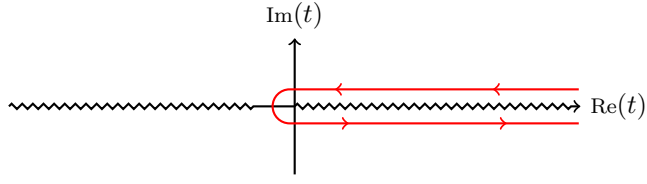
Utilising \eqref{e2} we find, after changing variables,
\begin{equation*}
	U(1-2\nu,2-2\nu-2\mu+j+k,z)=\mathrm{e}^{-\im \pi(1-2\nu)}\frac{\Gamma(2\nu)}{2\pi\i}\int_{\Sigma}\zeta^{j+k}w(\zeta;z,\mu,\nu)\mathrm{d}\zeta,
\end{equation*}
provided $1-2\nu\notin\mathbb{N}$, in terms of the weight function
\begin{equation}\label{e3}
	w(\zeta;z,\mu,\nu):=\mathrm{e}^{-z(\zeta-1)}\zeta^{-2\mu}(\zeta-1)^{-2\nu},\ \ \ \zeta\in\mathbb{C}\setminus\big((-\infty,0]\cup[1,\infty)\big),
\end{equation}
with $\zeta^{-2\mu}$ and $(\zeta-1)^{-2\nu}$ defined on the complex $\zeta$-plane cut along $(-\infty,0]$ and $[1,\infty)$ such that $-\pi<\textnormal{arg}\,\zeta<\pi$ and $0<\textnormal{arg}\,(\zeta-1)<2\pi$, and contour of integration $\Sigma$ shown in Figure \ref{fig2}. In turn, \eqref{i14} can be rewritten as follows:
\begin{figure}[tbh]
\centering
	\begin{tikzpicture}[xscale=0.9,yscale=0.9]
	\draw [thick,decorate,decoration={zigzag,segment length=4,amplitude=1,post=lineto,post length=0}] (-1.05*\xr,0) -- (0,0);
        \draw [thick,->,decorate,decoration={zigzag,segment length=4,amplitude=1,pre=lineto,pre length=15,post=lineto,post length=3}] (0,0) -- (1.05*\xr,0) node [right] {$\footnotesize{\textnormal{Re}(\zeta)}$};
  \draw [thick,->] (0,-\yr) -- (0,\yr) node [above] {$\footnotesize{\textnormal{Im}(\zeta)}$};
  \draw[thick, xshift=5,red,decoration={markings,mark=between positions 0.125 and 0.875 step 0.25 with \arrow{<}},postaction={decorate}] (\xr,-\yr/4) -- (\yr/4,-\yr/4) arc (-90:-270:\yr/4) (\yr/4,\yr/4) -- (\xr,\yr/4) node[above left] {$\Sigma$};
\end{tikzpicture}
\caption{The oriented integration contour $\Sigma$ used in \eqref{e4} in the complex $\zeta$-plane in \textcolor{red}{red} colour with branch cuts on $(-\infty,0]$ for $\zeta\mapsto\zeta^{-\alpha}$ and on $[1,\infty)$ for $\zeta\mapsto(\zeta-1)^{-\alpha}$.}
\label{fig2}
\end{figure}
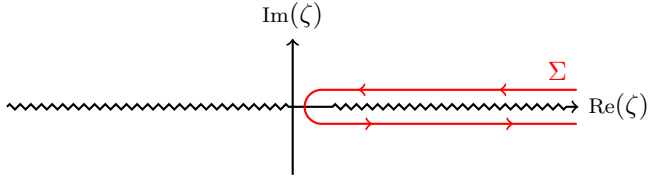
\begin{lem}\label{lem1} For any $N\in\mathbb{N},z\in\mathbb{C}:\textnormal{Re}(z)>0$ and $\mu,\nu\in\mathbb{C}: 1-2\nu\notin\mathbb{N}$, we have
\begin{equation*}
	\mathcal{J}_N(\mu,\nu,z)=\bigg[\mathrm{e}^{-\im \pi(1-2\nu)}\frac{\Gamma(2\nu)}{2\pi \im}\bigg]^NH_N[w]
\end{equation*}
with the Hankel determinant
\begin{equation}\label{e4}	H_n[w]:=\det\bigg[\int_{\Sigma}\zeta^{j+k}w(\zeta;z,\mu,\nu)\mathrm{d}\zeta\bigg]_{j,k=0}^{n-1},\ \ \ \ \ \ n\in\mathbb{N},
\end{equation}
generated by the weight \eqref{e3}.
\end{lem}
Seeing that \eqref{e4} is a moment determinant, we now apply the framework of Riemann-Hilbert problems for orthogonal polynomials \cite{FIK2,FIK3} to the same. First we establish that, for given $(n,\mu,\nu)\in\mathbb{N}\times\mathbb{C}^2$, the zeros of
\begin{equation}\label{e5}
	z\mapsto D_n(z,\mu,\nu):=\det\bigg[\int_{\Sigma}\zeta^{j+k}w(\zeta;z,\mu,\nu)\mathrm{d}\zeta\bigg]_{j,k=0}^{n-1}
\end{equation}
are isolated in the half-plane $z\in\mathbb{C}:\textnormal{Re}(z)>0$.
\begin{lem}\label{lem2} The function \eqref{e5} is analytic in $\mathbb{H}_+:=\{z\in\mathbb{C}:\,\textnormal{Re}(z)>0\}$, for any $(n,\mu,\nu)\in\mathbb{N}\times\mathbb{C}^2$, and not identically zero on $\mathbb{H}_+$.
\end{lem}
\begin{proof} We fix $(n,\mu,\nu)\in\mathbb{N}\times\mathbb{C}^2$. By multilinearity of the determinant,
\begin{equation}\label{e6}
	D_n(z,\mu,\nu)=\frac{1}{n!}\int_{\Sigma^n}\prod_{1\leq j<k\leq n}(\zeta_k-\zeta_j)^2\prod_{j=1}^{n}w(\zeta_{j};z,\mu,\nu)\mathrm{d} \zeta_{j},
\end{equation}
where the integrand in \eqref{e6} is bounded by a $z$-independent function in $L^1(\Sigma^n,\prod_1^n\mathrm{d}\zeta_j)$, provided $z$ lies in a compact subset of $\mathbb{H}_+$. Also, by Fubini's and Cauchy's theorem,
\begin{equation*}
	\oint_{\gamma}D_n(z,\mu,\nu)\mathrm{d}z=0,
\end{equation*}
for all triangles, $\gamma\subset\mathbb{H}_+$ since the integrand of $D_n$ is entire in $z$. Hence, \eqref{e5} is analytic in $\mathbb{H}_+$ by Morera's theorem. Lastly, if $z\mapsto D_n(z;\mu,\nu)$ vanishes identically on $\mathbb{H}_+$, then at least two columns in the matrix underwriting \eqref{e5} are linearly dependent on all of $\mathbb{H}_+$. But this contradicts the large $|z|$-asymptotics of 
\begin{equation*}
	z\mapsto\int_{\Sigma}\zeta^{j+k}w(\zeta;z,\mu,\nu)\mathrm{d}\zeta.
\end{equation*}
The proof of the Lemma is finished.
\end{proof}
Next, we relate $D_n$ to a set of orthonormal polynomials. Let 
\begin{equation}\label{zeroset}
    E_n^{\mu\nu}:=\big\{z\in\mathbb{H}_+:\ D_n(z,\mu,\nu)D_{n+1}(z,\mu,\nu)=0\big\}
\end{equation}
be the discrete zero locus of $D_nD_{n+1}$, for fixed $(n,\mu,\nu)\in\mathbb{Z}_{\geq 0}\times\mathbb{C}^2$, using the convention $D_0:=1$. Use
\begin{equation}\label{e7}
	\mathbb{C}[\zeta]\ni q_j(\zeta)=\kappa_j\Big(\zeta^j+\beta_j\zeta^{j-1}+\mathcal{O}\big(\zeta^{j-2}\big)\Big),\ \ \ \ \kappa_j=\kappa_j(z,\mu,\nu)\neq 0;\ \ \ \ \beta_j=\beta_j(z,\mu,\nu)\in\mathbb{C}
\end{equation}
to denote the degree $j\in\mathbb{Z}_{\geq 0}$ orthonormal polynomial with respect to $w(\zeta;z,\mu,\nu)$ on $\Sigma$. That is we have
\begin{equation*}
\int_{\Sigma}q_j(\zeta)q_k(\zeta)w(\zeta;z,\mu,\nu)\d\zeta=\delta_{jk},\ \ \ 0\leq k\leq j,
\end{equation*}
and $q_j\in\mathbb{C}[\zeta]$ exists precisely for $z\in\mathbb{H}_+\setminus E_j^{\mu\nu}$. Moreover, by classical theory,
\begin{equation}\label{e8}
	D_n(z,\mu,\nu)=\prod_{\ell=0}^{n-1}\kappa_{\ell}^{-2}(z,\mu,\nu),\ \ z\notin\bigcup_{j=1}^{n-1}E_j^{\mu\nu};\ \ \ \ \ \ \ \ \ \ \ \frac{D_{n+1}(z,\mu,\nu)}{D_n(z,\mu,\nu)}=\kappa_n^{-2}(z,\mu,\nu),\ \  z\notin E_n^{\mu\nu}.
\end{equation}
Having introduced $q_j$ in \eqref{e7} we proceed with the following well-known Riemann-Hilbert problem (RHP).
\begin{problem}\label{RHP1} Let $(n,\mu,\nu)\in\mathbb{N}\times\mathbb{C}^2$ and $z\in\mathbb{H}_+\setminus (E_n^{\mu\nu}\cup E_{n-1}^{\mu\nu})$. Now determine $Y(\lambda)=Y(\lambda;z,\mu,\nu,n)\in\mathbb{C}^{2\times 2}$ so that
\begin{enumerate}
	\item[(1)] $Y(\lambda)$ is analytic for $\lambda\in\mathbb{C}\setminus\Sigma$ and extends continuously to $\mathbb{C}$ from either side of $\Sigma$.
	\item[(2)] The continuous limiting values $Y_{\pm}(\lambda),\lambda\in\Sigma$ from either side of $\Sigma$, compare Figure \ref{fig3}, satisfy
	\begin{equation*}
		Y_+(\lambda)=Y_-(\lambda)\begin{bmatrix}1&w(\lambda;z,\mu,\nu)\\ 0&1\end{bmatrix}.
	\end{equation*}
	\begin{figure}[tbh]\centering
	\begin{tikzpicture}[xscale=0.9,yscale=0.9]
	\draw [thick] (-1.05*\xr,0) -- (0,0);
        \draw [thick,->] (0,0) -- (1.05*\xr,0) node [right] {$\footnotesize{\textnormal{Re}(\lambda)}$};
  \draw [thick,->] (0,-\yr) -- (0,\yr) node [above] {$\footnotesize{\textnormal{Im}(\lambda)}$};
  \draw[thick, xshift=5,red,decoration={markings,mark=between positions 0.125 and 0.875 step 0.25 with \arrow{<}},postaction={decorate}] (\xr,-\yr/4) -- (\yr/4,-\yr/4) arc (-90:-270:\yr/4) (\yr/4,\yr/4) -- (\xr,\yr/4) node[above left] {$\Sigma$};
\end{tikzpicture}
\caption{The oriented jump contour $\Sigma$ used in the RHP for $Y(\lambda)$, in the complex $\lambda$-plane.}
\label{fig3}
\end{figure}
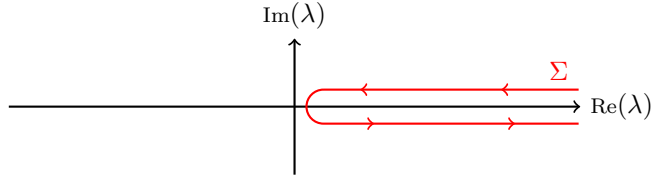
	\item[(3)] As $\lambda\rightarrow\infty,\lambda\notin\Sigma$,
	\begin{equation*}
		Y(\lambda)=\bigg\{I+\sum_{\ell=1}^2\frac{m_{\ell}}{\lambda^{\ell}}+\mathcal{O}\big(\lambda^{-3}\big)\bigg\}\lambda^{n\sigma_3},\ \ \ m_{\ell}=m_{\ell}(z,\mu,\nu,n)=\Big[m_{\ell}^{jk}(z,\mu,\nu,n)\Big]_{j,k=1}^2.
	\end{equation*}
\end{enumerate}
\end{problem}
By standard theory \cite{FIK2}, RHP \ref{RHP1} is uniquely solvable for $(n,\mu,\nu)\in\mathbb{N}\times\mathbb{C}^2$ and $z\in\mathbb{H}_+\setminus(E_n^{\mu\nu}\cup E_{n-1}^{\mu\nu})$, its solution being of the form
\begin{equation}\label{e9}
	Y(\lambda)=\begin{bmatrix}\kappa_n^{-1}q_n(\lambda) & \frac{\kappa_n^{-1}}{2\pi\im}\int_{\Sigma}q_n(\zeta)w(\zeta;z,\mu,\nu)\frac{\mathrm{d}\zeta}{\zeta-\lambda}\smallskip\\ -2\pi\im\kappa_{n-1}q_{n-1}(\lambda)&-\kappa_{n-1}\int_{\Sigma}q_{n-1}(\zeta)w(\zeta;z,\mu,\nu)\frac{\mathrm{d}\zeta}{\zeta-\lambda}\end{bmatrix},\ \ \ \lambda\in\mathbb{C}\setminus\Sigma
\end{equation}
in terms of \eqref{e7} and thus relating back to the determinant \eqref{e5} via the identity
\begin{equation}\label{e10}
	\frac{D_{n+1}(z,\mu,\nu)}{D_n(z,\mu,\nu)}\stackrel{\eqref{e8}}{=}\kappa_n^{-2}(z,\mu,\nu)=-2\pi\im \,m_1^{12}(z,\mu,\nu,n).
\end{equation}
In order to study the $z$-evolution of $D_n$, we proceed by reducing RHP \ref{RHP1} to a RHP with jumps that are $(\lambda,z,n)$-independent. This step is the key idea in the Riemann-Hilbert analysis of $D_n$ for finite $n$. Concretely, we define for $(n,\mu,\nu)\in\mathbb{N}\times\mathbb{C}^2$ and $-z\in\mathbb{H}_+\setminus(E_n^{\mu\nu}\cup E_{n-1}^{\mu\nu})$,
\begin{equation}\label{e11}
	\Psi(\lambda;z,\mu,\nu,n):=\lambda^{-\mu}(\lambda-1)^{-\nu}\mathrm{e}^{\frac{z}{2}\sigma_3}Y(\lambda;-z,\mu,\nu,n)\begin{bmatrix}\mathrm{e}^{\frac{z}{2}(\lambda-1)} & 0\\ 0 & \mathrm{e}^{-\frac{z}{2}(\lambda-1)}\lambda^{2\mu}(\lambda-1)^{2\nu}\end{bmatrix},\  \lambda\in\mathbb{C}\setminus\Sigma_{\Psi},
\end{equation}
with contour $\Sigma_{\Psi}:=\Sigma\cup(-\infty,0]\cup[1,\infty)$ shown in Figure \ref{fig4} and branch conventions as in \eqref{e3}: $\lambda^{\beta}$ and $(\lambda-1)^{\beta}$ are defined on the complex $\lambda$-plane cut along $(-\infty,0]$ and $[1,\infty)$ such that $-\pi<\textnormal{arg}\,\lambda<\pi$ and $0<\textnormal{arg}(\lambda-1)<2\pi$. Recalling RHP \ref{RHP1} we are led to the below RHP for $\Psi(\lambda)$.
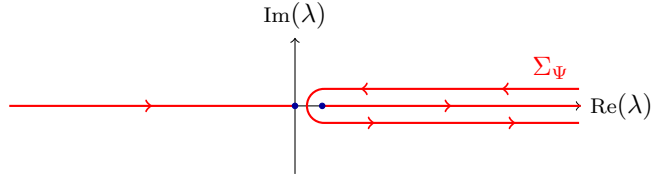
\begin{figure}[tbh]
\centering
	\begin{tikzpicture}[xscale=0.9,yscale=0.9]
	\draw [thick,red,decoration={markings,mark= at position 0.5 with {\arrow{>}}},postaction={decorate}] (-1.05*\xr,0) -- (0,0);
        \draw [->] (0,0) -- (1.05*\xr,0) node [right] {$\footnotesize{\textnormal{Re}(\lambda)}$};
        \draw [thick,red,decoration={markings,mark= at position 0.5 with {\arrow{>}}},postaction={decorate}] (0.4,0) -- (1.05*\xr,0);
  \draw [->] (0,-\yr) -- (0,\yr) node [above] {$\footnotesize{\textnormal{Im}(\lambda)}$};
  \draw[thick, xshift=5,red,decoration={markings,mark=between positions 0.125 and 0.875 step 0.25 with \arrow{<}},postaction={decorate}] (\xr,-\yr/4) -- (\yr/4,-\yr/4) arc (-90:-270:\yr/4) (\yr/4,\yr/4) -- (\xr,\yr/4) node[above left] {$\Sigma_{\Psi}$};
  \draw [fill, color=blue!60!black] (0.4,0) circle [radius=0.04];
  \draw [fill, color=blue!60!black] (0,0) circle [radius=0.04];

\end{tikzpicture}
\caption{The oriented jump contour $\Sigma_{\Psi}$ used in the RHP for $\Psi(\lambda)$, in the complex $\lambda$-plane drawn in \textcolor{red}{red}. The endpoints $\lambda=0,1$ are colored in \textcolor{blue!60!black}{blue}.}
\label{fig4}
\end{figure}

\begin{problem}\label{RHP2} Let $(n,\mu,\nu)\in\mathbb{N}\times\mathbb{C}^2$ and $-z\in\mathbb{H}_+\setminus (E_n^{\mu\nu}\cup E_{n-1}^{\mu\nu})$. The function $\Psi(\lambda)=\Psi(\lambda;z,\mu,\nu,n)\in\mathbb{C}^{2\times 2}$ defined in \eqref{e11} has the following properties:
\begin{enumerate}
	\item[(1)] $\Psi(\lambda)$ is analytic for $\lambda\in\mathbb{C}\setminus\Sigma_{\Psi}$ and extends continuously to $\mathbb{C}\setminus\{0,1\}$.
	\item[(2)] The continuous limiting values $\Psi_{\pm}(\lambda),\lambda\in\Sigma_{\Psi}\setminus\{0,1\}$ from either side of $\Sigma_{\Psi}\setminus\{0,1\}$, compare Figure \ref{fig4}, satisfy
	\begin{equation}\label{e12}
		\Psi_+(\lambda)=\Psi_-(\lambda)\begin{bmatrix}1&1\\ 0&1\end{bmatrix},\ \lambda\in\Sigma;\ \ \ \ \Psi_+(\lambda)=\Psi_-(\lambda)\begin{cases}\displaystyle
    \mathrm{e}^{-2\pi\im\mu\sigma_3},& \lambda\in(-\infty,0)\\
		\displaystyle\mathrm{e}^{2\pi\im \nu\sigma_3},&\lambda\in(1,\infty)\end{cases}.
	\end{equation}
	\item[(3)] There exist functions $\widehat{\Psi}_0(\lambda)$, resp. $\widehat{\Psi}_1(\lambda)$, analytic at $\lambda=0$, resp. at $\lambda=1$, such that
	\begin{equation}\label{e13}
		\Psi(\lambda)=\begin{cases}\displaystyle\widehat{\Psi}_0(\lambda)\lambda^{-\mu\sigma_3},& \lambda\rightarrow 0,\ \lambda\notin(-\infty,0]\\
		\widehat{\Psi}_1(\lambda)(\lambda-1)^{-\nu\sigma_3},& \lambda\rightarrow 1,\ \lambda\notin[1,\infty).
		\end{cases}
	\end{equation}
	Here, $\lambda^{\beta}$ and $(\lambda-1)^{\beta}$ are defined on $\mathbb{C}\setminus(-\infty,0]$ and $\mathbb{C}\setminus[1,\infty)$ with branches as in \eqref{e11}.
	\item[(4)] As $\lambda\rightarrow\infty$ and $\lambda\notin\Sigma_{\Psi}$,
	\begin{equation}\label{e14}
		\Psi(\lambda)=\bigg\{I+\sum_{\ell=1}^2\frac{n_{\ell}}{\lambda^{\ell}}+\mathcal{O}\big(\lambda^{-3}\big)\bigg\}\lambda^{n\sigma_3}\lambda^{-\mu\sigma_3}(\lambda-1)^{-\nu\sigma_3}\mathrm{e}^{\frac{z}{2}\lambda\sigma_3};
	\end{equation}
	with
	\begin{equation*}
		n_{\ell}(z,\mu,\nu,n)=\mathrm{e}^{\frac{z}{2}\sigma_3}m_{\ell}(-z,\mu,\nu,n)\mathrm{e}^{-\frac{z}{2}\sigma_3}.
	\end{equation*}
\end{enumerate}
\end{problem}
Seeing that \eqref{e12} are $(\lambda,z,n)$-independent unimodular matrices, we now derive the below result.
\begin{prop} RHP \ref{RHP2} is uniquely solvable for $(n,\mu,\nu)\in\mathbb{N}\times\mathbb{C}^2$ and $-z\in\mathbb{H}_+\setminus(E_n^{\mu\nu}\cup E_{n-1}^{\mu\nu})$ and its solution has $\det\Psi(\lambda)\equiv 1$ in $\mathbb{C}$. Moreover, $\Psi(\lambda)=\Psi(\lambda;z,\mu,\nu,n)$ satisfies the Lax system
\begin{equation}\label{e15}
	\frac{\partial\Psi}{\partial\lambda}(\lambda)=\bigg\{zA_{\infty}+\frac{A_0}{\lambda}+\frac{A_1}{\lambda-1}\bigg\}\Psi(\lambda),\ \ \ \frac{\partial\Psi}{\partial z}(\lambda)=\big\{\lambda B_{\infty}+B_0\big\}\Psi(\lambda),
\end{equation}
with $\lambda$-independent coefficients
\begin{equation*}
	A_{\infty}=B_{\infty}=\frac{1}{2}\sigma_3,\ \ A_0=-\mu\widehat{\Psi}_0(0)\sigma_3\widehat{\Psi}_0(0)^{-1},\ \ A_1=-\nu\widehat{\Psi}_1(1)\sigma_3\widehat{\Psi}_1(1)^{-1},\ \ B_0=\frac{1}{2}[n_1,\sigma_3].
\end{equation*}
The same coefficients satisfy the constraint
\begin{equation*}
	(n-\mu-\nu)\sigma_3+zB_0=A_0+A_1;\ \ \ \ \ \ [A,B]:=AB-BA,\ \ \ \ A,B\in\mathbb{C}^{2\times 2}.
\end{equation*}
\end{prop}
\begin{proof} By \eqref{e12}, $\det\Psi(\lambda)$ extends continuously across $\Sigma\cup(-\infty,0)\cup(1,\infty)$, so $\lambda\mapsto\det\Psi(\lambda)$ is analytic in $\mathbb{C}\setminus\{0,1\}$. However, \eqref{e13} shows that any singularities of $\det\Psi(\lambda)$ at $\lambda=0,1$ are removable, so $\lambda\mapsto\det\Psi(\lambda)$ in fact an entire function. But the same is normalized to unity at infinity by \eqref{e14}, and so $\det\Psi(\lambda)\equiv 1$ in $\mathbb{C}$ as consequence of Liouville's theorem. To arrive at \eqref{e15} we first note that $\lambda\mapsto\frac{\partial\Psi}{\partial \lambda}(\lambda)\Psi(\lambda)^{-1}$ is analytic in $\mathbb{C}\setminus\{0,1\}$ by \eqref{e12}. And, by \eqref{e14} and \eqref{e13},
\begin{equation*}
	\frac{\partial\Psi}{\partial \lambda}(\lambda)\Psi(\lambda)^{-1}=zA_{\infty}+\frac{1}{\lambda}\bigg\{(n-\mu-\nu)\sigma_3+\frac{z}{2}[n_1,\sigma_3]\bigg\}+\mathcal{O}\big(\lambda^{-2}\big),\ \ \ \lambda\rightarrow\infty;
\end{equation*}
\begin{equation*}
	\frac{\partial\Psi}{\partial\lambda}(\lambda)\Psi(\lambda)^{-1}=\frac{A_0}{\lambda}+\mathcal{O}(1),\ \ \ \ \ \lambda\rightarrow 0;\ \ \ \ \ \ \ \frac{\partial\Psi}{\partial\lambda}(\lambda)\Psi(\lambda)^{-1}=\frac{A_1}{\lambda-1}+\mathcal{O}(1),\ \ \ \ \ \lambda\rightarrow 1.
\end{equation*}
Thus, again by Liouville's theorem,
\begin{equation*}
	\frac{\partial\Psi}{\partial\lambda}(\lambda)\Psi(\lambda)^{-1}=zA_{\infty}+\frac{A_0}{\lambda}+\frac{A_1}{\lambda-1},\ \ \ \ \lambda\in\mathbb{C}\setminus\{0,1\},
\end{equation*}
and we have the identity
\begin{equation*}
	(n-\mu-\nu)\sigma_3+\frac{z}{2}[n_1,\sigma_3]=A_0+A_1.
\end{equation*}
A similar argument for $\lambda\mapsto\frac{\partial\Psi}{\partial z}(\lambda)\Psi(\lambda)^{-1}$ yields the second system in \eqref{e15} and completes this proof.
\end{proof}
At this point of our calculation we compare \eqref{e15} to system \cite[$(5.4.4),(5.4.5)$]{FIKN}. Firstly, the $\lambda$-equation in \eqref{e15} constitutes a $2\times 2$ ODE system with rational coefficients and three singular points: two Fuchsian singular points at $\lambda=0$ and $\lambda=1$, and one irregular singular point at $\lambda=\infty$ of Poincar\'e rank one. Secondly, the $z$-equation in \eqref{e15} states that the $z$-dependence of the coefficients in the $\lambda$-equation is monodromy preserving, i.e. we are necessarily dealing with a Painlev\'e type isomonodromic deformation of the $\lambda$-system, and are thus left to match our data \eqref{e15} with the data \cite[$(5.4.4),(5.4.5)$]{FIKN} in order to obtain the following connection of RHP \ref{RHP2} to Painlev\'e special function theory.
\begin{prop}\label{JMUpara} Introduce the parameters
\begin{equation*}
	\theta_0:=-2\mu,\ \ \ \ \ \ \theta_1:=-2\nu,\ \ \ \ \ \ \theta_{\infty}:=2(\mu+\nu-n)
\end{equation*}
and define, in terms of the data \eqref{e13},\eqref{e14}, as functions of $(z,\mu,\nu,n)$,
\begin{equation*}
	v:=-2\mu\widehat{\Psi}_0^{12}(0)\widehat{\Psi}_0^{21}(0),\ \ \ \ \ \ y:=\frac{\widehat{\Psi}_0^{12}(0)}{\widehat{\Psi}_0^{22}(0)},\ \ \ \ \ \ u:=\frac{y(v+\theta_0)-zn_1^{12}}{y(v+\frac{1}{2}(\theta_0+\theta_{\infty}-\theta_1))}.
\end{equation*}
Then the matrix-valued coefficients $A_0,A_1$ and $B_0$ in \eqref{e15} are of the form
\begin{equation*}
	A_0=\begin{bmatrix}v+\frac{1}{2}\theta_0 & -y(v+\theta_0)\smallskip\\ \frac{v}{y} &-v-\frac{1}{2}\theta_0\end{bmatrix},\ \ \ A_1=\begin{bmatrix}-v-\frac{1}{2}(\theta_0+\theta_{\infty}) & uy(v+\frac{1}{2}(\theta_0+\theta_{\infty}-\theta_1))\smallskip\\ -\frac{1}{uy}(v+\frac{1}{2}(\theta_0+\theta_{\infty}+\theta_1)) & v+\frac{1}{2}(\theta_0+\theta_{\infty})\end{bmatrix},
\end{equation*}
\begin{equation*}
	B_0=\frac{1}{z}\begin{bmatrix}0 & -y(v+\theta_0-u(v+\frac{1}{2}(\theta_0+\theta_{\infty}-\theta_1)))\smallskip\\ \frac{1}{y}(v-\frac{1}{u}(v+\frac{1}{2}(\theta_0+\theta_{\infty}+\theta_1)))& 0\end{bmatrix},
\end{equation*}
and $v,y,u$ satisfy the nonlinear ODE system
\begin{align}
	z\frac{\d u}{\d z}=&\,\,zu-2v(u-1)^2-\frac{1}{2}(u-1)\big((\theta_0-\theta_1+\theta_{\infty})u-(3\theta_0+\theta_1+\theta_{\infty})\big),\label{e16}\\
	z\frac{\d v}{\d z}=&\,\,uv\bigg(v+\frac{1}{2}(\theta_0-\theta_1+\theta_{\infty})\bigg)-\frac{1}{u}(v+\theta_0)\bigg(v+\frac{1}{2}(\theta_0+\theta_1+\theta_{\infty})\bigg),\label{e17}\\
	z\frac{\d}{\d z}\ln y=&\,-2v-\theta_0+u\bigg(v+\frac{1}{2}(\theta_0-\theta_1+\theta_{\infty})\bigg)+\frac{1}{u}\bigg(v+\frac{1}{2}(\theta_0+\theta_1+\theta_{\infty})\bigg).\label{e18}
\end{align}
\end{prop}
\begin{proof} Observe that $A_0,A_1$ in \eqref{e15} are traceless with $\det A_0=-\mu^2,\det A_1=-\nu^2$, and comparing our \eqref{e13} to \cite[Proposition $5.9$]{FIKN} we readily identify $\theta_0=-2\mu,\theta_1=-2\nu$. Likewise, by the same reference, the diagonal part of $A_0+A_1$ equals $-\frac{1}{2}\theta_{\infty}\sigma_3$ which by our constraint $(n-\mu-\nu)\sigma_3+zB_0=A_0+A_1$ yields $\theta_{\infty}=2(\mu+\nu-n)$. Then, using \eqref{e15}, we readily obtain the parametrization of $A_0$ in terms of $v,\theta_0$ and $y$ and the constraint $(n-
\mu-\nu)\sigma_3+zB_0=A_0+A_1$ yields $A_1^{11},A_1^{22}$ as written in terms of $v,\theta_0$ and $\theta_{\infty}$. In turn, $A_1^{12},A_1^{21}$ can be parametrized in terms of $u,y,v,\theta_0,\theta_{\infty}$ and $\theta_1$ as indicated, where $u$ is determined via
\begin{equation*}
	A_1=(n-\mu-\nu)\sigma_3+\frac{z}{2}[n_1,\sigma_3]-A_0.
\end{equation*}
The remaining formula for $B_0$ follows from $B_0=\frac{1}{z}(A_0+A_1-(n-\mu-\nu)\sigma_3)$ and with that our \eqref{e15} matches exactly \cite[$(5.4.4),(5.4.5)$]{FIKN}. Seeing that \eqref{e16},\eqref{e17},\eqref{e18} are equivalent to the Frobenius integrability constraint of \eqref{e15}, compare \cite[Theorem $5.5$]{FIKN}, we have thus completed our proof.
\end{proof}
\begin{remark} Equations \eqref{e16},\eqref{e17} yield Painlev\'e-V for $u(z)$, compare \cite[(C.41)]{JM},
\begin{equation*}
	\frac{\d^2 u}{\d z^2}=\bigg(\frac{1}{2u}+\frac{1}{u-1}\bigg)\bigg(\frac{\d u}{\d z}\bigg)^2-\frac{1}{z}\frac{\d u}{\d z}+\frac{(u-1)^2}{z^2}\bigg(\alpha u+\frac{\beta}{u}\bigg)+\frac{\gamma u}{z}+\frac{\delta u(u+1)}{u-1},
\end{equation*}
with
\begin{equation*}
	\alpha=\frac{1}{8}(\theta_0-\theta_1+\theta_{\infty})^2,\ \ \ \ \ \beta=-\frac{1}{8}(\theta_0-\theta_1-\theta_{\infty})^2,\ \ \ \ \gamma=1-\theta_0-\theta_1,\ \ \ \ \delta=-\frac{1}{2}.
\end{equation*}
\end{remark}
We are now left to express $D_n(z,\mu,\nu)$ in \eqref{e5}, consequently $J_N(\mu,\nu,z)$ in Lemma \ref{lem1}, in terms of Painlev\'e-V data: first, inserting \eqref{e14} into the $\lambda$-system in \eqref{e15} we have near $\lambda=\infty$,
\begin{align*}
	\frac{\partial\Psi}{\partial\lambda}(\lambda)\Psi(\lambda)^{-1}=zA_{\infty}+&\,\frac{1}{\lambda}\bigg\{(n-\mu-\nu)\sigma_3+\frac{z}{2}[n_1,\sigma_3]\bigg\}\\
	+&\,\frac{1}{\lambda^2}\bigg\{-n_1+(n-\mu-\nu)[n_1,\sigma_3]-\nu\sigma_3+\frac{z}{2}\Big([n_2,\sigma_3]-[n_1,\sigma_3]n_1\Big)\bigg\}+\mathcal{O}\big(\lambda^{-3}\big)
\end{align*}
and thus by comparison in \eqref{e15} the additional constraint
\begin{equation*}
	A_1+n_1=(n-\mu-\nu)[n_1,\sigma_3]+\frac{z}{2}\Big([n_2,\sigma_3]-[n_1,\sigma_3]n_1\Big)-\nu\sigma_3.
\end{equation*}
The same says, with $A_1^{11}$ taken from Proposition \ref{JMUpara},
\begin{align}
	n_1^{11}=&\,\,v+\frac{1}{2}(\theta_0+\theta_{\infty})-\nu+zn_1^{12}n_1^{21}=v+\frac{1}{2}(\theta_0+\theta_{\infty})-\nu\nonumber\\
	&\hspace{1cm}+\frac{1}{z}\bigg[v-\frac{1}{u}\bigg(v+\frac{1}{2}(\theta_0+\theta_{\infty}+\theta_1)\bigg)\bigg]\bigg[v+\theta_0-u\bigg(v+\frac{1}{2}(\theta_0+\theta_{\infty}-\theta_1)\bigg)\bigg],\label{e19}
\end{align}
and consequently, with
\begin{equation}\label{e20}
	H(z,\mu,\nu,n):=-n_1^{11}(z,\mu,\nu,n)-\nu,
\end{equation}
we arrive at the following intermediate result.
\begin{prop}\label{prop3} Let $(n,\mu,\nu)\in\mathbb{N}\times\mathbb{C}^2$ and $-z\in\mathbb{H}_+\setminus(E_n^{\mu\nu}\cup E_{n-1}^{\mu\nu})$. Introduce
\begin{equation}\label{e21}
	\sigma(z,\mu,\nu,n):=zH(z,\mu,\nu,n)+\frac{1}{2}(\theta_0+\theta_{\infty})z+\frac{1}{4}\big((\theta_0+\theta_{\infty})^2-\theta_1^2\big),
\end{equation}
with $(\theta_0,\theta_1,\theta_{\infty})$ as in Proposition \ref{JMUpara}. Then $\sigma=\sigma(z,\mu,\nu,n)$ solves the sigma-Painlev\'e equation
\begin{equation}\label{e22}
	\bigg(z\frac{\d^2\sigma}{\d z^2}\bigg)^2=\bigg(\sigma-z\frac{\d\sigma}{\d z}+2\Big(\frac{\d\sigma}{\d z}\Big)^2+\frac{\d\sigma}{\d z}\sum_{\ell=0}^3\nu_{\ell}\bigg)^2-4\prod_{\ell=0}^4\bigg(\nu_{\ell}+\frac{\d\sigma}{\d z}\bigg),
\end{equation}
with parameters $\{\nu_{\ell}\}_{\ell=0}^3$ equal to
\begin{equation*}
	\nu_0:=0,\ \ \ \nu_1:=-\frac{1}{2}(\theta_0-\theta_1+\theta_{\infty}),\ \ \ \nu_2:=-\theta_0,\ \ \ \nu_3:=-\frac{1}{2}(\theta_0+\theta_1+\theta_{\infty}).
\end{equation*}
\end{prop}
\begin{proof} With \eqref{e19} and \eqref{e16},\eqref{e17},\eqref{e18} in place we merely need to match our $n_1^{11}$ with \cite[(C.43)]{JM}. This is achieved by \eqref{e20}, noting that our $H$ is $H_V$ in \cite[(C.43)]{JM}. The rest is copy-paste from \cite[(C.44),(C.45)]{JM}.
\end{proof}
We now return to $D_n$ in \eqref{e5}, and thus ultimately $\mathcal{J}_N(\mu,\nu,z)$, compare Lemma \ref{lem1}. Both quantities link to \eqref{e21} and hence to the sigma-Painlev\'e equation \eqref{e22}.
\begin{prop}\label{prop4} Let $(n,\mu,\nu)\in\mathbb{N}\times\mathbb{C}^2$ and $z\in\mathbb{H}_+\setminus\bigcup_{j=1}^nE_j^{\mu\nu}$. Then
\begin{equation*}
	\frac{\partial}{\partial z}\ln D_n(z,\mu,\nu)=n-\nu-H(-z,\mu,\nu,n),
\end{equation*}
with $H$ defined in \eqref{e20} and expressed in terms of sigma-Painlev\'e-V in \eqref{e21},\eqref{e22}. The branch for the logarithm is arbitrary, but fixed throughout.
\end{prop}
\begin{proof} By \eqref{e8} and \eqref{e10}, where we require $z\notin\bigcup_{j=1}^nE_j^{\mu\nu}$ so $\kappa_{\ell}(x,\alpha)\neq 0$ for $\ell\in\{0,1,\ldots,n-1\}$,
\begin{equation}\label{e23}
	\ln D_n(z,\mu,\nu)=\ln D_1(z,\mu,\nu)+\sum_{\ell=1}^{n-1}\ln\kappa_{\ell}^{-2}(z,\mu,\nu),\ \ \ \ \ \kappa_{\ell}^{-2}(z,\mu,\nu)=-2\pi\im\,\mathrm{e}^zn_1^{12}(-z,\mu,\nu,\ell).
\end{equation}
On the other hand, $\Psi_n(\lambda)\equiv\Psi(\lambda;z,\mu,\nu,n)$ is such that, by another application of Liouville's theorem, utilizing \eqref{e12},\eqref{e13} and \eqref{e14},
\begin{equation}\label{e24}
	\Psi_{\ell+1}(\lambda)=\Big\{\lambda U_{\infty}+U_0^{(\ell)}\Big\}\Psi_{\ell}(\lambda),\ \ \ \ \ell\in\{1,\ldots,n-1\}\subset\mathbb{N},
\end{equation}
with coefficients
\begin{equation*}
	U_{\infty}=\begin{bmatrix}1&0\\ 0&0\end{bmatrix},\ \ \ U_0^{(\ell)}=\begin{bmatrix}n_1^{11}(z,\mu,\nu,\ell+1)-n_1^{11}(z,\mu,\nu,\ell) & -n_1^{12}(z,\mu,\nu,\ell)\smallskip\\ n_1^{21}(z,\mu,\nu,\ell+1) & 0\end{bmatrix}.
\end{equation*}
Hence, the $z$-system in \eqref{e15}, written with $B_0=B_0^{(\ell)}$ to indicate explicitly its $\ell$-dependence, and the difference system \eqref{e24} yield the equations
\begin{equation*}
	U_{\infty}B_0^{(\ell)}=\Big[B_{\infty},U_0^{(\ell)}\Big]+B_0^{(\ell+1)}U_{\infty},\ \ \ \ \ \ \frac{\partial U_0^{(\ell)}}{\partial z}=B_0^{(\ell+1)}U_0^{(\ell)}-U_0^{(\ell)}B_0^{(\ell)},
\end{equation*}
which say that, for any $\ell\in\{1,\ldots,n-1\}$ and $z\in\mathbb{H}_+\setminus\bigcup_{j=1}^nE_j^{\mu\nu}$,
\begin{equation*}
	\frac{\partial}{\partial z}\Big(n_1^{11}(z,\mu,\nu,\ell+1)-n_1^{11}(z,\mu,\nu,\ell)\Big)=n_1^{12}(z,\mu,\nu,\ell)n_1^{21}(z,\mu,\nu,\ell)-n_1^{12}(z,\mu,\nu,\ell+1)n_1^{21}(z,\mu,\nu,\ell+1),
\end{equation*}
followed by
\begin{align*}
	\begin{cases}\displaystyle\frac{\partial n_1^{12}}{\partial z}(z,\mu,\nu,\ell)=-n_1^{12}(z,\mu,\nu,\ell)\big(n_1^{11}(z,\mu,\nu,\ell+1)-n_1^{11}(z,\mu,\nu,\ell)\big)&\bigskip\\
	\displaystyle\frac{\partial n_1^{21}}{\partial z}(z,\mu,\nu,\ell+1)=n_1^{21}(z,\mu,\nu,\ell+1)\big(n_1^{11}(z,\mu,\nu,\ell+1)-n_1^{11}(z,\mu,\nu,\ell)\big)
	\end{cases}.
\end{align*}
In turn, by \eqref{e23},
\begin{align}
	\frac{\partial}{\partial z}\ln D_n(z,\mu,\nu)=&\,\,\frac{\partial}{\partial z}\ln D_1(z,\mu,\nu)+(n-1)+\sum_{\ell=1}^{n-1}\frac{\partial}{\partial z}\ln\big(n_1^{12}(-z,\mu,\nu,\ell)\big)\nonumber\\
	=&\,\,\frac{\partial}{\partial z}\ln D_1(z,\mu,\nu)+(n-1)+n_1^{11}(-z,\mu,\nu,n)-n_1^{11}(-z,\mu,\nu,1),\label{e25}
\end{align}
i.e. we are left to investigate $\frac{\partial}{\partial z}\ln D_1(z,\mu,\nu)$ and $n_1^{11}(-z,\mu,\nu,1)$. To that end, we begin with the identity
\begin{equation}\label{e26}
	Y(\lambda;z,\mu,\nu,1)\stackrel{\eqref{e9}}{=}\begin{bmatrix}\lambda+\beta_1 & \frac{1}{2\pi\im}\int_{\Sigma}(\zeta+\beta_1)w(\zeta)\frac{\d\zeta}{\zeta-\lambda}\smallskip\\ -2\pi\im\kappa_0^2 & -\kappa_0^2\int_{\Sigma}w(\zeta)\frac{\d\zeta}{\zeta-\lambda}\end{bmatrix},\ \ \ \lambda\in\mathbb{C}\setminus\Sigma,
\end{equation}
where $q_1(\zeta)=\kappa_1(\zeta+\beta_1)$, resp. $q_0(\zeta)=\kappa_0$, are the degree one, resp. zero, orthonormal polynomials with respect to the weight $w$ on $\Sigma$, compare \eqref{e7}. Note that \eqref{e26} solves RHP \ref{RHP1} for $n=1$, provided $z\in\mathbb{H}_+\setminus E_1^{\mu\nu}$. However
\begin{equation*}
	D_1(z,\mu,\nu)=\kappa_0^{-2}=\int_{\Sigma}w(\zeta)\d\zeta\ \ \Rightarrow\ \ \frac{\partial}{\partial z}\ln D_1(z,\mu,\nu)\stackrel{\eqref{e3}}{=}-\kappa_0^2\int_{\Sigma}w(\zeta)(\zeta-1)\d\zeta=1-\kappa_0^2\int_{\Sigma}w(\zeta)\zeta\,\d\zeta,
\end{equation*}
and since by \eqref{e26}, exploiting en route $\tr(m_1(z,\mu,\nu,n))=0$ since $\det Y(\lambda)\equiv 1$,
\begin{equation*}
	-m_1^{11}(z,\mu,\nu,1)=m_1^{22}(z,\mu,\nu,1)=\kappa_0^2\int_{\Sigma}w(\zeta)\zeta\,\d\zeta,
\end{equation*}
we find from condition $(4)$ in RHP \ref{RHP2} that
\begin{equation*}
	\frac{\partial}{\partial z}\ln D_1(z,\mu,\nu)=1+m_1^{11}(z,\mu,\nu,1)=1+n_1^{11}(-z,\mu,\nu,1).
\end{equation*}
Inserting the latter into \eqref{e25} completes the proof via \eqref{e20}.
\end{proof}
We are ready to state our main result for $D_n(z,\mu,\nu)$.
\begin{theo}\label{theo1} Fix $(n,\mu,\nu)\in\mathbb{N}\times\mathbb{C}^2$ and suppose $z\in\mathbb{H}_+\setminus\bigcup_{j=1}^nE_j^{\mu\nu}$. Then
\begin{equation*}
	\tau_n(z,\mu,\nu):=z\frac{\d}{\d z}\ln D_n(z,\mu,\nu)\stackrel{\eqref{e5}}{=}z\frac{\d}{\d z}\ln\det\bigg[\int_{\Sigma}\zeta^{i+j}w(\zeta;z,\mu,\nu)\d\zeta\bigg]_{i,j=0}^{n-1},
\end{equation*}
solves the $\sigma$-Painlev\'e {\rm{V}} equation
\begin{align}
	\bigg(z\frac{\d^2\tau_n}{\d z^2}\bigg)^2=\bigg(\tau_n+n(n-2\nu)-&z\frac{\d\tau_n}{\d z}+2\Big(\frac{\d\tau_n}{\d z}\Big)^2-2(\mu-\nu+n)\frac{\d\tau_n}{\d z}\bigg)^2\nonumber\\
	&+4\frac{\d\tau_n}{\d z}\bigg(n-2\nu-\frac{\d\tau_n}{\d z}\bigg)\bigg(2\mu-\frac{\d\tau_n}{\d z}\bigg)\bigg(n-\frac{\d\tau_n}{\d z}\bigg).\label{e27}
\end{align}

\end{theo}
\begin{proof} By \eqref{e21} and Proposition \ref{prop4}, using $\theta_0+\theta_{\infty}=2(\nu-n),\theta_1=-2\nu$,
\begin{equation*}
	z\frac{\d}{\d z}\ln D_n(z,\mu,\nu)=\sigma(-z,\mu,\nu,n)-n(n-2\nu),
\end{equation*}
and $\sigma=\sigma(z,\mu,\nu,n)$ solves \eqref{e22}. In turn, $\tau_n(z,\mu,\nu)=\sigma(-z,\mu,\nu,n)-n(n-2\nu)$ solves \eqref{e27}.
\end{proof}

The above Theorem \ref{theo1} concludes the current section on the Riemann-Hilbert analysis of $\mathcal{J}_N$ for by Lemma \ref{lem1},
\begin{equation}\label{0403eve}
	z\frac{\mathrm{d}}{\mathrm{d} z}\ln\mathcal{J}_N(\mu,\nu,z)=z\frac{\mathrm{d}}{\mathrm{d}z}\ln D_N(z,\mu,\nu)=\tau_N(z,\mu,\nu),\ \ \ z\in\mathbb{H}_+\setminus\bigcup_{j=1}^NE_j^{\mu\nu}
\end{equation}
with $N\in\mathbb{N}$ and $\mu,\nu\in\mathbb{C}:1-2\nu\notin\mathbb{N}$.

\section{$F_{N,\delta}(s,h)$ in terms of a solution to a $\sigma$-Painlev\'e V equation}\label{representation}
We recall \eqref{i13}, the special parameter choice $\mu=\frac{1}{2}(N+w)$, $\nu=\frac{1}{2}(N+\overline{w})$ as well as $z=\xi\in\mathbb{R}_{\geq 0}$. Let $w\in \mathbb{C}$ with $\textnormal{Re}(2w)>-1$. Introduce the exceptional set
\bea\label{zeroparameterw}
E_{N,w}:=\bigg\{z\in\mathbb{H}_+:\ \mathcal{J}_{N}\bigg(z,\frac{N+w}{2},\frac{N+\overline{w}}{2}\bigg)\mathcal{J}_{N+1}\bigg(z,\frac{N+w}{2},\frac{N+\overline{w}}{2}\bigg)=0\bigg\},
\eea
where $\mathcal{J}_{N}(z,\mu,\nu)$ is defined in \eqref{i14}. Recall that $E_{N}^{\mu\nu}$ in \eqref{zeroset} coincides with $E_{N,w}$ when $w \notin \mathbb{Z}_{\ge 0}$. 
When $w\in \mathbb{Z}_{\ge 0}$, or more generally, when $w\in \mathbb{R}$ with $2w>-1$, it is known that $E_{N,w}$ is a discrete subset of $\mathbb{H}_{+}$ and that $\mathbb{R}_{+}\cap E_{N,w}=\emptyset$; see, e.g., \cite[Propositions 2.5 and 2.7]{BW}. What results from our workings in Section \ref{RHmethod} is the following $\tau$-function identity:

\begin{prop}\label{detailedanalysisofthesolution}
Let $N\in\mathbb{N}$ and let $w\in\mathbb{C}$ with $\textnormal{Re}(2w)>-1$. Choose $\mu=\frac{1}{2}(N+w)$, $\nu=\frac{1}{2}(N+\overline{w})$. Let $\mathcal{J}_{N}(\mu,\nu, \xi)$ be given as in \eqref{i13}. Then for $\xi>0$,
\bea\label{e28}
\mathcal{J}_N(\mu,\nu, \xi) = \mathcal{J}_{N}(\mu,\nu, 0)\exp\left( 
\int_0^{\xi} u_N(z;w) \frac{\mathrm{d}z}{z}
\right),
\eea
where the path of integration lies in $\mathbb{H}_+ \setminus E_{N,w}$ with $E_{N,w}$ in \eqref{zeroparameterw}, and we assume that $\mathcal{J}_N(\mu,\nu, \xi)\neq 0$ for the same parameters $(w,N,\xi)$. Here 
\bea\label{theevaluationofJNat0}
\mathcal{J}_{N}(\mu,\nu, 0)=(-1)^{\frac{N}{2}(N-1) }\Gamma(\overline{w}+N)^N\prod_{j=1}^N \frac{ 
\Gamma(j) \Gamma(w+\overline{w}+j)
}{\Gamma(w+j) \Gamma(\overline{w}+j)}
\eea
and $u_N(z;w)$ solves the  $\sigma$-Painlev\'e {\rm V} equation 
\bea\label{equationbeforescaling}
\left(z \frac{\d^2u_N}{\d z^2} \right)^2 &=& \left( u_N - N\overline{w} - z \frac{\d u_N}{\d z} + 2 \left( \frac{\d u_N}{\d z} \right)^2 
- (w - \overline{w} + 2N) \frac{\d u_N}{\d z} \right)^2 \nonumber \\
&& +\, 4 \frac{\d u_N}{\d z} \left( -\overline{w} - \frac{\d u_N}{\d z} \right) \left( N + w - \frac{\d u_N}{\d z} \right) \left( N - \frac{\d u_N}{\d z} \right)
\eea
with the following boundary constraints as $z\rightarrow 0$ for $z\in \mathbb{H}_+ \setminus E_{N,w}$. If $w=0$, then $u_N\equiv0$. If $w\in \mathbb{Z}_{\geq 1}$,
\bea\label{Cauchydata1new}
u_N(z;w) \sim z \left(n_0(N,w)+ \sum_{j=1}^{\infty} n_{j}(N,w) z^{j} \right).
\eea
If $w+\overline{w} \in \mathbb{Z}_{\geq 0}$ and $w\notin \mathbb{Z}_{\geq 0}$,
\bea\label{Cauchydata2new}
u_N(z;w) \sim z \left(h_0 (N,w)+ \sum_{j=1}^\infty h_{j}(N,w) z^{j}
+\sum_{k=1}^\infty \sum_{j=k-1}^\infty m_{jk}(N,w) 
(\ln z)^k z^{j+k(w+\overline{w})} \right).
\eea
Finally, if $w+\overline{w} \notin \mathbb{Z}_{\geq 0}$, 
\bea\label{Cauchydata3new}
u_N(z;w) \sim z \left(f_0(N,w)+ \sum_{j=1}^\infty f_{j}(N,w) z^{j}
+\sum_{k=1}^\infty \sum_{j=
k-1}^\infty g_{jk}(N,w) z^{j+k(w+\overline{w})} \right).
\eea
Moreover, if $w+\overline{w} \neq 0$, then 
\bea\label{initialconditionnew}
n_0(N,w)=f_0(N,w)=h_0(N,w)=\frac{N\overline{w}}{w+\overline{w} }.
\eea
\end{prop}
\begin{proof}
By Proposition \ref{Hankelrepresentation}, we have the explicit formula for $\mathcal{J}_N(\mu,\nu,0)$ in \eqref{theevaluationofJNat0}. For $z\in \mathbb{H}_+ \setminus E_{N,w}$, let
\bea\label{continuationofuN}
u_{N}(z;w)=z\frac{\mathrm{d}}{\mathrm{d}z}\ln \mathcal{J}_{N}(\mu,\nu,z).
\eea
Note that $\textnormal{Re}(\nu)>0$. So if $2\nu \in \mathbb{Z}_{\geq 1}$, then $1-2\nu$ is a non-positive integer. So $U(1-2\nu, 2-2\nu-2\mu+k,z)$ is a multiple of a generalized Laguerre polynomial. Then
\bea\label{data1}
u_{N}(z;w)=z\frac{\mathrm{d}}{\mathrm{d}z}\ln \det \big[f_{j+k}(z)\big]_{j,k=0}^{N-1},
\eea
with $f_{\ell}(z)$ a polynomial in $z$ of degree at most $2\nu-1$ and $f_{\ell}(0)\neq 0$.
If $2\nu \notin \mathbb{Z}_{\geq 1}$, note that  $\textnormal{Re}(\mu)>0$,  $\textnormal{Re}(\nu)>0$, so if
$2(\mu+\nu)\in \mathbb{Z}_{\geq 1} $, we note that $\textnormal{Re}(2\mu+2\nu)>2N-1$,  so $2\mu+2\nu\geq 2N$ and $2\mu \notin \mathbb{Z}_{\geq 1}$. By Kummer's transformation \cite[$13.2.40$]{NIST},
\bea\label{kummertransformation}
U(1-2\nu, 2-2\nu-2\mu+k,z)
=z^{2\nu+2\mu-k-1} U(2\mu-k, 2\nu+2\mu-k, z).
\eea
Note that 
$-1+2\mu-k\notin \mathbb{Z}$, by \cite[$13.2.9$]{NIST},
we have
\bea\label{data2}
u_{N}(z;w)=z\frac{\mathrm{d}}{\mathrm{d}z}\ln \det \big[f_{j+k}(z)+z^{2(\mu+\nu)-1-j-k}g_{j+k}(z)\ln z\big]_{j,k=0}^{N-1}
\eea
with $z\mapsto f_k(z)$ and $z\mapsto g_k(z)$ analytic at $z=0$ and $f_{k}(0)g_{k}(0)\neq 0$. 
If $2\nu \notin \mathbb{Z}_{\geq 1}$ and 
$2(\mu+\nu) \notin \mathbb{Z}_{\geq 1}$,
by \cite[$13.2.42$]{NIST},
\begin{align*}
U(1-2\nu, 2-2\nu-2\mu+k, z)
&=\frac{\Gamma(2\nu+2\mu-k-1)}{\Gamma(2\mu-k)} M(1-2\nu, 2-2\nu-2\mu+k,z) \\
& + \frac{\Gamma(1-2\nu-2\mu+k)}{\Gamma(1-2\nu)} z^{-1+2\nu+2\mu-k} M(2\mu-k, 2\nu+2\mu-k, z),
\end{align*}
where $M(a,b;z)$ is the confluent hypergeometric function of the first kind and principal branches are chosen.
Note that $\textnormal{Re}(2\mu+2\nu)>2N-1$, so for $0\leq \ell \leq 2N-2$, so $\textnormal{Re}(-1+2\nu+2\mu-\ell)>0$. Since $z\mapsto M(a,b,z)$ is entire we thus have 
\bea\label{data3}
u_{N}(z;w)=z\frac{\mathrm{d}}{\mathrm{d}z}\ln \det \big[f_{j+k}(z)+z^{2(\mu+\nu)-1-j-k}g_{j+k}(z)\big]_{j,k=0}^{N-1}
\eea
with $z\mapsto f_k(z)$ and $z\mapsto g_k(z)$ entire, $f_k(0)g_k(0)\neq 0$. Observe that in (\ref{data2}) and (\ref{data3}), $2(\mu+\nu)-1-j-k=w+\bar{w}+2N-1-(j+k)$. Combining this with  \eqref{data1}, \eqref{data2}, \eqref{data3} and Jacobi's formula, we readily obtain the structure of the asymptotic expansion of $u_N(z;w)$ at $z=0$ stated in \eqref{Cauchydata1new},\eqref{Cauchydata2new},\eqref{Cauchydata3new}. As a consequence, it follows that the function
 $z\mapsto u_N(z;s)\frac{1}{z}$ is integrable near $z=0$ in $\mathbb{H}_+$. Therefore, \eqref{e28} and \eqref{theevaluationofJNat0} hold with the indicated value of $\mathrm{C}_{N,w}$ given in (\ref{definitionoftheconstant0407}). Regarding (\ref{equationbeforescaling}), when $w\notin \mathbb{Z}_{\geq 0}$, note that it holds for 
\( z \in \mathbb{H}_+ \setminus \bigcup_{j=1}^N E_{j}^{\mu\nu} \) by Theorem \ref{theo1} and (\ref{0403eve}); 
and when $w\in \mathbb{Z}_{\geq 0}$, it holds for $z\in \mathbb{R}_{>0}$, see, e.g., \cite[eqs.~(3-26) and (3-83)]{Basor_2019}.
Then, by the definition of \( u_N(z;w) \) in (\ref{continuationofuN}) and by analytic continuation, 
it holds for \( z \in \mathbb{H}_+ \setminus E_{N,w} \). Substituting the asymptotic expansions \eqref{Cauchydata1new}--\eqref{Cauchydata3new} into (\ref{equationbeforescaling}) and letting $z \to 0$, we obtain the values of $n_{0},h_{0},f_{0}$ for $\textnormal{Re}(w)>-1/2$ and $w+\overline{w}\neq 0$, as given in
\eqref{initialconditionnew}. When $w=0$, 
\begin{equation*}
	\mathbb{E}_{N}^{(w)}\bigg[\mathrm{e}^{-\im\frac{\xi}{2}\sum_{j=1}^{N}x_{j}^{(N)}}\bigg]=\mathrm{e}^{-\frac{N}{2}|\xi|}
\end{equation*}
for all $N\in \mathbb{N}$ by the proof of \cite[Theorem 1.4]{ABGS}, and so combining with Proposition \ref{Hankelrepresentation}, $\mathcal{J}_{N}(\mu,\nu,z)$ is a constant  for $z\in \mathbb{R}$. Thus $u_N(z;w)\equiv0$.
\end{proof}
\begin{remark}
We emphasize that the three expansions \eqref{Cauchydata1new}--\eqref{Cauchydata3new} of $u_{N}(z;w)$ near $z=0$ are not only asymptotic expansions, but actually uniformly and absolutely convergent Puiseux, log-power or power series representations of $u_{N}(z;w)$ in a vicinity of $z=0$ in $\mathbb{H}_+$. This observation follows directly from \eqref{i14}, using analyticity of \( z \mapsto U(a,b,z) \) for \( \textnormal{Re}(z) > 0 \).
\end{remark}

As a direct consequence of Proposition~\ref{detailedanalysisofthesolution}, \eqref{i13} and \eqref{e28}, we obtain the following asymptotic expansion of the exponential in \eqref{e28} at the origin.

\begin{lem}\label{detailedanalysisofthesolution2}
Let $N\in\mathbb{N}$ and $w\in\mathbb{C}$ with $\textnormal{Re}(2w)>-1$. Choose $\mu=\frac{1}{2}(N+w)$, $\nu=\frac{1}{2}(N+\overline{w})$. Let $u_N(z;w)$ be as in \eqref{e28} and put
\bea\label{defofsigma}
\sigma_{N}(z;w)=u_N(z;w)-\frac{N}{2}z.
\eea
Then for $\xi>0$,
\bea\label{e621}
\mathbb{E}_{N}^{(w)}\bigg[\mathrm{e}^{-\im \frac{\xi}{2}\sum_{j=1}^{N}x_{j}^{(N)}}\bigg]=\exp\left[ 
\int_0^{\xi}\sigma_{N}(z;w)\frac{\mathrm{d}z}{z}
\right]
\eea
where the path of integration lies in $\mathbb{H}_+ \setminus E_{N,w}$ with $E_{N,w}$ given as in \eqref{zeroparameterw} and we assume that the left-hand side in \eqref{e621} does not vanish 
for the same parameter $(w,N,\xi)$. Moreover, $\xi\mapsto\textnormal{RHS in}\,\eqref{e621}$ has the following behaviour as $\xi\rightarrow 0$ for $\xi\in\mathbb{H}_+\setminus E_{N,w}$.
%
If $w\in \mathbb{Z}_{\geq 1}$,
\bea\label{Cauchydata1new3}
\exp\left[
\int_0^{\xi} \sigma_{N}(z;w) \frac{\mathrm{d}z}{z}
\right] \sim 1+\sum_{j=1}^{2w} \mathsf{n}_{j-1}(N,w)\xi^{j}+\mathcal{O}(|\xi|^{2w+1})
\eea
If $w+\overline{w} \in \mathbb{Z}_{\geq 0}$ and $w\notin \mathbb{Z}_{\geq 0}$,
\begin{align}\label{Cauchydata2new3}
\exp\left[
\int_0^{\xi} \sigma_{N}(z;w) \frac{\mathrm{d}z}{z}
\right]\sim   1+\sum_{j=1}^{w+\overline{w}}&\mathsf{h}_{j-1}(N,w)\xi^{j}
+\mathcal{O}(|\xi|^{2\textnormal{Re}(w)+1}|\ln\xi|).
\end{align}
Finally, if $w+\overline{w} \notin \mathbb{Z}_{\geq 0}$ and $\textnormal{Re}(w)\in(m,m+1],m\in\mathbb{Z}_{\geq 0}$ or $\textnormal{Re}(2w)\in(-1,0)$, 
\begin{align}\label{Cauchydata3new3}
\exp\left[
\int_0^{\xi} \sigma_{N}(z;w) \frac{\mathrm{d}z}{z}
\right] \sim 1+ \sum_{j=1}^{2m+1}\mathsf{f}_{j-1}(N,w)&\xi^{j}+\mathcal{O}(|\xi|^{2m+2})+\mathcal{O}(|\xi|^{2\textnormal{Re}(w)+1}),
\end{align}
where we set $m=0$ when $\textnormal{Re}(2w)\in(-1,0)$. Moreover, if $w+\overline{w} \neq 0$, then 
\bea\label{initialconditionnew2}
\mathsf{n}_0(N,w)=\mathsf{f}_0(N,w)=\mathsf{h}_0(N,w)=N\frac{\overline{w}-w}{2(w+\overline{w})}.
\eea
\end{lem}
Using Lemma \ref{detailedanalysisofthesolution2}, we establish the following gap structure in the asymptotic expansion of 
\begin{equation*}
	\xi\mapsto\exp\left[
\int_0^{\xi}\sigma_{N}(z;w)\frac{\mathrm{d}z}{z}
\right]+\exp\left[ 
\int_0^{\xi}\sigma_{N}(z;\overline{w})\frac{\mathrm{d}z}{z}
\right]
\end{equation*}
at $\xi=0$. 
\begin{prop}\label{gap structure}
Let $g_{1}(z;w)$ and $g_{2}(z;w)$ be defined as in Proposition \ref{mainthreorem}. Under the same assumption as in Lemma~\ref{detailedanalysisofthesolution2}, then $g_{1}(z;w)+g_{2}(z;w)$ satisfies the following boundary constraints as $z\rightarrow 0$ for $z\in \mathbb{H}_+ \setminus  (E_{N,w}\cup E_{N,\overline{w}})$.
If $w\in \mathbb{Z}_{\geq 1}$,
\bea\label{Cauchydata1new2}
g_{1}(z;w)+g_{2}(z;w) \sim 2+2\sum_{j=1}^{w} \mathsf{n}_{2j-1}(N,w) z^{2j}+\mathcal{O}(|z|^{2w+1})
\eea
If $w+\overline{w} \in \mathbb{Z}_{\geq 0}$ and $w\notin \mathbb{Z}_{\geq 0}$,
\begin{align}\label{Cauchydata2new2}
g_{1}(z;w)+g_{2}(z;w) \sim   2+2\sum_{\substack{j=1\\j~\mathrm{even}~}}^{w+\overline{w}}&\mathsf{h}_{j-1}(N,w) z^{j}
+\mathcal{O}(|z|^{2\textnormal{Re}(w)+1}|\ln z|).
\end{align}
Finally, if $w+\overline{w} \notin \mathbb{Z}_{\geq 0}$ and $\textnormal{Re}(w)\in(m,m+1],m\in\mathbb{Z}_{\geq 0}$ or $\textnormal{Re}(2w)\in(-1,0)$, 
\begin{align}\label{Cauchydata3new2}
g_{1}(z;w)+g_{2}(z;w) \sim 2+ 2 \sum_{j=1}^{m}\mathsf{f}_{2j-1}(N,w)&z^{2j}+\mathcal{O}(|z|^{2m+2})+\mathcal{O}(|z|^{2\textnormal{Re}(w)+1}),
\end{align}
where we set $m=0$ when $\textnormal{Re}(2w)\in(-1,0)$.
\end{prop}
\begin{proof}
By (\ref{theevaluationofJNat0}), Lemmas \ref{lem1} and \ref{lem2}, $\mathcal{J}_{N}(\mu,\nu,0)\neq 0$, so $\mathcal{J}_{N}(\mu,\nu,z)$ is non-zero in a neighborhood of $\mathbb{C}\ni z=0$.  By \eqref{i13}, the function 
\beas
	z\mapsto\frac{1}{\mathrm{C}_{N,w}}(-1)^{N(N-1)/2}\bigg(\frac{2\pi\mathrm{e}^{-z/2}}{\Gamma(N+\overline{w})}\bigg)^N2^{-N^2-N(w+\overline{w})}\,\mathcal{J}_N(\mu,\nu,z)
\eeas
constitutes an analytic continuation of $t\mapsto\mathbb{E}_{N}^{(w)}[\mathrm{e}^{-\im \frac{t}{2}\sum_{j=1}^{N}x_{j}^{(N)}}]$ from $\mathbb{R}_+$ to $\mathbb{H}_{+}$. 
Then,
by \eqref{e621} and \eqref{Cauchydata1new3}-\eqref{Cauchydata3new3}, we conclude that 
\begin{enumerate}
	\item[(a)] if $w+\overline{w} \notin \mathbb{Z}_{\geq 0}$, and there is an $m\in \mathbb{N}$ such that $\textnormal{Re}(w) \in (m,m+1]$, then for $j=1,\ldots,2m$, 
\bea\label{coefficient1}
\mathsf{f}_{j}(N,w)=\frac{1}{(j+1)!}\lim_{t\rightarrow 0^{+}}\frac{\mathrm{d}^{j+1}}{\mathrm{d}t^{j+1}} \mathbb{E}_{N}^{(w)}\bigg[\mathrm{e}^{-\im \frac{t}{2}\sum_{j=1}^{N}x_{j}^{(N)}}\bigg]
\eea
with $\mathsf{f}_{j}(N,w)$ given as in \eqref{Cauchydata3new3}.

	\item[(b)] If $w+\overline{w} \in \mathbb{Z}_{\geq 2}$ and $w\notin \mathbb{Z}_{\geq 0}$,
then for $j=1,\ldots,w+\overline{w}-1$, 
\bea\label{coefficient2}
\mathsf{h}_{j}(N,w)=\frac{1}{(j+1)!}\lim_{t\rightarrow 0^{+}}\frac{\mathrm{d}^{j+1}}{\mathrm{d}t^{j+1}} \mathbb{E}_{N}^{(w)}\bigg[\mathrm{e}^{-\im \frac{t}{2}\sum_{j=1}^{N}x_{j}^{(N)}}\bigg]
\eea
with $\mathsf{h}_{j}(N,w)$ given as in \eqref{Cauchydata2new3}.

\item[(c)] If $w\in \mathbb{Z}_{\geq 1}$,
then for $j=1,\ldots, 2w-1$, 
\bea\label{coefficient3}
\mathsf{n}_{j}(N,w)=\frac{1}{(j+1)!}\lim_{t\rightarrow 0^{+}}\frac{\mathrm{d}^{j+1}}{\mathrm{d}t^{j+1}} \mathbb{E}_{N}^{(w)}\bigg[\mathrm{e}^{-\im \frac{t}{2}\sum_{j=1}^{N}x_{j}^{(N)}}\bigg]
\eea
with $\mathsf{n}_{j}(N,w)$ given as in \eqref{Cauchydata1new3}.
\end{enumerate}
Observe that for $t<0$, 
\bea\label{symmetry}
\mathbb{E}_{N}^{(w)}\bigg[\mathrm{e}^{-\im \frac{t}{2}\sum_{j=1}^{N}x_{j}^{(N)}}\bigg]=\mathbb{E}_{N}^{(\overline{w})}\bigg[\mathrm{e}^{\im \frac{t}{2}\sum_{j=1}^{N}x_{j}^{(N)}}\bigg],
\eea
and $\frac{\mathrm{d}^{j}}{\mathrm{d}t^{j}}\mathbb{E}_{N}^{(w)}[\mathrm{e}^{-\im \frac{t}{2}\sum_{j=1}^{N}x_{j}^{(N)}}]$ exists at $t=0$ when $j<\textnormal{Re}(2w)+1$. Thus,
\beas
\lim_{t\rightarrow 0^{+}}\frac{\mathrm{d}^{j}}{\mathrm{d}t^{j}} \mathbb{E}_{N}^{(w)}\bigg[\mathrm{e}^{-\im \frac{t}{2}\sum_{j=1}^{N}x_{j}^{(N)}}\bigg]=(-1)^{j}\lim_{t\rightarrow 0^{+}}\frac{\mathrm{d}^{j}}{\mathrm{d}t^{j}} \mathbb{E}_{N}^{(\overline{w})}\bigg[\mathrm{e}^{-\im \frac{t}{2}\sum_{j=1}^{N}x_{j}^{(N)}}\bigg]
\eeas
provided $0\leq j<\textnormal{Re}(2w)+1$. Combining the above with \eqref{coefficient1}--\eqref{coefficient3} and \eqref{initialconditionnew2}, the proof of this proposition is complete.
\end{proof}
 
The gap structure obtained in Proposition \ref{gap structure} allows us to represent $F_{N,\delta}(s,h)$ in terms of a solution of the $\sigma$-Painlev\'e V equation for general parameters $h$. We are now ready to prove Proposition \ref{mainthreorem}. First, by an argument similar to that in \cite[Proposition 2.3]{BW}, we obtain the following result.

\begin{lem}\label{kernalexpresion} Let $s \in \mathbb{R}$ and $\delta \in \mathbb{C}$ with $2s + \textnormal{Re}(2\delta) > -1$. Assume $h \in \mathbb{C}$ satisfies $0< \textnormal{Re}(2h) < 2s + \textnormal{Re}(2\delta) +1$. Then
\begin{align*}
F_{N,\delta}(s,h)&=F_{N,\delta}(s,0)2^{-2h}\lim_{\epsilon\downarrow 0}
\int_{0}^{\infty}K_{2h}^{\epsilon}(\xi)\Big(\mathbb{E}_{N}^{(s+\overline{\delta})}\left[\mathrm{e}^{\im \xi\sum_{j=1}^{N}x_{j}^{(N)}}\right]+\mathbb{E}_{N}^{(s+\delta)}\left[\mathrm{e}^{\im \xi\sum_{j=1}^{N}x_{j}^{(N)}}\right]\Big)\mathrm{d}\xi,
\end{align*}
where
\begin{equation*}\label{i5}
	K_h^{\epsilon}(\xi):=\frac{1}{2\pi}\int_{-\infty}^{\infty}|x|^{h}\mathrm{e}^{-\epsilon|x|-\im\xi x}\mathrm{d} x=\frac{\Gamma(1+h)}{2\pi}\bigg[\frac{1}{(\epsilon+\im\xi)^{h+1}}+\frac{1}{(\epsilon-\im\xi)^{h+1}}\bigg],\ \ \xi\in\mathbb{R}.
\end{equation*}
\end{lem}
\begin{proof}[Proof of Proposition \ref{mainthreorem}]
Let $\sigma_{N}(z;w)$ be defined as in \eqref{defofsigma}. By Proposition \ref{prop4},\eqref{0403eve},\eqref{e28}, and \eqref{defofsigma}, we conclude that $\sigma_{N}(z;w)$ has the form stated in \eqref{0403formula} and solves the $\sigma$-Painlev\'e V equation \eqref{transformpainleve}. On the other hand, for any fixed positive integer \(N\), by Proposition \ref{Hankelrepresentation},  equations (\ref{largezasymptotics}) and (\ref{e621}), the functions \(g_{1}(z;w)\), \(g_{2}(z;w)\), together with all their derivatives of arbitrary order, decay exponentially as \(z\to\infty\). Hence, using this fact, Lemma \ref{kernalexpresion}, and combining Proposition \ref{gap structure} with an argument similar to that used in \cite[Propositions 7.3 and 7.5]{BW}, the  difference being that \(\exp\left(\int_{0}^{z} v(t;s)/t\,\mathrm{d}t\right)\) therein is replaced by \(g_{1}(z;w)+g_{2}(z;w)\) in Proposition \ref{gap structure}, we obtain the representation stated in the current proposition.
\end{proof}

\section{Proof of Proposition \ref{convergencetheorem}}\label{proofofproposition14}

In this section, we shall study the range of $s, h$ for which the limit
\bea\label{convergenceproblem}
\frac{F_{N,\delta}(s,h)}{N^{s^2 + s(\delta + \overline{\delta}) + 2h}}
\eea
exists as $N \to \infty$. We then express this limit, that is, the leading coefficient in the large-$N$ asymptotic of \eqref{convergenceproblem}, in terms of the expectation of the moments of a random variable. The idea of using a probabilistic approach to study the convergence of the quantity \eqref{convergenceproblem} first appeared in  Assiotis, Keating and Warren's work \cite{AKW}, where the authors proved that for $\delta=0$, the limit of \eqref{convergenceproblem}, as $N\rightarrow\infty$, exists when $2s>-1$ and $0 \leq 2h < 2s+1$. In \cite{AGS}, Assiotis, Gunes and Soor studied an analogue of \eqref{convergenceproblem} for the circular Jacobi $\beta$-ensemble (note that \eqref{convergenceproblem} corresponds to $\beta=2$). In particular, the next statement follows directly from \cite[Proposition 3.11]{AGS} for $\beta=2$.
\begin{lem}\label{convergenceat0}
Let $w\in \mathbb{C}$ with $\textnormal{Re}(w)>0$. Let $h\in \mathbb{R}$ with $0<2h<\textnormal{Re}(2w)+1$. Then
\beas
\lim_{N\rightarrow \infty}\E_{N}^{(w)}\left[\left|\frac{\sum_{j=1}^{N}x_{j}^{(N)}}{N}\right|^{2h}\right]=\mathbb{E}\left[\left|\mathsf{X}_{1}(w)\right|^{2h}\right]<\infty.
\eeas
\end{lem}
This follows from the uniform integrability of the sequence of random variables 
\beas
	\left(\left|\frac{1}{N}\sum_{j=1}^{N}x_{j}^{(N)}\right|^r\right)_{N=1}^{\infty},
\eeas
which, in turn, follows from a special case of \cite[Proposition 3.8]{AGS} with $\beta=2$.

\begin{lem}\label{uniform integrability}
Let $w\in \mathbb{C}$ with $\textnormal{Re}(w)>0$ and $0\leq r<\textnormal{Re}(2w)+1$.
Then
\beas
\sup_{N\geq 1}\mathbb{E}_{N}^{(w)}\left[\left|\frac{\sum_{j=1}^{N}x_{j}^{(N)}}{N}\right|^r\right]<\infty.
\eeas
\end{lem}
By Lemma \ref{convergenceat0} and Proposition \ref{arepresentationofasanexpectation}, one can easily show that \eqref{convergenceproblem} converges, as $N\rightarrow\infty$, for $s,h \in \mathbb{R}$, $\delta \in \mathbb{C}$ with
\beas\label{thepreviousrange}
\mathrm{Re}(\delta)>-\tfrac12,\quad s+\mathrm{Re}(\delta)>0,\quad \text{and} \quad 0 \leq h < s+\mathrm{Re}(\delta)+\tfrac12.
\eeas
We now extend the convergence range to $2s+\mathrm{Re}(2\delta)>-1$ and $h \in \mathbb{C}$ with $0< \textnormal{Re}(2h) < 2s+\mathrm{Re}(2\delta)+1$, as stated in Proposition \ref{convergencetheorem}. This comes from a generalisation of Lemma \ref{convergenceat0} to the full range of $h$ and $w$, as described in the following result.
\begin{prop}\label{fullrange}
Let $w\in \mathbb{C}$ with $\textnormal{Re}(2w)>-1$. Let $h\in \mathbb{C}$ with $0<\textnormal{Re}(2h)<\textnormal{Re}(2w)+1$. 
\bea\label{solvedquestion}
\lim_{N\rightarrow \infty}\E_{N}^{(w)}\left[\left|\frac{\sum_{j=1}^{N}x_{j}^{(N)}}{N}\right|^{2h}\right]
=\E\left[\left|\mathsf{X}_{1}(w)\right|^{2h}\right]
\eea
\end{prop}
Since the proof of Proposition \ref{fullrange} is similar to that of \cite[Theorem 1.2]{AKW}, after extending the parameter $w$ from $\mathbb{R}$ to $\mathbb{C}$, we postpone the relevant details to Appendix \ref{full range}.

\begin{proof}[Proof of Proposition \ref{convergencetheorem}]
We combine Propositions \ref{arepresentationofasanexpectation} and \ref{fullrange}, and the limit
\beas
\lim_{N\rightarrow \infty}\frac{F_{N,\delta}(s,0)}{N^{s^2 + s(\delta + \overline{\delta})}}=\mathrm{d}(s,\delta),
\eeas
which follows from Lemma \ref{dsillustration}, where $\mathrm{d}(s,\delta)$ is given in \eqref{defofds}. The proof is complete.
\end{proof}

\section{Proof of Theorem  \ref{maintheoreminnote2}}\label{proofoftheorem16}
We first recall the Hankel determinant  representation in Proposition \ref{Hankelrepresentation}, i.e. for any $w$ with $\textnormal{Re}(2w)>-1$ and any $t\geq 0$ the identity
\begin{align}\label{i140630}
\mathbb{E}_{N}^{(w)}\bigg[\mathrm{e}^{-\im\frac{t}{2N}\sum_{j=1}^{N}x_{j}^{(N)}}\bigg]=
\frac{(-1)^{N(N-1)/2}}{\mathrm{C}_{N,w}}\bigg(\frac{2\pi}{\Gamma(N+\overline{w})}\bigg)^N2^{-N^2-N(w+\overline{w})}&\nonumber\\
\times\,\mathrm{e}^{-\frac{t}{2}}\mathcal{J}_N\bigg(\frac{N+w}{2},\frac{N+\overline{w}}{2},\frac{t}{N}\bigg).&
\end{align}
Next, we extend the left-hand side in \eqref{i140630} as function of $t$ from $\mathbb{R}_+$ to $\mathbb{H}_{+}$. To provide a more explicit expression of this analytic continuation, we first apply the authors' previous work \cite[Lemma 2.6]{BW} on certain transformations of the determinant appearing in (\ref{i140630}). 
\begin{lem}\label{transformation}
Let $N\in\mathbb{N}$ and let $\mu,\nu\in \mathbb{C}$ with $\textnormal{Re}(2\mu)>N-1$. Let $U(a,b,z)$ be the confluent hypergeometric function of the second kind. Then for all $z\in\mathbb{C}$ with $\textnormal{Re}(z)>0$,
\begin{align*}
\det\Big[U(1-2\nu,2-2\nu-2\mu+i+j,z)\Big]_{i,j=0}^{N-1}
=&\left(\prod_{\ell=0}^{N-2}(1-2\nu+\ell)^{N-1-\ell}\right)z^{N(2\mu+2\nu-N)}\left(\prod_{j=0}^{N-1}\frac{1}{\Gamma(2\mu-j)}\right)\\
&\times\,\det\left[\int_{0}^{\infty}\mathrm{e}^{-zx}x^{2\mu-j-1}(1+x)^{2\nu-i-1}dx\right]_{i,j=0}^{N-1}.
\end{align*}
\end{lem}

Applying (\ref{i140630}), Lemma \ref{transformation} and Andr\'{e}ief's identity, we obtain the following integral representation of the left-hand side in \eqref{i140630}.
\begin{prop}\label{integralrepresentation}
Let $N\in\mathbb{N}$ and $w\in \mathbb{C}$ with $\textnormal{Re}(2w)>-1$. Then for $t\geq 0$, using principal branches,
\begin{align}
\E_N^{(w)}\bigg[\mathrm{e}^{-\im \frac{t}{2N}\sum_{j=1}^{N}x_{j}^{(N)}}\bigg]=&\frac{\mathrm{e}^{-\frac{t}{2}}}{N!\prod_{j=1}^{N}\Gamma(j)\Gamma(w+\bar{w}+j)}\nonumber\\
&\times\int_{0}^{\infty}\cdots\int_{0}^{\infty}\prod_{i=1}^{N}\bigg(y_{i}+\frac{t}{N}\bigg)^{\overline{w}}y_{i}^{w}\mathrm{e}^{-y_{i}}
\prod_{1\leq j<k\leq N}|y_{k}-y_{j}|^2\prod_{j=1}^{N}\mathrm{d}y_{j}.
\end{align}
\end{prop}
Now we are ready to state an analytic continuation of the left-hand side in \eqref{i140630}, it comes from a similar argument to that used in \cite[Proposition 2.4]{ABGS}. For completeness, we move its proof details to Appendix \ref{proof of Proposition 6.3}.
\begin{prop}\label{theanalyticfunction}
Given $w\in \mathbb{C}$ with $\textnormal{Re}(2w)>-1$. For $z\in \mathbb{H}_{+}$, define 
\begin{align}\label{definitionofFN}
\mathsf{G}_{N}(z)=&\frac{\mathrm{e}^{-\frac{z}{2}}}{N!\prod_{j=1}^{N}\Gamma(j)\Gamma(w+\overline{w}+j)}\nonumber\\
&\times\int_{0}^{\infty}\cdots\int_{0}^{\infty}\prod_{i=1}^{N}\bigg(y_{i}+\frac{z}{N}\bigg)^{\overline{w}}y_{i}^{w}\mathrm{e}^{-y_{i}}\prod_{1\leq j<k\leq N}(y_{k}-y_{j})^2\prod_{j=1}^{N}\d y_{j}. 
\end{align}   
Here and throughout, the complex power
\[
\left(y+\frac{z}{N}\right)^{\overline{w}}
\]
is defined using the principal branch of the logarithm. This choice is well defined for \(y>0\) and \(z\in\mathbb{H}_{+}\), since \(y+z/N\) lies in the right half-plane. Then $z\mapsto\mathsf{G}_{N}(z)$ is analytic on $\mathbb{H}_{+}$ and there is an analytic function $\mathsf{G}(z)$ on 
$\mathbb{H}_{+}$, such that $\mathsf{G}(0)=1$ and 
\bea\label{restrictionontherealnumbers}
\mathsf{G}(t)=\E\Big[\mathrm{e}^{-\frac{\im t}{2}\mathsf{X}_{1}(w)}\Big],\ \ t>0.
\eea
Moreover, for $z\in \mathbb{H}_{+}$ and $m\in\mathbb{Z}_{\geq 0}$, we have
\bea\label{convergenceofthederivative}
\lim_{N\rightarrow \infty}\frac{\mathrm{d}^{m}}{\mathrm{d}z^{m}}\mathsf{G}_{N}(z)=\frac{\mathrm{d}^{m}}{\mathrm{d}z^{m}}\mathsf{G}(z)
\eea
uniformly on compact subsets of $\mathbb{H}_{+}\ni z$.
\end{prop}
Recall that the discrete zero locus of $\mathcal{J}_{N}\mathcal{J}_{N+1}$ is $E_{N,w}$ as given in \eqref{zeroparameterw}, 
that is
\bea\label{zero set}
E_{N,w}=\left\{z\in \mathbb{H}_{+}: \mathcal{J}_{N}\bigg(\frac{N+w}{2},\frac{N+\overline{w}}{2},z\bigg) \mathcal{J}_{N+1}\left(\frac{N+w}{2},\frac{N+\overline{w}}{2},z\right)=0\right\}.
\eea
Based on the above results, we obtain the following representation of 
$\mathsf{G}_N(z)$ in terms of a solution to the $\sigma$-Painlev\'e \textnormal{V} equation.

\begin{prop}\label{equation for finite N}
Let $N\geq 1$ be an integer. Let $w\in\mathbb{C}$ with $\textnormal{Re}(w)>-1/2$. Let $\mathsf{G}_{N}(z)$ be defined as in (\ref{definitionofFN}).
Let $E_{N,w}$ be given as in (\ref{zeroparameterw}).
 For $z\in \mathbb{H}_{+}\setminus(NE_{N,w})$, we have 
\bea\label{definitionoftauN}
u_{N}\left(\frac{z}{N};w\right)-\frac{z}{2}=z\frac{\mathrm{d}}{\mathrm{d}z}\log \mathsf{G}_{N}(z)
\eea
with $u_{N}(z;w)$ given as in (\ref{e28}). Moreover, define \bea\label{defoftauNw}
\tau_{N,w}(z):=u_{N}\left(\frac{z}{N};w\right)-\frac{z}{2},
\eea
then the same function satisfies the following $\sigma$-Painlev\'e {\rm V} equation,
\bea\label{theequationsatisfiedbytauNt}
&&\left(z\frac{\mathrm{d}^{2}\tau_{N,w}}{\mathrm{d} z^{2}}\right)^{2}=-4z \left(\frac{\mathrm{d}\tau_{N,w}}{\mathrm{d} z}\right)^{3}+\Bigg(4\tau_{N,w} +\Big( (w+\overline{w})^2+z^{2}/N^2\Big) +2(w-\overline{w})z/N\Bigg) \left(\frac{\mathrm{d}\tau_{N,w}}{\mathrm{d} z} \right)^{2}\nonumber \\
&& -\Bigg(2 \Big((w-\overline{w})/N +z/N^2\Big)\tau_{N,w}+ \Big((\overline{w}^{2}-w^{2}-z)  -z (\overline{w}+w )/N\Big)\Bigg)\frac{\mathrm{d}\tau_{N,w}}{\mathrm{d} z}+(\tau_{N,w})^{2}/N^2 \nonumber\\
&& -\left(1 +(w+\overline{w})/N\right) \tau_{N,w}+\frac{1}{4} \left(w-\overline{w}\right)^2.
\eea
\end{prop}
\begin{proof}
By Proposition \ref{Hankelrepresentation} and Proposition \ref{theanalyticfunction}, we have
\beas
\mathsf{G}_{N}(z)=\frac{(-1)^{\frac{N(N-1)}{2} } }{\Gamma(\overline{w}+N)^N}\left(
\prod_{j=1}^N \frac{ \Gamma(w+j) \Gamma(\overline{w}+j) }{\Gamma(j) \Gamma(w+\overline{w}+j)}\right)\mathrm{e}^{-\frac{z}{2}}
\mathcal{J}_{N}\left(\frac{N+w}{2},\frac{N+\overline{w}}{2},\frac{z}{N}\right),
\eeas
and so the zeros of $z\mapsto\mathsf{G}_{N}(z)$ are the same as those of $z\mapsto\mathcal{J}_{N}(\frac{N+w}{2},\frac{N+\overline{w}}{2},\frac{z}{N})$. Thus, 
\begin{equation*}
-\frac{1}{2}+\frac{\d}{\d z}\log \mathcal{J}_{N}\bigg(\frac{N+w}{2},\frac{N+\overline{w}}{2},\frac{z}{N}\bigg)=\frac{\d}{\d z}\log \mathsf{G}_{N}(z),\ \ z\in \mathbb{H}_{+}\setminus (NE_{N,w}).
\end{equation*}
and so \eqref{definitionoftauN} and \eqref{theequationsatisfiedbytauNt} follow from Proposition \ref{detailedanalysisofthesolution}.
\end{proof}

\begin{rem}
By \eqref{defofsigma}, we see that
$
\tau_{N,w}(z)=\sigma_{N}(z/N;w)$ for $z\in\mathbb{H}_+\setminus(NE_{N,w})$.
\end{rem}

At this point we are prepared to prove Theorem \ref{maintheoreminnote2}.
%

\begin{proof}[Proof of Theorem \ref{maintheoreminnote2}]

Let $\mathcal{G}_{w}:= \{z\in \mathbb{H}_{+}: \mathsf{G}(z)=0\}\cup\bigcup_{N=1}^{\infty}(NE_{N,w})$, where $E_{N,w}$ is given in \eqref{zeroparameterw}, 
and $\mathsf{G}(z)$ is the uniform limit of $\mathsf{G}_{N}(z)$ in  Proposition \ref{theanalyticfunction}. We claim that if the set $\mathcal{G}_{w}$ has an accumulation point in $\mathbb{H}_+$, then the point must belong to
$
\{ z \in \mathbb{H}_+ : \mathsf{G}(z) = 0 \}.
$ Indeed, let $\alpha \in \mathbb{H}_+$ be an accumulation point of $\mathcal{G}_{w}$. Note that both $E_{N,w}$ and $\{ z\in\mathbb{H}_+ : \mathsf{G}(z) = 0 \}$ are discrete. Then there is a strictly increasing sequence $\{ n_j \}_{j=1}^\infty$ and points $\alpha_j \in E_{n_j, w}$ such that
$
n_j \alpha_j \to \alpha$ as $j \to \infty
$. Since $\mathsf{G}_{n_j}(n_j \alpha_j) = 0$ and $\mathsf{G}_N(z) \to \mathsf{G}(z)$ uniformly on any bounded closed subset of $\mathbb{H}_+$, we conclude that
$
\mathsf{G}(\alpha) = 0.
$
So $\mathcal{G}_{w}$ has no accumulation point in $\mathbb{H}_{+} \setminus \mathcal{G}_{w}$. 
For $z\in \mathbb{H}_{+}\setminus \mathcal{G}_{w}$, there is a neighborhood of $z$ in $\mathbb{H}_{+}$ disjoint from $\mathcal{G}_{w}$. By $(\ref{definitionoftauN})$ and by Proposition \ref{theanalyticfunction}, $\frac {\ \mathrm{d}^{p}}{\ \mathrm{d}z^{p}}\tau_{N,w}(z)$ converges uniformly to $\frac {\mathrm{\d}^{p}}{\mathrm{\d} z^{p}}\tau^{(w)}(z)$ in this neighbourhood, as $N\rightarrow\infty$, for $p\in\mathbb{Z}_{\geq 0}$, where 
\beas
	\tau_w(z):=z\frac{\mathrm{d}}{\mathrm{d}z}\log \mathsf{G}(z).
\eeas
Consequently we can take the limit $N \to \infty$ in both sides of equation \eqref{theequationsatisfiedbytauNt}. Note that $\tau_{w}(z)$ is analytic on
$
\mathbb{H}_{+}\setminus \{z\in \mathbb{H}_{+}:\mathsf{G}(z)=0\}.
$
Hence, equation \eqref{theequationsatisfiedbytauNt} holds for every
$
z\in \mathbb{H}_{+}\setminus \{z\in \mathbb{H}_{+}:\mathsf{G}(z)=0\}.
$
Let
\begin{equation}\label{defpfmathcalzw}
\mathcal{Z}_{w}
:=
\{z\in \mathbb{H}_{+}:\mathsf{G}(z)=0\}\cap\mathbb{R}_{+}.
\end{equation}
Since $\mathsf{G}(t)=\E[\mathrm{e}^{-\frac{\im t}{2}\mathsf{X}_{1}(w)}]$ by \eqref{restrictionontherealnumbers} for $t>0$, we obtain \eqref{limitequation} for $t\in\mathbb{R}_+\setminus\mathcal{Z}_w$ .
%
 For $t<0$, by the symmetry equation \eqref{symmetry} and a similar argument to that of $t>0$, we  conclude equation \eqref{limitequation} holds for $t\in\mathbb{R}\setminus\{\{0\}\cup\mathcal{Z}_w\cup(-\mathcal{Z}_{\overline{w}})\}$. Also, $0\notin\mathcal{Z}_w\cup(-\mathcal{Z}_{\overline{w}})$ and by \eqref{the convergence for the expectation}, $t=0$ is not an accumulation point of $\mathcal{Z}_{w} \cup (-\mathcal{Z}_{\overline{w}})$. We are now left to establish the boundary condition at $t=0$, see \eqref{boundaryconditionsforthelinitcaseatt=0}. By Lemma \ref{convergenceat0}, we have 
\beas
\lim_{N\rightarrow \infty}\frac{\mathrm{d}^{p}}{\mathrm{d}t^{p}}\E_{N}^{(w)}\bigg[\mathrm{e}^{-\frac{\im t}{2N}\sum_{j=1}^{N}x_{j}^{(N)}}\bigg]\bigg|_{t=0}=\frac{\mathrm{d}^{p}}{\mathrm{d}t^{p}}\E\bigg[\mathrm{e}^{-\frac{\im t}{2}\mathsf{X}_1(w)}\bigg]\bigg|_{t=0}
\eeas
holds for $p=0,1$ whenever $\textnormal{Re}(w)>0$, and $p=0,1,2$ whenever $\textnormal{Re}(2w)>1$. So 
\beas
\lim_{N\rightarrow \infty}\frac{\mathrm{d}^{p}}{\mathrm{d}t^{p}}\tau_{N,w}(t)\bigg|_{t=0}=\frac{\mathrm{d}^{p}}{\mathrm{d}t^{p}}\tau_w(t)\bigg|_{t=0}
\eeas
for $p=0$ whenever $\textnormal{Re}(w)>0$, and for $p=0,1$ whenever $\textnormal{Re}(2w)>1$. By \eqref{Cauchydata1new3}-\eqref{initialconditionnew2} and \eqref{defoftauNw}, we record the boundary conditions 
\begin{equation*}
	\tau_{N,w}(0)=0\ \ \ \ \ \textnormal{and}\ \ \ \   \frac{\mathrm{d}}{\mathrm{d}t}\tau_{N,w}(t)\Big|_{t=0}=\frac{\overline{w}-w}{2(w+\overline{w})}.
\end{equation*}
Hence, letting $N\rightarrow \infty$, we arrive at the conclusion in Theorem \ref{maintheoreminnote2}.
\end{proof}
\begin{prop}\label{maintheorem2innote3}
Let $(N,h)\in\mathbb{N}^2$ and $w\in \mathbb{C}$ with $\textnormal{Re}(2w)>2h-1$. Then
\bea\label{expressionfortheabosolutemomentsofXw}
\E\left[\left|\mathsf{X}_{1}(w)\right|^{2h}\right]=(-1)^{h}2^{2h}\frac{\mathrm{d}^{2h}}{\mathrm{d}t^{2h}}\exp\left[\int_{0}^{t}\tau_w(z)\frac{\mathrm{d}z}{z}\right]\Bigg|_{t=0},
\eea
with $\tau_w(t)$ as in Theorem \ref{maintheoreminnote2} and which solves the $\sigma$-Painlev\'e {\rm III$'$} equation \eqref{limitequation}, and where the integration path lies in a neighbourhood of the origin disjoint from $\mathcal{Z}_{w} \cup(-\mathcal{Z}_{\overline{w}})$, with $\mathcal{Z}_{w}$ as in Theorem \ref{maintheoreminnote2}.
Moreover,
\bea\label{arepresentationofthelimitderivativeintermsoftaylorcoefficientsnew2}
\E\left[\left|\mathsf{X}_{1}(w)\right|^{2h}\right]
=(-1)^{h}2^{2h}\frac{\mathrm{d}^{2h}}{\mathrm{d}t^{2h}}\exp\Big(\sum_{j=0}^{2h-1}\frac{\mathsf{n}_{j}(w,\overline{w})}{j+1}t^{j+1}\Big)\Bigg|_{t=0},
\eea
where $\mathsf{n}_{0}(w,\overline{w}) = \frac{\overline{w}-w}{2(w + \overline{w})}$, $\mathsf{n}_{1}(w,\overline{w}) = \frac{-w\overline{w}}{\left((w + \overline{w})^2 - 1\right)(w + \overline{w})^2}$,
and the remaining coefficients $\mathsf{n}_{2}, \ldots, \mathsf{n}_{2h-1}$ are determined recursively, utilising \eqref{limitequation}.
\end{prop}
\begin{proof}
By the almost sure convergence of the sequence $\left(\frac{1}{N}\sum_{j=1}^{N}x_{j}^{(N)}\right)_{N=1}^{\infty}$ to $\mathsf{X}_{1}(w)$, together with \eqref{the convergence for the expectation} and Lemma \ref{uniform integrability}, we obtain that for any integer $p \geq 0$,
\bea\label{thelimitexistswhentgoestozero}
\lim_{N\rightarrow\infty}\frac{\mathrm{d}^{p}}{\mathrm{d}t^{p}}\mathbb{E}_N^{(w)}\left[\mathrm{e}^{-\mathrm{i} \frac{t}{2N}\sum_{j=1}^{N}x_{j}^{(N)}}\right]
= \frac{\mathrm{d}^{p}}{\mathrm{d}t^{p}}\mathbb{E}\left[\mathrm{e}^{-\mathrm{i} \frac{t}{2}\mathsf{X}_{1}(w)}\right]
\eea
holds for any $t \in \mathbb{R}$, provided that $0 \leq p < \mathrm{Re}(2w) + 1$.  By Theorem \ref{maintheoreminnote2}, we know that $0$ is not an accumulation point of $\mathcal{Z}_{w}$ as defined in (\ref{defpfmathcalzw}). Therefore, there is a neighborhood of $0$ such that for any $t$ in this neighborhood,
\[
\mathbb{E}\left[\mathrm{e}^{-\mathrm{i} \frac{t}{2}\mathsf{X}_{1}(w)}\right]
= \exp\left( \int_{0}^{t} \tau_w(z) \frac{\mathrm{d}z}{z} \right),
\]
where the path of integration lies in $t\in \mathbb{R}\setminus \left(\mathcal{Z}_{w} \cup (-\mathcal{Z}_{\overline{w}})\right)$.
By \eqref{thelimitexistswhentgoestozero}, for integer $h$ with $0 \leq 2h < \mathrm{Re}(2w) + 1$,
\beas
\lim_{t\rightarrow 0}\frac{\mathrm{d}^{2h}}{\mathrm{d}t^{2h}}\mathbb{E}\left[\mathrm{e}^{-\mathrm{i} \frac{t}{2}\mathsf{X}_{1}(w)}\right]=(-1)^{h}\frac{1}{2^{2h}}\E\left[\left(\mathsf{X}_{1}(w)\right)^{2h}\right]
\eeas
Note that $2h$ is an even integer and $\mathsf{X}_{1}(w)$ is a real-valued random variable. Hence, the last right-hand side above can also be written as $\mathbb{E}\left[\left|\mathsf{X}_{1}(w)\right|^{2h}\right]$. So, we have arrived at \eqref{expressionfortheabosolutemomentsofXw}. Regarding the coefficients $\mathsf{n}_{j}(w,\overline{w})$ in \eqref{arepresentationofthelimitderivativeintermsoftaylorcoefficientsnew2}: by \eqref{thelimitexistswhentgoestozero}, we can perform a Taylor expansion of $\frac{\mathrm{d}}{\mathrm{d}t}\ln \mathbb{E}[\mathrm{e}^{-\mathrm{i} \frac{t}{2}\mathsf{X}_{1}(w)}]$ at $t=0$ up to order $2h-1$, that is, $\sum_{j=0}^{2h-1}\mathsf{n}_{j}(w,\bar{w})t^{j}$. Combining this with the definition of $\tau_w(t)$ in \eqref{defofthelimitsolution0701} and the boundary conditions \eqref{boundaryconditionsforthelinitcaseatt=0},
 we obtain \eqref{arepresentationofthelimitderivativeintermsoftaylorcoefficientsnew2}. Lastly, substituting the asymptotic expansions \eqref{Cauchydata1new}--\eqref{Cauchydata3new} at $0$ into \eqref{theequationsatisfiedbytauNt} and passing to the limit as $N\to\infty$, we obtain the stated expressions for $\mathsf{n}_{0}$ and $\mathsf{n}_{1}$. Since $\tau_w(t)$ solves the $\sigma$-Painlev\'e {\rm III$'$} equation \eqref{theequationsatisfiedbytauw}, we can recursively obtain $\mathsf{n}_{j}(w,\overline{w})$ from $\mathsf{n}_{0}, \ldots, \mathsf{n}_{j-1}$ for $j = 2, \ldots, 2h-1$. The proof is complete.
\end{proof}
It is known that $\mathsf{X}_{1}(w)$ is not almost surely zero for real $2w>-1$ in \cite{AKW}. This was done using the representation of $\mathsf{X}_{1}(w)$ in
terms of the principal value sum of points of a determinantal point process as proven in \cite{Q}, while for $w\notin \mathbb{R}$, to the best of our knowledge, there is no analogous representation for $\mathsf{X}_{1}(w)$. Instead of using this approach to prove the almost surely non-zero property of $\mathsf{X}_{1}(w)$ for general $w$, we use the connections established in Theorem \ref{maintheoreminnote2} between $\mathsf{X}_{1}(w)$ and integrable systems theory. 
\begin{proof}[Proof of Corollary \ref{almostsurely}]
Equation \eqref{thesecondmomentofXw} follows from
Proposition \ref{maintheorem2innote3}.
Since $\mathbb{E}\left[(\mathsf{X}_{1}(w))^2\right]>0$, we have in particular that $\mathsf{X}_{1}(w)$ is not almost surely zero.
\end{proof}

\begin{proof}[Proof of Proposition \ref{representationforintegerh}]
Now a simple combination of Propositions \ref{convergencetheorem} and \ref{maintheorem2innote3}.
\end{proof}

\section{Proof of Proposition \ref{representation0407} and Corollary \ref{non-zeroproperty}}\label{proofsofProposition111and112}

\begin{lem}\label{dsillustration}
Let $F_{N,\delta}(s,h)$ be defined as \eqref{object}. Then for $\delta\in \mathbb{C}$ with $\textnormal{Re}(2\delta)>-1$ and for $s\in \mathbb{R}$ with $2s+\textnormal{Re}(2\delta)>-1$,
\beas
\lim_{N\rightarrow \infty}\frac{F_{N,\delta}(s,0)}{N^{s^2+s(\delta+\overline{\delta})}}=\mathrm{d}(s,\delta),
\eeas
with $\mathrm{d}(s,\delta)$ as in \eqref{defofds}.
\end{lem}
\begin{proof}
By definition of $Z_{A}(\theta)$,  $\mu_N^{(\delta)}(A)$ and $\mathrm{C}_{N,\delta}$ (or use \cite[Formula (4.2)]{BNR1} with $\beta=2$), 
\begin{align}
F_{N,\delta}(s,0)&=\int_{\mathbb{U}(N)}\left|Z_{A}(0)\right|^{2s}\mathrm{d}\mu_N^{(\delta)}(A)=2^{2Ns}\mathrm{C}_{N,\bar{\delta}+s}/\mathrm{C}_{N,\bar{\delta}}\label{33026}\\
&=\prod_{j=0}^{N-1} \frac{\Gamma(j+1+\delta)\Gamma(j+1+\overline{\delta})\Gamma(2s+\delta+\overline{\delta}+j+1)}{\Gamma(s+\delta+j+1)\Gamma(s+\overline{\delta}+j+1)\Gamma(\delta+\overline{\delta}+j+1)}\nonumber.
\end{align}
The above holds for $\delta\in \mathbb{C}$ with $\textnormal{Re}(2\delta)>-1$ and for $s\in \mathbb{R}$ with $2s+\textnormal{Re}(2\delta)>-1$.
Then, by \cite[Theorem~5.3 and Lemma~7.1]{DHR} with $\beta=2$,
\beas
\lim_{N\rightarrow \infty}\frac{F_{N,\delta}(s,0)}{N^{s^2+s(\delta+\overline{\delta})}}=\mathrm{d}(s,\delta),
\eeas
where 
\beas
\mathrm{d}(s,\delta)=\exp\left(f(2s+\delta+\overline{\delta})-f(\delta+\overline{\delta})+f(\delta)+f(\overline{\delta})-f(\delta+s)-f(\overline{\delta}+s)\right)
\eeas
with 
\begin{align*}
	f(z)=\frac{z}{2}\ln 2\pi-\log G(1+z),
\end{align*}
in terms of the Barnes G-function $G(z)$. From this \eqref{defofds} follows. 
\end{proof}

\begin{proof}[Proof of Proposition \ref{representation0407}]
We first recall that  $\mathbb{E}_{N}^{(w)}[\cdot]$ denotes the expectation with respect to $\mathfrak{m}^{(w)}_N$ in \eqref{defofmathfrakmN}. Then, by an argument similar to that of \cite[Proposition 3.1]{AGKW}, but with the Haar measure on $\mathbb{U}(N)$ replaced by \eqref{generalized measure}, 
we have the following expressions. With $K_{N,\delta}=2^{2sN}\mathrm{C}_{N,\bar{\delta}+s}/\mathrm{C}_{N,\bar{\delta}}$ and $\mathrm{C}_{N,w}$ defined in \eqref{definitionoftheconstant0407},
\begin{align*}
	\frac{\int_{\mathbb{U}(N)}\prod_{j=1}^{k}\big|Z_{A}^{(n_{j})}(0)\big|^{2h_{j}}\mathrm{d}\mu_N^{(\delta)}(A)}{K_{N,\delta}}=&\,{2^{-\sum_{j=1}^{k}2h_{j}n_{j}}}\mathbb{E}_{N}^{(s+\bar{\delta})}\left[\prod_{j=1}^{k}\left|\,\Theta_{n_{j}}\left(x_{1}^{(N)},\ldots,x_{N}^{(N)}\right)\right|^{2h_{j}}\right], 
\end{align*}
and 
\begin{align*}
	&\frac{\int_{\mathbb{U}(N)}\prod_{j=1}^{k}\big|V_{A}^{(n_{j})}(0)\big|^{2h_{j}}\mathrm{d}\mu_N^{(\delta)}(A)}{K_{N,\delta}}\\
	= &\,2^{-\sum_{j=1}^{k}2h_{j}n_{j}}\mathbb{E}_{N}^{(s+\bar{\delta})}\left[\prod_{j=1}^{k}\left|\sum_{m_{j}=0}^{n_{j}}(-\textnormal{i}N)^{m_{j}}\binom{n_{j}}{m_{j}}
	\Theta_{n_{j}-m_{j}}\left(x_{1}^{(N)},\ldots,x_{N}^{(N)}\right)\right|^{2h_j}\right],
\end{align*}
where the polynomials $\Theta_m$ equal
\begin{align}\label{defoftheta}
	\Theta_{m}(x_{1},x_2\ldots,x_{N})&=\sum_{\ell=0}^{m} a_{m,\ell}\,\mathrm{e}_{\ell}(x_{1},x_2,\ldots,x_{N}),
\end{align}
with coefficients $a_{n,\ell}$ given as follows,
\be\label{defofa}
	a_{n,\ell}=(-1)^{\frac{n+\ell}{2}}\sum_{\substack{\sum_{i=1}^{N}m_{i}=n\\\textnormal{only}~m_{1},m_{2},\ldots,m_{\ell}~\textnormal{are}~\textnormal{odd~integers}}}\binom{n}{m_{1}, \ldots, m_{N}}.
\ee 
By Lemma \ref{dsillustration} and (\ref{33026}), 
\beas
	\lim_{N\rightarrow \infty}\frac{K_{N,\delta}}{N^{s^2+s(\delta+\overline{\delta})}}=\mathrm{d}(s,\delta),
\eeas
with $\mathrm{d}(s,\delta)$ as in \eqref{defofds}. Let $\mathbb{E}^{(s+\delta)}[\cdot]$ denote averages taken with respect to the Hua-Pickrell measure $m^{(s+\delta)}$ on $\mathbf{H}$. So
\begin{equation*}
	\mathbb{E}^{(s+\delta)}\left[\left|\,\det\left(X_{n_j}\right)\right|^{2s}\right]= \mathbb{E}_{n_j}^{(s+\delta)} \left[\prod_{m=1}^{n_j} \left|\lambda_m^{(n_j)}(X_{n_j})\right|^{2s}\right]<\infty.
\end{equation*}
Hence, by applying \cite[Proposition 2.11]{AGKW}, while taking the probability measure $\mu_{\infty}$ therein to be $m^{(s+\delta)}$, and using an argument similar to that of \cite[Theorem 1.4]{AGKW}, we obtain the convergence result as claimed in the current proposition.
\end{proof}

\begin{proof}[Proof of Corollary \ref{non-zeroproperty}]
Let
\beas
f_{\delta}(\theta)=\frac{1}{2\pi}\frac{\Gamma(1+\delta)\Gamma(1+\overline{\delta})}{\Gamma(1+\delta+\overline{\delta})}\big(1-\mathrm{e}^{-\im\theta}\big)^{\delta}\big(1-\mathrm{e}^{\im\theta}\big)^{\overline{\delta}}.
\eeas
By \cite[Proposition 3.4]{B}, $f_{\delta}(\theta)d\theta$ is a probability measure on $[0,2\pi)$ associated with the monic orthonormal polynomials 
\beas
p_{n}^{(\delta)}(\mathrm{e}^{\im \theta})=\mathrm{e}^{\im n\theta}{}_2F_1\big(\delta,-n;-n-\overline{\delta};\mathrm{e}^{-\im \theta}\big),\ \ \ n\in\mathbb{Z}_{\geq 0},\ \ \ \theta\in[0,2\pi),
\eeas
where ${}_2F_1(a,b;c;z)$ is the hypergeometric function. Moreover,
\beas
	\big\|p_{n}^{(\delta)}\big\|^2_{L^2(\d f_{\delta})}=\frac{(\overline{\delta}+\delta+1)^{(n)} n! }{ (\overline{\delta}+1)^{(n)}(\delta+1)^{(n)} }.
\eeas
By \cite[Theorem 1]{AV} (with $  K=L=1$ therein), 
\bea\label{kernelexpression}
	\int_{\mathbb{U}(N)} V_{A}(\theta) V_{A^*}(-\theta) \mathrm{d}\mu_N^{(\delta)}(A) 
	=\big\|p_N^{(\delta)}\big\|^2_{L^2(\d f_{\delta})}
	\sum_{j=0}^N \frac{p_j^{(\delta)}(\mathrm{e}^{\im \tau}) \overline{p_j^{(\delta)}(\mathrm{e}^{\im \theta}) }  }{\big\|p_{j}^{(\delta)}\big\|^2_{L^2(\d f_{\delta})}}.
\eea
Observe that
\[
	\frac{\d^n}{\d\theta^n}V_{A^{\ast}}(\theta) \bigg|_{\theta=0} = 
	\overline{\frac{\mathrm{d}^n}{\d\theta^n} V_{A}(-\theta) \bigg|_{\theta=0}  },\ \ \ \ 
\frac{\mathrm{d}^n}{\d\tau^n} Z_{A^{\ast}}(-\tau ) \bigg|_{\tau=0} = (-1)^N
\overline{ \frac{\mathrm{d}^n}{\d\tau^n} Z_{A}(\tau) \bigg|_{\tau=0}},
\]
and therefore
\beas
	\int_{\mathbb{U}(N)} \big|V_A^{(n)}(0)\big|^2 \mathrm{d}\mu_N^{(\delta)}(A)
=\frac{\mathrm{d}^n}{\mathrm{d}\theta^n} \frac{\mathrm{d}^n}{\mathrm{d}\tau^n}   \int_{\mathbb{U}(N)} V_A(\theta) V_{A^{\ast}}(-\tau) \mathrm{d}\mu_N^{(\delta)}(A)\Big|_{\theta=\tau=0},\eeas
as well as
\bea\label{theexpressionofthederivativeofV}
	\int_{\mathbb{U}(N)} \big|Z_A^{(n)}(0)\big|^2 \mathrm{d}\mu_N^{(\delta)}(A)
	=\frac{\mathrm{d}^n}{\mathrm{d}\theta^n} \frac{\mathrm{d}^n}{\mathrm{d}\tau^n}  \mathrm{e}^{\frac{\im\theta}{2}N} \mathrm{e}^{-\frac{\im\tau}{2}N} \int_{\mathbb{U}(N)} V_A(\theta) V_{A^{\ast}}(-\tau) \mathrm{d}\mu_N^{(\delta)}(A) \bigg|_{\theta=\tau=0}. 
\eea
Additionally, by the properties of ${}_2F_1(a,b;c;z)$, for $n\in\mathbb{N}$ sufficiently large,
\bea\label{thederivativeofhyper}
	\frac{\mathrm{d}^{m}}{\d z^m}{}_2F_1(a,-n;c-n;z)\bigg|_{z=1}
= \frac{(a)^{(m)} (-n)^{(m)}}{(c-n)^{(m)}}  {}_2F_1(a+m, -n+m; c-n+m;1),
\eea
where, by the Chu-Vandermonde identity, 
\bea\label{derivativeat1}
{}_2F_1(a+m, -n+m; c-n+m;1)=\frac{(1-c+a+m)^{(n-m)}}{(1-c)^{(n-m)}}.
\eea
Consequently, by identities \eqref{theexpressionofthederivativeofV},\eqref{kernelexpression} and \eqref{thederivativeofhyper} and \eqref{derivativeat1}, 
\begin{align}
  \lim_{N\rightarrow \infty} &\,\frac{1}{N^{\delta+\overline{\delta}+1+2n}}
	\int_{\mathbb{U}(N)} \big|Z_A^{(n)}(0)\big|^2 \mathrm{d}\mu_N^{(\delta)}(A) \nonumber\\
=&\,\,\Gamma(\delta+1) \Gamma(\overline{\delta}+1) 
\sum_{i=0}^n\sum_{j=0}^n \binom{n}{i} \binom{n}{j} \bigg(\frac{1}{2} \bigg)^{2n-i-j}
\sum_{\ell_1=0}^i \sum_{\ell_2=0}^j (-1)^{i-\ell_1+j-\ell_2} \nonumber\\
& \times \, \frac{\delta^{(\ell_{1})} \bar{\delta}^{(\ell_{2})}}{\Gamma(\delta+\overline{\delta}+1+\ell_1) \Gamma(\delta+\overline{\delta}+1+\ell_2)}
\frac{1}{i+j+\delta+\overline{\delta}+1}.\label{relatedtotheexpectation} 
\end{align}
By Proposition \ref{representation0407}, 
\begin{align}\label{expected}
	\lim_{N\rightarrow \infty} \frac{1}{N^{\delta+\overline{\delta}+1+2n}}
	\int_{\mathbb{U}(N)} \big|Z_A^{(n)}(0)\big|^2 \mathrm{d}\mu_N^{(\delta)}(A) =\frac{1}{4^{n}}\mathrm{d}(1,\delta)\mathbb{E}\left[\big|\mathsf{R}_{n}(\bar{\delta}+1)\big|^2\right],
\end{align}
in terms of
\beas
	\mathrm{d}(1,\delta)=\lim_{N\rightarrow \infty} \frac{1}{N^{\delta+\overline{\delta}+1}}\int_{\mathbb{U}(N)} |Z_A(0)|^2 \mathrm{d}\mu_N^{(\delta)}(A);
\eeas
and by \eqref{defofds},
\bea\label{c1w}
	\mathrm{d}(1,\delta)=\frac{\Gamma(\delta+1)\Gamma{(\bar{\delta}+1)}}{\Gamma(\delta+\overline{\delta}+2)\Gamma(\delta+\overline{\delta}+1)}.
\eea
Hence, abbreviating $w =\delta + 1$ in the above and combining \eqref{relatedtotheexpectation} with \eqref{expected} and \eqref{c1w}, we obtain the first expression claimed in Corollary \ref{non-zeroproperty}. The second expression follows from a similar argument, so we are left with the positivity claim. To that end we note that the matrix
\begin{equation*}
	\mathbf{M}=\left[\frac{1}{1+i+j+\delta+\overline{\delta}}\right]_{i,j=0}^n
\end{equation*}
is Hermitian and positive definite due to properties of a Cauchy matrix. Since the right-hand side of \eqref{relatedtotheexpectation} equals $|\Gamma(\delta+1)|^2\,\mathbf{x}^T \mathbf{M}\,\overline{\mathbf{x}}$ with $\mathbf{x}=(x_{0},x_1,\ldots,x_{n})^{T}$ where
\beas
x_{i}=\binom{n}{i}\bigg(\frac{1}{2} \bigg)^{n-i}
\sum_{\ell=0}^i(-1)^{i-\ell} \frac{\delta^{(\ell)}} {\Gamma(\delta+\overline{\delta}+1+\ell) },
\eeas
the proof is now complete.
\end{proof}
\section{Connection to a $\sigma$-Painlev\'e~{\rm III$'$} equation}\label{sectiononjointmofasequenceofrandomvariables}
In this section, we shall prove Theorem \ref{painlevethm0328}. As mentioned earlier, for $w \in \mathbb{R}$, there is a description of $\mathsf{X}_{1}(w)$, see \eqref{describtionofX1w}, which expresses it as a certain principal value sum over the points of a determinantal point process. However, it is unknown whether an analogous result holds for general $w \in \mathbb{C}$. Consequently, it seems difficult to study the statistics related to $\mathsf{X}_{1}(w)$ using this approach through such a determinantal point process representation, especially the structure stated in Theorem \ref{painlevethm0328}, which represents all joint moments of the characteristic function of $\mathsf{X}_{1}(w)$ and $\mathsf{X}_{n}(w)$, $n=2,\ldots$, in terms of a solution of the $\sigma$-Painlev\'e {\rm III$'$} equation.\smallskip

Instead, we approximate the joint moments of $\mathsf{X}_1(w),\ldots,\mathsf{X}_{n}(w)$ by an appropriate finite-$N$ expression with a similar structure to the one stated in Theorem~\ref{painlevethm0328}. Our approach is similar to that of \cite[Theorem~1.9]{AGKW}, with the main difference being that all quantities in their proof depend only on a single parameter $w$ (which is $s$ in \cite{AGKW}), whereas in this paper, the analogous quantities depend on both $\textnormal{Re}(w)$ and $w$. When $w \in \mathbb{R}$, our setting reduces to the situation considered in \cite{AGKW}.
Due to the introduction of this extra parameter, we need to re-compute the corresponding expressions. In the following, we outline the proof steps and point out the differences compared to the analogous steps in \cite{AGKW}.

Recall that $\mathbb{E}_N^{(w)}\left[\cdot\right]$ is the average with respect to the finite dimensional Hua-Pickrell measure $\mathfrak{m}_N^{(w)}$ on the Weyl chamber
\[
\mathbb{W}^N
=
\left\{
(x_1,\ldots,x_N)\in\mathbb{R}^N:
x_1\geq\ldots\geq x_N\right\}.
\]
\begin{prop}\label{structureforgeneralfinitesize}
Let $k\in\mathbb{Z}_{\geq 2}$ and $n_{2},\ldots,n_{k}$ be non-negative integers with $\sum_{j=2}^{k} n_j >0$. Let $w\in \mathbb{C}$ with $\textnormal{Re}(w)>2^{-1}(\sum_{q=2}^k qn_q-1)$. Then, for all $t_{1}\geq 0$,
\begin{align}\label{general structure1}
    &\frac{(-2\textnormal{i})^{-\sum_{q=2}^{k}qn_{q}}}{N^{n_{2}+\cdots+n_{k}}}\mathbb{E}_{N}^{(w)}\Bigg[\mathrm{e}^{-\textnormal{i} t_{1}\frac{1}{N}\sum_{j=1}^{N}x^{(N)}_{j}}\prod_{q=2}^{k}\left(\sum_{j=1}^{N}\left(x^{(N)}_{j}-\textnormal{i}\right)^q\right)^{n_{q}}\Bigg]\nonumber \\=&\frac{1}{t_{1}^{\sum_{q=2}^{k}qn_{q}-1}}
    \sum_{m=0}^{\sum_{q=2}^{k}(q-1)n_{q}}t_{1}^{m-1}P_{m}^{(N,w,\overline{w})}(t_{1})\frac{\mathrm{d}^{m}}{\mathrm{d}t_{1}^{m}}\mathbb{E}_{N}^{(w)}\Bigg[\mathrm{e}^{-\textnormal{i} t_{1}\frac{1}{N}\sum_{j=1}^{N}x^{(N)}_{j}}\Bigg],
\end{align}
where $ t_1^{m-1} P_{m}^{(N,w,\overline{w})}(t_{1})$ are polynomials in $t_{1}$ of degree at most $\sum_{q=2}^{k}qn_{q}-1$.  Moreover, the coefficients of these polynomials are polynomials in $(N,w,\overline{w})$, and with degrees in terms of $N$ or $w$ or $\overline{w}$ not exceeding $\sum_{q=2}^{k}(q-1)n_{q}$. 
\end{prop}

In order to avoid confusion, we clarify that for $w\in \mathbb{C}$, the expectation  $\mathbb{E}_{N}^{(w)}$ used in \cite[Theorem 1.9]{AGKW} refers to taking averages with respect to 
\beas
	\frac{1}{\widetilde{\mathrm{C}}_{N,w}} \prod_{j=1}^N \left(1+x_j^2\right)^{-w-N}\prod_{1\leq j<k\leq N}|x_{k}-x_{j}|^2\prod_{\ell=1}^{N}\mathrm{d}x_{\ell},
\eeas while in our paper,  $\mathbb{E}_{N}^{(w)}$ denotes the average with respect to the probability measure
\beas
	\frac{1}{\mathrm{C}_{N,w}} \prod_{j=1}^N \left(1+\im x_j\right)^{-w-N}\left(1-\im x_j\right)^{-\overline{w}-N} \prod_{1\leq j<k\leq N}|x_{k}-x_{j}|^2\prod_{\ell=1}^{N}\mathrm{d}x_{\ell}
\eeas
with $\mathrm{C}_{N,w}$ given as in \eqref{definitionoftheconstant0407}. The main difference between Proposition \ref{structureforgeneralfinitesize} and  \cite[Theorem 1.9]{AGKW} is that we need to handle two parameters, $\textnormal{Re}(w)$ and $\overline{w}$, rather than a single parameter in \cite[Theorem 1.9]{AGKW}. In fact, our Proposition \ref{structureforgeneralfinitesize} is a direct consequence (after setting $v=\overline{w}$) of the following result about averages in the non-symmetric Cauchy ensemble.

\begin{prop}\label{mainproposition1}
Let $N\in\mathbb{N}$ and let $v\in \mathbb{C}$ with  $\textnormal{Re}(v)>-1$. Let $k\in\mathbb{Z}_{\geq 2}$ and $n_{2},\ldots,n_{k}\geq 0$ be integers. Suppose $w\in \mathbb{C}$ satisfies 
\bea\label{restrictiononparameters2h0318}
	\textnormal{Re}(w)>-1,  \ \ \ \ \textnormal{Re}(v+w)>\sum_{q=2}^{k}qn_{q}-1.
\eea
Abbreviate
\bea\label{definitionofflnopowerterms}
	\mathcal{F}(N; v,w)(t_{1})
	:=\int_{\mathbb{R}^N} \mathrm{e}^{-\im \frac{t_{1}}{N}\sum_{j=1}^{N}x_{j}}\left( \prod_{j=1}^N (1+\im x_j)^{-N-w} (1-\im x_j)^{-N-v}\right)\big(\Delta(\mathbf{x})\big)^2\mathrm{d}\mathbf{x}.
\eea
and 
\begin{align}\label{definitionoffl}
\mathcal{F}_{n_{2},\ldots,n_{k}}&\,(N; v,w)(t_{1})\nonumber\\
:=&\int_{\mathbb{R}^N} \mathrm{e}^{-\im \frac{t_{1}}{N}\sum_{j=1}^{N}x_{j}}\prod_{q=2}^{k}\Big(\sum_{j=1}^{N}(x_{j}-\im)^q\Big)^{n_{q}}\left( \prod_{j=1}^N (1+\im x_j)^{-N-w} (1-\im x_j)^{-N-v}\right)\big(\Delta(\mathbf{x})\big)^2\mathrm{d}\mathbf{x}.
\end{align}
Here, we use the shorthands $\Delta(\mathbf{x}):=\prod_{1\leq j<k\leq N}|x_{k}-x_{j}|^2$ and $\mathrm{d}\mathbf{x}:=\prod_{\ell=1}^{N}\mathrm{d}x_{\ell}$.
Then, for all $t_{1}\geq 0$,
\begin{align}\label{mainstructure1}
\frac{\mathcal{F}_{n_{2},\ldots,n_{k}}(N; v,w)(t_{1})}{N^{n_{2}+\cdots+n_{k}}}=\frac{(-2\textnormal{i})^{\sum_{q=2}^{k}qn_{q}}}{t_{1}^{\sum_{q=2}^{k}qn_{q}-1}}
    \sum_{m=0}^{\sum_{q=2}^{k}(q-1)n_{q}}t_{1}^{m-1}P_{m}^{(N,v,w)}(t_{1})\frac{\mathrm{d}^{m}}{\mathrm{d}t_{1}^{m}}\mathcal{F}(N; v,w)(t_{1}),
\end{align}
where $ t_1^{m-1} P_{m}^{(N,v,w)}(t_{1})$ are polynomials in $t_{1}$ of degree at most $\sum_{q=2}^{k}qn_{q}-1$.  Moreover, the coefficients of these polynomials are polynomials in $N,v,w$, with degrees in $N$ or $v$ or $w$ not exceeding $\sum_{q=2}^{k}(q-1)n_{q}$. Additionally, when $w=v=s\in \mathbb{R}$, $P_{m}^{(N,v,w)}(t_{1})$ coincides with $P_{m}^{(s,N)}(t_{1})$, as stated in \cite[Theorem 1.9]{AGKW}.
\end{prop}

\subsection{Proof of Theorem \ref{painlevethm0328} }
\begin{proof}[Proof of Theorem \ref{painlevethm0328} (assuming Proposition \ref{mainproposition1}) ]
Let $\left(x_1^{(N)},\dots,x_N^{(N)}\right)\in \mathbb{W}_N$ denote the eigenvalues of the top-left $N\times N$ corner of the infinite random matrix whose distribution is given by $m^{(w)}$, the Hua-Pickrell measure on $\mathbf{H}$. We first use a convergence result from \cite[Proposition 2.11]{AGKW} about the normalised elementary symmetric polynomial 
\begin{equation*}
	\frac{1}{N^{k}}\mathrm{e}_k\big(x_1^{(N)},x_2^{(N)},\ldots,x_N^{(N)}\big),\ \ \ k\in\mathbb{N},
\end{equation*}
which guarantees its almost sure convergence and its convergence in $L^{p}(m^{(w)})$ for any $p$ with  $1\leq p\leq \sum_{m=2}^{k}mn_{m}$. Secondly, we recall $\Omega$ in \eqref{definitionofOmega} and the probability measure $\mathcal{M}^{(w)}$ on $\Omega$. Next let $\mathbf{H}_{0}$ denote the subset of $\mathbf{H}$ consisting of all infinite matrices $X$ such that
\[
	\frac{1}{N}\Big(\max\big\{x_{i}^{(N)},0\big\}\Big)_{i=1}^\infty,\ \ \ \ \ \ 
	\frac{1}{N}\Big(\max\big\{-x_{N+1-i}^{(N)},0\big\}\Big)_{i=1}^\infty,
\]
\[
	\frac{1}{N}\sum_{i=1}^N x_i^{(N)},\ \ \ \ \ \ 
	\frac{1}{N^2}\sum_{i=1}^N \big(x_i^{(N)}\big)^2,
\]
converge as $N\to\infty$ almost surely. Correspondingly, we denote the limits by
$\big(\alpha_i^+\big)_{i=1}^\infty$, $\big(\alpha_i^-\big)_{i=1}^\infty$, $\gamma_{1}$ and $\delta$, respectively.
By \cite[Theorem 5.1]{Borodin_2001}, we have
$
m^{(w)}(\mathbf{H}_{0})=1.
$ By \cite[Theorem 1.1]{Q}, we have 
$
\delta=\sum_{i=1}^{\infty}(\alpha_{i}^{+})^2+\sum_{i=1}^{\infty}(\alpha_{i}^{-})^2
$
almost surely under $m^{(w)}$. Furthermore, by virtue of Proposition 2.3 and its proof in \cite{AN}, for $m\geq 3$,
\begin{equation*}
\frac{\sum_{j=1}^N \left(x^{(N)}_j\right)^m}{N^m}
\xrightarrow[]{N\to \infty} 
\sum_{i=1}^\infty \Big[\left(\alpha_i^{+}\right)^m+\left(-\alpha_{i}^{-}\right)^m\Big],
\end{equation*}
almost surely.
Note that there is a natural Borel map from $\mathbf{H}_{0}$ to $\Omega$, which maps $X$ to
\[
\left(\big(\alpha_i^+\big)_{i=1}^\infty, \big(\alpha_i^-\big)_{i=1}^\infty, \gamma_{1}, \delta\right).
\]
By \cite[Theorem 5.2]{Borodin_2001}, the pushforward of the measure $m^{(w)}|_{\mathbf{H}_{0}}$ under this Borel map is $\mathcal{M}^{(w)}$. 
Then, combining this with \eqref{defOfmathcalC} and the definition of
$\mathsf{X}_m(w)$ given in \eqref{defofxn}, we obtain
\begin{equation*}
\frac{\sum_{j=1}^N \left(x^{(N)}_j\right)^m}{N^m}
\xrightarrow[]{N\to \infty} \mathsf{X}_m(w),
\end{equation*}
in distribution for all $m=2,\ldots,k$, and thus
\begin{equation*}
\frac{\sum_{j=1}^N \left(x^{(N)}_j-\im\right)^m}{N^m}
\xrightarrow[]{N\to \infty} 
\mathsf{X}_m(w),
\end{equation*}
in distribution.
To prove
\begin{align}\label{eq:convergenceforpainleve0328}
       \frac{1}{N^{\sum_{m=2}^k mn_m}}\mathbb{E}_N^{(w)}\left[\mathrm{e}^{-\textnormal{i}\frac{t_1}{N}\sum_{j=1}^N x^{(N)}_j} \prod_{m=2}^k\left(\sum_{j=1}^N (x^{(N)}_j-\textnormal{i})^m\right)^{n_m}\right] &\xrightarrow[]{N\to \infty} \mathbb{E}\left[\mathrm{e}^{-\textnormal{i}t_1 \mathsf{X}_1(w)} 
        \prod_{m=2}^k \left(\mathsf{X}_m(w)\right)^{n_m}
        \right], 
\end{align}
it suffices to prove that the sequence  
\beas
	\left(\frac{1}{N^{\sum_{m=2}^k mn_m}}\prod_{m=2}^k\left|\sum_{j=1}^N \left(x^{(N)}_j-\textnormal{i}\right)^m\right|^{n_m}\right)_{N=1}^{\infty}
\eeas
is uniformly integrable and then it suffices to prove that
\beas
	\left(\prod_{\ell=1}^{k}\left(\frac{1}{N^{\ell}}\left|\mathrm{e}_{\ell}\big(x_1^{(N)},x_2^{(N)},\ldots,x_N^{(N)}\big)\right|\right)^{r_{\ell}}\right)_{N=1}^{\infty}
\eeas
is uniformly integrable for $1\leq \sum_{\ell=1}^{k}\ell r_{\ell}= \sum_{m=2}^{k}mn_{m}$. This follows from \cite[Proposition 2.11]{AGKW}. Note that the sequence of random variables $\{\frac{1}{N} {\rm Tr}(X_N)\}_{N\geq 1}$ converges almost surely to $\mathsf{X}_{1}(w)$. Then, using Lemma \ref{uniform integrability}, for any $m$ with $m<\textnormal{Re}(2w)+1$, 
\bea\label{0328afternoon1}
\frac{\mathrm{d}^m}{\mathrm{d}t_1^m} \mathbb{E}_N^{(w)}\left[\mathrm{e}^{-\textnormal{i}t_1\sum_{j=1}^N \frac{x^{(N)}_j}{N}}\right] &\xrightarrow[]{N\to \infty} \displaystyle\frac{\mathrm{d}^m}{\mathrm{d}t_{1}^m} \mathbb{E}\left[\mathrm{e}^{-\textnormal{i}t_1\mathsf{X}_1(w)}\right]
\eea
holds for $t_{1}\in \mathbb{R}$. Now by limits \eqref{eq:convergenceforpainleve0328} and \eqref{0328afternoon1}, and by Proposition \ref{structureforgeneralfinitesize}, we have that \eqref{formula10328} holds with the polynomials $\mathcal{A}_m^{(w,\overline{w})}(t_{1})$ given by,
    \begin{equation*}
        \mathcal{A}_m^{(w,\overline{w})}(t_{1})=\lim_{N\to \infty}\frac{P_{m}^{(N,w,\overline{w})}(t_{1})}{N^{\sum_{q=2}^k (q-1)n_q}}.
    \end{equation*}
This establishes the structure claimed in Theorem \ref{painlevethm0328}, except for the claim that the degrees of the polynomials $t_{1}^{m-1} \mathcal{A}_m^{(w,\overline{w})}(t_{1})$ in $t_1$ does not exceed $\sum_{q=2}^k (q-1)n_q-1$. To prove this, we use the following process. We begin by expanding \( \mathcal{S}^{(n_2,\ldots,n_k)}_{N,\ell,q}(v, w; t_1) \), defined in equation \eqref{definitionofpsi2} below, as a power series in \( N \):
\[
	\mathcal{S}^{(n_2,\ldots,n_k)}_{N,\ell,q}(v, w; t_1) = N^{\ell + \sum_{m=2}^{k}(m-1)n_m}\mathcal{R}^{(n_2,\ldots,n_k)}_{N,\ell,q}(v, w; t_1) + \mathcal{O}\left(N^{\ell + \sum_{m=2}^{k}(m-1)n_m - 1}\right),
\]
for all \( 1 \leq q \leq\ell \). A similar expansion applies to \( \mathcal{S}^{(n_2,\ldots,n_k)}_N(v, w; t_1) \), which is defined in equation \eqref{definitionofpsi1} below. We then examine the leading-order behaviour in \( N \) of the recursive relations satisfied by \( \mathcal{S}^{(n_2,\ldots,n_k)}_{N,\ell,q} \) and \( \mathcal{S}^{(n_2,\ldots,n_k)}_N \). These relations are similar to  those  in \cite[Propositions 4.13 and 4.15]{AGKW}. By comparing the highest order terms in \( N \), we obtain corresponding recursive formulas for the coefficients \( \mathcal{R}^{(n_2,\ldots,n_k)}_{N,\ell,q}(v, w; t_1) \). Comparing these with the recursive formulas for
\(\mathcal{S}^{(n_2,\ldots,n_k)}_{N,\ell,q}(v,w;t_1)\), we note that
some terms vanish because they contribute only to subleading orders in \(N\). Then, applying the inductive process described in the proof of \cite[Theorem 1.9]{AGKW} shows that the polynomials in $t_1$ contained in the terms $\mathcal{R}_{N,\ell,q}^{(n_2,\ldots,n_k)}(v,w;t_{1})$ have the desired degrees as stated in Theorem \ref{painlevethm0328}.  Finally, we let $v=\overline{w}$ and this completes the proof of the theorem.
\end{proof}
In the remaining part of this section, we prove Proposition \ref{mainproposition1}. We begin by writing $\mathcal{F}_{n_{2},\ldots,n_{k}}(N; v, w)(t_1)$ as a Hankel determinant. The approach follows a similar argument to that of \cite[Proposition 4.2]{AGKW}, with the key difference being that the determinant of the Hankel matrix in that case involves entries corresponding to the multivariate generating function of the Fourier transform of a certain function against the symmetric Cauchy-type weight $1/(1+x^2)^{N+w}$. In contrast, this paper replaces this weight with the non-symmetric  
$
1/((1 + \im x)^{N+w}(1 - \im x)^{N + \overline{w}}),
$
see, e.g., \cite[p.~398, formula~(9)]{GR} for its Fourier transform.
We omit the details of the proof here.

\begin{prop}\label{proprepresented as a Hankel0321}
Let $N\in\mathbb{N}$ and $v \in \mathbb{Z}_{\geq 0}$.
Let $k\geq 2$ and $n_{2},\ldots,n_{k}\geq 0$ be integers. Suppose $w\in \mathbb{C}$ satisfies \eqref{restrictiononparameters2h0318}. Let $U(a,b;z)$ be the confluent hypergeometric function of the second kind, and define
\bea\label{definition of am0318}
	 \xi_{m}(v,w;t):=U(1-N-v, 2-2N-v-w+m; 2t)\e^{-t}.
\eea 
Let $\mathcal{F}_{n_{2},\ldots,n_{k}}(N; v,w)$ be defined as in \eqref{definitionoffl}. Then, we have for $t_{1}\geq 0$,
\begin{align}
	(-2\im)^{-\sum_{q=2}^k q n_{q}}&\,\mathcal{F}_{n_{2},\ldots,n_{k}}(N; v,w)(Nt_{1})\\
=&\,\,
\mathrm{c}_{N} 
 \prod_{q=2}^k \frac{\partial^{n_q}}{\partial t_q^{n_{q}}}
 \det_{0\leq i,j\leq N-1}\left[\sum_{m_2,\ldots,m_k=0}^\infty \frac{t_2^{m_2}\cdots t_k^{m_k}}{m_2 !\cdots m_k!} \xi_{i+j+\sum_{q=2}^k q m_q}(v,w;t_1)
\right]\Bigg|_{t_2=\ldots=t_k=0}\nonumber,
\end{align}
where
\bea\label{definitioncnsnew0318}
\mathrm{c}_{N}=(-1)^{\frac{N}{2}(N-1)}N!\frac{2^{-N(N+w+v)}(2\pi)^{N}}{\Gamma(N+v)^{N}}.
\eea
\end{prop}
By symmetry, for $t_{1}<0$, we can provide an expression analogous to the one in Proposition \ref{proprepresented as a Hankel0321}.
\begin{prop}
Subject to the assumption in Proposition \ref{proprepresented as a Hankel0321}, let
\begin{align*}
	\mathcal{G}_{n_{2},\ldots,n_{k}}&\,(N; v,w)(t_{1})\nonumber\\
=&\,\,\int_{\mathbb{R}^N} \mathrm{e}^{-\im \frac{t_{1}}{N}\sum_{j=1}^{N}x_{j}}\prod_{q=2}^{k}\Big(\sum_{j=1}^{N}(x_{j}+\im)^q\Big)^{n_{q}}\left( \prod_{j=1}^N (1+\im x_j)^{-N-w} (1-\im x_j)^{-N-v}\right)\big(\Delta(\mathbf{x})\big)^2\mathrm{d}\mathbf{x}.
\end{align*}
Then, for $t_{1}<0$, with $\mathrm{c}_N$ as in \eqref{definitioncnsnew0318},
\begin{align*}
	(2\im)^{-\sum_{q=2}^k q n_{q}}&\,\mathcal{G}_{n_{2},\ldots,n_{k}}(N; v,w)(Nt_{1})\\
=&\,\,
\mathrm{c}_{N}
 \prod_{q=2}^k \frac{\partial^{n_q}}{\partial t_q^{n_{q}}}
 \det_{0\leq i,j\leq N-1}\left[\sum_{m_2,\ldots,m_k=0}^\infty \frac{t_2^{m_2}\cdots t_k^{m_k}}{m_2 !\cdots m_k!}
 \xi_{i+j+\sum_{q=2}^k q m_q}(w,v;-t_1)
\right]\Bigg|_{t_2=\ldots=t_k=0}.\nonumber
\end{align*}
\end{prop}

Although Proposition \ref{proprepresented as a Hankel0321} imposes a restriction on $v\in \mathbb{Z}_{\geq 0}$, it is sufficient for our purposes via an analytic continuation argument based on Carlson's lemma, following a similar argument to that in \cite[Lemma 4.1]{AGKW}.
\begin{lem}\label{reductiontoints}
Assume the assertion of Proposition \ref{mainproposition1} holds for all $v\in \mathbb{Z}_{\geq 0}$, $w\in \mathbb{C}$ with $w$ as in \eqref{restrictiononparameters2h0318}. Then, the same statement holds for all $v\in \mathbb{C}$ with $\textnormal{Re}(v)>-1$, and $w\in \mathbb{C}$ as in \eqref{restrictiononparameters2h0318}.
\end{lem}

Based on Lemma \ref{reductiontoints}, from now on, we will restrict our attention to $v\in \mathbb{Z}_{\geq 0}$ in proving Proposition \ref{mainproposition1}.

\subsection{Setup for the proof of Proposition \ref{mainproposition1}}
Since we focus on the large $N$-behaviour of  
\beas
\det_{0\leq i,j\leq N-1}\left[ \sum_{m_2,\ldots,m_k=0}^\infty \frac{t_2^{m_2}\ldots t_k^{m_k}}{m_2 !\ldots m_k!} \xi_{i+j+\sum_{q=2}^k q m_q}(v,w;t_1)
\right]\Bigg|_{t_2=\cdots=t_k=0},
\eeas
we first apply the variable changes  $t_{1}\rightarrow t_{1}/N$, $\ldots$, $t_{k}\rightarrow t_{k}/N$ and introduce the shorthand
\bea\label{definitionofphi0324}
    \boldsymbol{\xi}_{n}(v,w; t_1,\ldots,t_k):=\sum_{m_2,\ldots,m_k=0}^\infty \frac{(t_2/N)^{m_2}\cdots (t_k/N)^{m_k}}{m_2 !\cdots m_k!} \xi_{n+\sum_{l=2}^k l m_l}(v,w;t_{1}/N).
\eea
Accordingly we also define
\[
\mathbf{\Xi}_{N}(v,w;t_1,\ldots,t_k):= \det_{0\leq i,j \leq N-1} \Big[\boldsymbol{\xi}_{i+j}(v,w; t_1,\ldots,t_k) \Big].
\]
which is well-defined by Proposition \ref{proprepresented as a Hankel0321}. Now, consider the following shifted and weighted versions of the Hankel matrix $\mathbf{\Xi}_{N}(v,w; t_1,\ldots,t_k)$: namely, for each integer partition $\boldsymbol{\lambda}=(\lambda_1\geq \cdots \geq \lambda_p > 0)$, we let
\beas
\mathbf{\Xi}_{N,\boldsymbol{\lambda}}(v,w;t_1,\ldots,t_{k}):=\det_{0\leq i,j\leq N-1}\left[\boldsymbol{\xi}_{i+j+\lambda_{N-j}}(v,w;t_1,\ldots,t_{k})\right]
\eeas
where if $p\leq N-1$, we let 
$\lambda_{p+1}=\cdots=\lambda_N=0.$ When $p>N$, we define $\mathbf{\Xi}_{N,\boldsymbol{\lambda}}(v,w;t_1,\ldots,t_{k})=0$.
Note that in the special case, $\boldsymbol{\lambda}=\emptyset = (0,0,\ldots,0)$, one has $\mathbf{\Xi}_{N,\emptyset}=\mathbf{\Xi}_N$.
Related to this, we also define the weighted version:
\[
\mathbf{\Xi}_{N,\boldsymbol{\lambda}}^{(\omega)}(v,w; t_1,\ldots,t_{k}) =
\det_{0\leq i,j\leq N-1}\left[(i+j+\lambda_{N-j})\boldsymbol{\xi}_{i+j+\lambda_{N-j}}(v,w; t_1,\ldots,t_{k})\right].
\]
Finally, we define the main objects that will allow us to obtain recurrences for derivatives of $\mathbf{\Xi}_{N,\boldsymbol{\lambda}}$.

\begin{defn}\label{shiftoperator0320}
Let $m\in\mathbb{Z}_{\geq 0}$. For each integer partition $\boldsymbol{\lambda}=(\lambda_1,\ldots, \lambda_p)$ with $p\leq N$, we define the functions $\mathbf{\Xi}_{N,\boldsymbol{\lambda},m}(v,w;t_1,\ldots,t_{k})$ by
\beas
\mathrm{Tr}\left[\mathrm{adj}\left(\left[\boldsymbol{\xi}_{i+j+\lambda_{N-j}}(v,w; t_1,\ldots,t_{k})\right]_{i,j=0}^{N-1}\right)\left[\boldsymbol{\xi}_{i+j+\lambda_{N-j}+m}(v,w; t_1,\ldots,t_{k})\right]_{i,j=0}^{N-1}\right], 
\eeas
and $\mathbf{\Xi}_{N,\boldsymbol{\lambda},m}^{(\omega)}(v,w;t_1,\ldots,t_{k})$ by
\begin{equation*}
\mathrm{Tr}\left[\mathrm{adj}\left(\left[\boldsymbol{\xi}_{i+j+\lambda_{N-j}}(v,w; t_1,\ldots,t_{k})\right]_{i,j=0}^{N-1}\right)\left[(i+j+\lambda_{N-j}+m)\boldsymbol{\xi}_{i+j+\lambda_{N-j}+m}(v,w; t_1,\ldots,t_{k})\right]_{i,j=0}^{N-1} \right],
\end{equation*}
where $\mathrm{adj}$ denotes the adjugate of a given matrix, i.e. the transpose of its cofactor matrix.
\end{defn}
The following proposition from the known properties of the confluent hypergeometric function, cf. \cite[p.257, formula (5) and p.258, formula (10)]{EMOT}.
\begin{prop}\label{basic recursive formula}
Let $\alpha \in \mathbb{C}$. Let $\xi_\alpha(v,w;t)$ be defined as in \eqref{definition of am0318}. Then
\[
\frac{\mathrm{d}}{\d t} \xi_\alpha(v,w;t) = \xi_\alpha(v,w;t) - 2\xi_{\alpha+1}(v,w;t)
\]
\[
\xi_{\alpha+2}(v,w;t)= \frac{2-2N-v-w+\alpha+2t }{2t} \xi_{\alpha+1}(v,w;t)
- \frac{1-N-w+\alpha}{2t } \xi_{\alpha}(v,w;t).
\]
\end{prop}
From the last Proposition we deduce the following two recurrence relations for $\boldsymbol{\xi}_{n}(v,w; t_{1},\ldots,t_{k})$.
\begin{lem}
For each integer $n\geq 0$, 
\bea\label{partialderivativewithrespecttot10321}
N\frac{\partial \boldsymbol{\xi}_{n}}{\partial t_1}=\boldsymbol{\xi}_{n}-2\boldsymbol{\xi}_{n+1},
\eea
and
\bea\label{indexaddedby20321}
\boldsymbol{\xi}_{n+2}&=&\frac{N(N+w-n-1)}{2t_1+2t_2}\boldsymbol{\xi}_{n}+\left(\frac{N(2-2N-v-w+n)}{2(t_1+t_2)}+\frac{t_{1}}{t_{1}+t_{2}}\right)  \boldsymbol{\xi}_{n+1}\nonumber\\
&&+ \frac{1}{2(t_{1}+t_{2})}\sum_{q=3}^{k}((q-1)t_{q-1}-qt_{q})\boldsymbol{\xi}_{n+q}+\frac{kt_k}{2(t_1+t_2)}  \boldsymbol{\xi}_{n+k+1}.
\eea
\end{lem}
We note that although the definition of  $\mathbf{\Xi}_{N}(v,w; t_1,\ldots,t_k)$ differs from that of $\mathbf{\Psi}_{N}( t_1,\ldots,t_k)$ as given in \cite[Section 4.3]{AGKW} (they coincide when $v=w$), the following Propositions \ref{aboutshiftedby10324} and \ref{representationshiftedby20324} are identical to \cite[Lemma 4.9]{AGKW} and 
\cite[Lemma 4.12]{AGKW}, respectively.

\begin{prop}\label{aboutshiftedby10324}
For each $q=2,\ldots,k$ and for any integer partition $\boldsymbol{\lambda}$,
\bea\label{partialderivativewithrespecttot1}
\mathbf{\Xi}_{N,\boldsymbol{\lambda},1} = - \frac{N}{2} \frac{\partial \mathbf{\Xi}_{N,\boldsymbol{\lambda}}}{\partial t_1} + \frac{N}{2} \mathbf{\Xi}_{N,\boldsymbol{\lambda}},
\eea
and
\bea\label{parialderivativewithrespecttotm}
\frac{\partial \mathbf{\Xi}_{N,\boldsymbol{\lambda}}}{\partial t_q}=
\frac{1}{N}\mathbf{\Xi}_{N,\boldsymbol{\lambda},q}.
\eea
\end{prop}

\begin{proof}
This follows from Definition \ref{shiftoperator0320} and from relation \eqref{partialderivativewithrespecttot10321}.
\end{proof}

The next proposition represents $\Xi_{N,\boldsymbol{\lambda},2}$ as in terms of partial derivatives of $\Xi_{N,\boldsymbol{\lambda}}$.
\begin{prop}\label{representationshiftedby20324}
Let $\boldsymbol{\lambda}_{(2,1)}=(2)$, $\boldsymbol{\lambda}_{2,2}=(1,1)$ and $\boldsymbol{\lambda}_{1,1}=(1)$. Then 

\beas
\mathbf{\Xi}_{N,\boldsymbol{\lambda}_{2,1}} &=& \frac{N^2}{8} \frac{\partial^2 \mathbf{\Xi}_N}{\partial t_1^2} - \frac{N^2}{4} \frac{\partial \mathbf{\Xi}_N}{\partial t_1} + \frac{N^2}{8} \mathbf{\Xi}_N + \frac{N}{2} \frac{\partial \mathbf{\Xi}_N}{\partial t_2}, \\
\mathbf{\Xi}_{N,\boldsymbol{\lambda}_{2,2}} &=& \frac{N^2}{8} \frac{\partial^2 \mathbf{\Xi}_N}{\partial t_1^2} - \frac{N^2}{4} \frac{\partial \mathbf{\Xi}_N}{\partial t_1} + \frac{N^2}{8} \mathbf{\Xi}_N - \frac{N}{2} \frac{\partial\mathbf{\Xi}_N}{\partial t_2} .
\eeas
\end{prop}
\begin{proof}
By Definition \ref{shiftoperator0320}, it is not hard to check that
\beas
\begin{pmatrix}
\mathbf{\Xi}_{N,\boldsymbol{\lambda}_{2,1}} \\
\mathbf{\Xi}_{N,\boldsymbol{\lambda}_{2,2}}
\end{pmatrix} &=& 
\begin{pmatrix}
1/2 & 1/2 \\
1/2 & -1/2
\end{pmatrix}
\begin{pmatrix}
\mathbf{\Xi}_{N,\boldsymbol{\lambda}_{1,1},1} \\
\mathbf{\Xi}_{N,\emptyset,2}
\end{pmatrix}.
\eeas
The claim then follows from Proposition \ref{aboutshiftedby10324}.
\end{proof}
By formula \eqref{indexaddedby20321} and the linearity of the trace, we obtain the following relation for $\mathbf{\Xi}_{N,\boldsymbol{\lambda},m}(v,w;t_1,\ldots,t_{k})$, which reduces to  \cite[Lemma 4.8]{AGKW} when $v=w$.
\begin{prop}\label{generalrecursiveforshiftbypartition0325}
Let $(m,k)\in\mathbb{Z}_{\geq 2}^2$ and let $\boldsymbol{\lambda}$ be a partition. Then
\begin{align*}
\mathbf{\Xi}_{N,\boldsymbol{\lambda},m}&\,(v,w;t_1,\ldots,t_{k})\\
=&\,\,\frac{N(N+w-1)}{2(t_1+t_2)} \mathbf{\Xi}_{N,\boldsymbol{\lambda},m-2}(v,w;t_1,\ldots,t_{k})
-\frac{N}{2(t_1+t_2)} \mathbf{\Xi}_{N,\boldsymbol{\lambda},m-2}^{(\omega)}(v,w;t_1,\ldots,t_{k}) \\
&+ \frac{N(1-2N-v-w)}{2(t_1+t_{2})}
\mathbf{\Xi}_{N,\boldsymbol{\lambda},m-1}(v,w;t_1,\ldots,t_{k})+\frac{N}{2(t_1+t_2)} \mathbf{\Xi}_{N,\boldsymbol{\lambda},m-1}^{(\omega)}(v,w;t_1,\ldots,t_{k})\\
& + \frac{t_{1}}{t_1+t_{2}}
\mathbf{\Xi}_{N,\boldsymbol{\lambda},m-1}(v,w;t_1,\ldots,t_{k})
+  \frac{kt_{k}}{2(t_1+t_{2})}
\mathbf{\Xi}_{N,\boldsymbol{\lambda},m+k-1}(v,w;t_1,\ldots,t_{k}) \\
& +\frac{1}{2(t_{1}+t_{2})}\sum_{l=3}^{k}\left((l-1)t_{l-1}-lt_{l}\right)\mathbf{\Xi}_{N,\boldsymbol{\lambda},m+l-2}(v,w;t_1,\ldots,t_{k}).
\end{align*}
\end{prop}

\subsection{Proof of Proposition \ref{mainproposition1}}
We observe that, following the setup in the preceding subsection, the proof of Proposition \ref{mainproposition1} is similar to that of \cite[Theorem 1.9]{AGKW}, with the main difference being that in \cite{AGKW} one has $v=w$, whereas for us $v$ and $w$ are two independent parameters. So we defer the proof to Appendix \ref{proofofmainproposition1}.

\subsection{Proof of Corollary \ref{painlevecor0324}}

We begin by providing a few examples of Proposition
\ref{mainproposition1} as follows, which lead to Corollary \ref{painlevecor0324}.

\begin{lem}\label{formulaforf1f20324}
Let $N\in\mathbb{N}$ and $v\in \mathbb{C}$ with  $\textnormal{Re}(v)>-1$. Then for $w\in \mathbb{C}$ with $\textnormal{Re}(v+w)>1$,
\bea
&&\mathcal{F}_{1}(N;v,w)(t_{1})\label{expressionforf1t10324}
\\
&=&
\left(-\frac{(v+w)N^2}{t_{1}}+2N\right)\frac{\mathrm{d}}{\d t_1}\mathcal{F}(N;v,w)(t_{1})+\left(\frac{(v-w)N^2}{t_{1}}-2N\right)\mathcal{F}(N;v,w)(t_{1})\nonumber.
\eea
and for $w\in \mathbb{C}$ with $\textnormal{Re}(v+w)>3$,
\begin{align}
\mathcal{F}_{2}&\,(N;v,w)(t_{1})\label{expressionforf2t10324} \\
&= \frac{4N^4}{t_{1}^2} \left(  (-\frac{v}{2}-\frac{w}{2}+\frac{t_{1}}{N})^2 + 1/2\right) \frac{\mathrm{d}^2}{\d t_1^2}\mathcal{F}(N;v,w)(t_{1}) \nonumber\\
&+ \frac{2N^4}{t_{1}^3} \left( -4\frac{t_{1}^3}{N^2}+4v\frac{t_{1}^2}{N}-4 \frac{t_{1}}{N}\left(-\frac{w^2}{4}N+\frac{v^2}{4}N-\frac{w}{4}-\frac{3v}{4}\right)-\frac{3}{2}(v+w)^2\right) 
	\frac{\mathrm{d}}{\d t_1} \mathcal{F}(N;v,w)(t_{1})\nonumber\\
&+ \frac{4N^4}{t_{1}^3}
\left( 
\frac{(v-w)^2t-2t+3v^2-3w^2}{4}
+\frac{-\left((v-w)t+3v/2-w/2\right)tN+t^3}{N^2}\right) \mathcal{F}(N;v,w)(t_{1})\nonumber.
\end{align}

\end{lem}
\begin{proof}
By Proposition \ref{proprepresented as a Hankel0321}, we have 
\beas
\mathcal{F}_{n_{2}}(N;v,w)(t_{1})
=\mathrm{c}_{N}(-4)^{n_{2}}N^{n_{2}}\frac{\partial^{n_{2}}}{\partial t_{2}^{n_{2}}}\mathbf{\Xi}_{N}(v,w;t_{1},t_{2})\Big|_{t_{2}=0}.
\eeas
Then, using Proposition~\ref{generalrecursiveforshiftbypartition0325} with $k=2$, $m=2$, and $\boldsymbol{\lambda}=\emptyset$, and evaluating at $t_{2}=0$, together with Proposition~\ref{aboutshiftedby10324}, we obtain equation \eqref{expressionforf1t10324}. For equation \eqref{expressionforf2t10324}, we first apply Proposition~\ref{generalrecursiveforshiftbypartition0325} with $k=2$, $m=2$, and $\boldsymbol{\lambda}=\emptyset$, which produces the term $\mathbf{\Xi}_{N,\emptyset,3}$. We then use Proposition~\ref{generalrecursiveforshiftbypartition0325} again, this time with $k=2$, $m=3$, to express this term. Next, we differentiate both sides of the resulting equation with respect to $t_{2}$ and evaluate at $t_{2}=0$. Finally, combining this with Propositions~\ref{aboutshiftedby10324} and~\ref{representationshiftedby20324} , we obtain equation~(\ref{expressionforf2t10324}).
\end{proof}
\begin{proof}[Proof of Corollary \ref{painlevecor0324}]
This follows from Lemma \ref{formulaforf1f20324} by setting $v=\overline{w}$, dividing \eqref{expressionforf1t10324} by $N^2$ and \eqref{expressionforf2t10324} by $N^4$ respectively,  and taking the limit as $N\rightarrow\infty$.
\end{proof}

\addcontentsline{toc}{section}{Appendix}

\appendix 
\section{Proof of Proposition \ref{fullrange}}\label{full range}
\begin{proof}
Firstly, we prove that $(\ref{solvedquestion})$ holds for $w,h\in \mathbb{C}$ with $\textnormal{Re}(w)>0$ and $0<\textnormal{Re}(2h)<\textnormal{Re}(2w)+1$. Let 
\begin{equation*}
	\mathsf{Y}_{N}:=\left|\frac{\sum_{j=1}^{N}x_{j}^{(N)}}{N}\right|^{2h}\ \ \ \ \  \textnormal{and}\ \ \ \ \ \  \mathsf{Y}:=\left|\mathsf{X}_{1}(w)\right|^{2h}.
\end{equation*}
Since $(\sum_{j=1}^{N}x_{j}^{(N)}/N)_{N=1}^{\infty}$ converges almost surely to $\mathsf{X}_{1}(w)$, $\mathsf{Y}_{N}$ converges almost surely to $\mathsf{Y}$. For a fixed $h$, choose a suitable $\delta>0$ such that $\textnormal{Re}(2h)(1+\delta)<\textnormal{Re}(2w)+1$, then by Lemma \ref{convergenceat0}, we have 
\begin{equation*}
	\sup_{N\geq 1}\E_{N}^{(w)}\Big[\left|\mathsf{Y}_{N}\right|^{1+\delta}\Big]<\infty,
\end{equation*}
i.e. the sequence $(\mathsf{Y}_{N})_{N=1}^{\infty}$ is uniformly integrable, and so $\E[\mathsf{Y}_{N}]$ converges to $\E[\mathsf{Y}]$ as $N\rightarrow \infty$, that is, \eqref{solvedquestion} holds.\smallskip

Secondly,  we show that \eqref{solvedquestion} holds for $w\in \mathbb{C}$, $h\in \mathbb{R}$ with $-1<\textnormal{Re}(2w)\leq 0$ and $0<h<\textnormal{Re}(w)+1/2$. By \cite[Lemma 2.12]{AKW}, it suffices to prove that
\bea\label{reduction1}
\mathbb{E}\left[\left(\sum_{i=1}^{\infty} (\alpha^+_i)^2 + \sum_{i=1}^{\infty} (\alpha^-_i)^2+\gamma_{2}\right)^{h}\right]<\infty,
\eea
where $X \mapsto \left(\{\alpha_{i}^{+}\}_{i=1}^{\infty},\{\alpha_{i}^{-}\}_{i=1}^{\infty},\mathsf{X}_{1}(w),\gamma_{2}\right)$ is a map from $\mathbf{H}$ to $\Omega$ with $\Omega$ in \eqref{definitionofOmega}. Here we recall that $\gamma_2=\delta-\sum(\alpha_{i}^{+})^2-\sum(\alpha_{i}^{-})^2$. Note that the pushforward measure of $m^{(w)}$ under this map is the spectral measure on $\Omega$ by \cite[Theorem 5.3]{Borodin_2001}, and the correlation function of the corresponding point process $\mathcal{C}^{(w)}$ exists by \cite[Theorem 6.1]{Borodin_2001}. Let $\mathrm{K}^{(w)}(x,y)$ denote the underlying correlation kernel. Then
\beas
\mathrm{K}^{(w)} (x',x'')
&=& \frac{1}{2\pi} \frac{\Gamma(w+1) \Gamma(\overline{w}+1)}{\Gamma(\textnormal{Re}(2w)+1) \Gamma(\textnormal{Re}(2w)+2)}
\frac{\widetilde{P}(x')Q(x'')-Q(x')\widetilde{P}(x'') }{x'-x''},
\eeas
where
\begin{align*}
	\widetilde{P}(x) =& \left|2/x \right|^{\textnormal{Re}(w)}
\mathrm{e}^{-\im/x+\pi \textnormal{Im}(w) \sgn(x)/2} {}_1F_1(w,\textnormal{Re}(2w)+1 ,2\im/x),\\
	Q(x)=&\,(2/x)\left|2/x \right|^{\textnormal{Re}(w)}
\mathrm{e}^{-\im/x+\pi \textnormal{Im}(w) \sgn(x)/2} {}_1F_1(w+1,\textnormal{Re}(2w)+2 ,2\im/x)
\end{align*}
with ${}_1F_1(a,b,z)$  the confluent hypergeometric function of the first kind.
Moreover by \cite[Theorem 2.1]{Q}, we have $\gamma_{2}=0$ for $m^{(w)}$-a.e. $X\in \mathbf{H}$. So by $(\ref{reduction1})$ and the above analysis, combining with a similar argument to that of \cite[formula (19)]{AKW}, it suffices to show that for $0<R<\infty$,
\beas\label{reduction2}
\int_{\{|x|\geq R\}}x^{2h}\mathrm{K}^{(w)} (x,x) dx<\infty,\ \ \ \ \ \ \ \int_{\{|x|< R\}} x^{2}\mathrm{K}^{(w)} (x,x) dx<\infty.
\eeas
Note that $\mathrm{K}^{(w)}(x',x'')$ remains invariant when $x',x'',w$ are replaced by $-x',-x'',\overline{w}$ (there is one more change of sign in the denominator $x'-x''$).
So without loss of the generality, we only show that
\bea\label{reduction 4}
\int_{R}^{\infty}x^{2h}\mathrm{K}^{(w)} (x,x) dx<\infty,\ \ \ \ \ \ \ \ \ \int_{0}^{R}x^{2}\mathrm{K}^{(w)} (x,x) dx<\infty.
\eea
But these follow simply from the definition of $\mathrm{K}^{(w)}(x',x'')$ and from the asymptotic formul\ae,
\beas
{}_1F_1(a,b,z)\sim 1+\frac{a}{b}z,\ \ \ \  z\rightarrow 0,
\eeas
and
\beas
{}_1F_1(a,b,z)\sim \Gamma(b)\left(\frac{\mathrm{e}^{z}z^{a-b}}{\Gamma(a)}+\frac{(-z)^{-a}}{\Gamma(b-a)}\right),\ \ \ \  |z|\rightarrow \infty\ \ \ \text{with}\ \ \ -\frac{3\pi}{2}<\arg z \leq \frac{\pi}{2},
\eeas
when combined with the change of variables $x\rightarrow 1/x$ in the integrand of \eqref{reduction 4}, we have obtained the desired. By \eqref{reduction1}, this completes the proof of the current case.\smallskip

Finally, we show that \eqref{solvedquestion} holds also for $w\in \mathbb{C}$, $h\in \mathbb{C}$ with $-1<\textnormal{Re}(2w)\leq 0$ and $0<\textnormal{Re}(2h)<\textnormal{Re}(2w)+1$. To that end recall Scheff\'e's theorem: if a sequence of the random variables $(\mathsf{Z}_{N})_{N=1}^{\infty}$ with finite expectations converges almost surely to another random variable $\mathsf{Z}$ with a finite expectation, and $\lim_{N\rightarrow \infty}\mathsf{E}[|\mathsf{Z}_{N}|]=\mathsf{E}[|\mathsf{Z}|]$, where $\mathsf{E}[\cdot]$ is the underlying expectation, then
\begin{equation*}
	 \lim_{N\rightarrow \infty}\mathsf{E}[\mathsf{Z}_{N}]=\mathsf{E}[\mathsf{Z}].
\end{equation*}
In our case we put $\mathsf{Z}_{N}=\left|\sum_{j=1}^{N}x_{j}^{(N)}/N\right|^{2h}$ and $\mathsf{Z}=\left|\mathsf{X}_{1}(w)\right|^{2h}$. Then $\mathsf{Z}_{N}$ converges almost surely to $\mathsf{Z}$. Moreover,
\beas
\E[|\mathsf{Z}_{N}|]=\E\left[\left|\frac{\sum_{j=1}^{N}x_{j}^{(N)}}{N}\right|^{2\textnormal{Re} (h)}\right]=\E_{N}^{(w)}\left[\left|\frac{\sum_{j=1}^{N}x_{j}^{(N)}}{N}\right|^{2\textnormal{Re} (h)}\right]<\infty,
\eeas
and $\E[|\mathsf{Z}|]=\E[|\mathsf{X}_{1}(w)|^{2\textnormal{Re}(h)}]<\infty$, as well as $\lim_{N\rightarrow \infty}\E[|\mathsf{Z}_{N}|]=\E[|\mathsf{Z}|]$. Consequently, by 
Scheff\'e's theorem, $$\lim_{N\rightarrow \infty}\E[\mathsf{Z}_{N}]=\E[\mathsf{Z}].$$
That is $(\ref{solvedquestion})$ holds. We have now completed the proof of Proposition \ref{fullrange}.
\end{proof}

\section{Proof of Proposition \ref{theanalyticfunction}}\label{proof of Proposition 6.3}
\label{Proof of a proposition}
\begin{proof}
Firstly, it is easy to see that for all $t\in (0,\infty)$, $\mathsf{G}_N(t) =\E_N^{(w)}[\mathrm{e}^{-\im \frac{t}{2N}\sum_{j=1}^{N}x_{j}^{(N)}}]$ by Proposition \ref{integralrepresentation} and $z\mapsto\mathsf{G}_N(z)$ is holomorphic on $\mathbb{H}_{+}$ by Morera's theorem. We now show that $\{\mathsf{G}_N(z)\}_{N=1}^{\infty}$ is uniformly bounded on any compact subsets of $\mathbb{H}_{+}$. If $K$ is such a compact subset, then there is a constant $M_K$ such that for $y\geq 0$,
\[
	\Big|\,y+\frac{z}{N}\Big| \leq y+\frac{M_K}{N}.
\]
Note that the pre-factor to the integral of right-hand side of \eqref{definitionofFN} equals
$
\frac{1}{N!\prod_{j=1}^{N}\Gamma(j)\Gamma(\textnormal{Re}(2w)+j)}>0$. Therefore, in the case $\textnormal{Re}(w)>0$, writing $\mathrm{e}^{w \log(y+\frac{z}{N})} = \mathrm{e}^{w \log|y+\frac{z}{N}| + \im w \arg (y+\frac{z}{N}) }$, it follows that for fixed $w$,
\[
|\mathsf{G}_N(z)| \leq \mathrm{e}^{\frac{\pi}{2}|\textnormal{Im}(w)|}\bigg|\mathrm{e}^{-\frac{z}{2}+\frac{1}{2}M_K}\bigg|\cdot \bigg|\E_N^{(\textnormal{Re}(w))} \bigg[\mathrm{e}^{-\im \frac{M_{k}}{2N}\sum_{j=1}^{N}x_{j}^{(N)}}\bigg]\bigg|\leq \mathrm{e}^{\frac{\pi}{2}|\textnormal{Im}(w)|}\mathrm{e}^{\frac{1}{2}M_K}.
\]
In the case $0\geq \textnormal{Re}(2w) >-1$, we note that for fixed $w$,
\[
	\bigg|\Big(y_j+\frac{z}{N}\Big)^w\bigg| \leq y_j^{\textnormal{Re}(w)}\mathrm{e}^{\frac{\pi}{2}|\textnormal{Im}(w)|}, \quad \textnormal{Re}(z) >0.
\]
So
\[
|\mathsf{G}_N(z)|\leq \mathrm{e}^{\frac{\pi}{2}|\textnormal{Im}(w)|} \Big|\E_N^{(\textnormal{Re}(w))}[1]\Big| =\mathrm{e}^{\frac{\pi}{2}|\textnormal{Im}(w)|}.
\]
By Montel's theorem, the family $\{\mathsf{G}_N(z)\}_{N=1}^{\infty}$ is normal and so every subsequence has a sub-subsequence converging uniformly on compact subsets of $\mathbb{H}_{+}$, and all these limits are holomorphic on $\mathbb{H}_+$. All these limits are analytic and take the same value on $(0,\infty)$ by (\ref{the convergence for the expectation}). Thus the sequence $(\mathsf{G}_{N}(z))_{N=1}^{\infty}$ converges uniformly on compact subsets to $\mathsf{G}(z)$, an analytic function on $\mathbb{H}_{+}$. Moreover, applying properties of uniform limits of analytic functions, any order of derivatives of $\mathsf{G}_{N}(z)$ also converges uniformly on compacts to the corresponding derivative of $\mathsf{G}(z)$.
\end{proof}

\section{Detailed proof of Proposition \ref{mainproposition1}}\label{proofofmainproposition1}
In this section, we provided the detailed proof of Proposition \ref{mainproposition1}. Before that, we need some intermediate results. We first introduce the following notation.
\begin{defn}
For $(n,m)\in\mathbb{N}^2$ with $1\leq m\leq n$, define the partition $\boldsymbol{\lambda}_{n,m}$ of $n$ to be the partition  $(n-m+1,1,\ldots,1)$ where the number of $1$'s that appear at the end is $m-1$. Moreover, for integer $q\geq 0$, define the column vector
\begin{equation*}
    \Xi_N^{(q)}[n;m]:=\left(\Xi_{N,\boldsymbol{\lambda}_{n,m}}, \Xi_{N,\boldsymbol{\lambda}_{n,m+1}},\ldots, \Xi_{N,\boldsymbol{\lambda}_{n,n}},0,\ldots,0 \right)^T,
\end{equation*}
where $q$ is the number of $0$'s that appear at the end. 
\end{defn}

The next lemma can be derived using an argument analogous to that of \cite[Theorem 17]{keating-fei}.
\begin{lem}\label{generalresult0325}
Let $k,m$ be integers with $k\geq 2$ and $0\leq m\leq k+1$.  Let $\ell\in\mathbb{Z}_{\geq 3}$ and denote the $\ell$-vectors $\mathcal{U}_n$, $n=1,2$ by 
\begin{equation*}
	(\mathcal{U}_n)_j:= 
\begin{cases}
        0,  & j=1, \\
             \noalign{\vskip4pt}
         \displaystyle\sum_{h=2}^{j} (-1)^h \mathbf{\Xi}_{N,\boldsymbol{\lambda}_{\ell-h,j+1-h},h+\alpha(n)-2} &  j=2,\ldots, \ell-1, \\ 
         \noalign{\vskip4pt}
        \mathbf{\Xi}_{N,\emptyset, \ell+\alpha(n)-2} & j=\ell,
       \end{cases}
 \end{equation*}
    with $\alpha(1)=1$ and $\alpha(2)=0$. Denote by 
${\mathcal{U}}_1^{(\omega)}
$ and ${\mathcal{U}}_2^{(\omega)}$ the vectors in $\mathbb{R}^{\ell}$ defined in the same way as $\mathcal{U}_{1}$, $\mathcal{U}_{2}$ above with $\mathbf{\Xi}$ replaced by its weighted version  $\mathbf{\Xi}^{(\omega)}$. Denote the $\ell$-vectors ${\mathcal{U}}_n^{'}$, $n=3,\ldots,k+1$ by
\begin{align*}
        (\mathcal{U}_n')_j &:= 
    \begin{cases}
         \displaystyle\sum_{h=1}^{j+m-2} (-1)^{h+m-2} \mathbf{\Xi}_{N,\boldsymbol{\lambda}_{\ell+m-2-h,j+m-1-h},h} &  j=1,\ldots, \ell-1, \\ 
         \noalign{\vskip4pt}
         \mathbf{\Xi}_{N,\emptyset, \ell+m-2} & j=\ell.
       \end{cases}
\end{align*}
Then there are matrices $\mathbf{B}^{(\ell)}$, $\mathbf{Q}_1^{(\ell)}$, $\widetilde{\mathbf{Q}}_{1}^{(\ell)}$, $\widehat{\mathbf{Q}}_{2}^{(\ell)}$, $\widetilde{\mathbf{Q}}_{2}^{(\ell)}$, $\mathbf{Q}_{n}^{(\ell)}$, $n=3,\ldots,k+1$ as in Appendix \ref{app:scripts}, independent of the parameters $v,w$, such that
\begin{align*}
\mathbf{B}^{(\ell)}\mathcal{U}_1=\mathbf{Q}_1^{(\ell)} \mathbf{\Xi}_{N}^{(0)}[\ell-1;1]\: &,\;\;\;\;  \mathbf{B}^{(\ell)}{\mathcal{U}}_1^{(\omega)}=\widetilde{\mathbf{Q}}_1^{(\ell)} \mathbf{\Xi}_{N}^{(0)}[\ell-1;1],\\
\mathbf{B}^{(\ell)}\mathcal{U}_2=\widehat{\mathbf{Q}}_2^{(\ell)} \mathbf{\Xi}_{N}^{(0)}[\ell-2;1]\: &,\;\;\;\;  \mathbf{B}^{(\ell)}{\mathcal{U}}_2^{(\omega)}=\widetilde{\mathbf{Q}}_2^{(\ell)} \mathbf{\Xi}_{N}^{(0)}[\ell-2;1],\\ \noalign{\vskip4pt}
    \mathbf{B}^{(\ell)}\mathcal{U}_n'=  \: \mathbf{Q}_n^{(\ell)} & \mathbf{\Xi}_{N}^{(0)}[\ell+n-2;1].
\end{align*}
\end{lem}
\begin{rem}
We note that the proof of Lemma \ref{generalresult0325} relies only on the fact that $\mathbf{\Xi}_{N}$
is a Hankel determinant and it does not depend on the specific elements in its entries. 
\end{rem}
Next, we give a linear representation for the vector $(\mathbf{\Xi}_{N,\boldsymbol{\lambda}_{\ell,1}},\ldots, \mathbf{\Xi}_{N,\boldsymbol{\lambda}_{\ell,1}})^T$ in terms of vectors formed by Hankel determinants $\mathbf{\Xi}_{N}$ shifted by all partitions of $\ell-2,\ell-1,\ell+1,\ell+2,\ldots,\ell+k-1$ of hook forms and their derivatives. All the coefficients are functions of $t_1,t_2,\ldots,t_k$.

The following conclusion follows from Lemma \ref{generalresult0325} and a similar argument to that of \cite[Proposition 4.11]{AGKW}. This result is an analogue of \cite[Proposition 4.11]{AGKW}, where the matrices $\mathbf{Q}_{0}^{(\ell)}(s)$
and $\mathbf{Q}_{2}^{(\ell)}(s)$ are replaced by $\mathbf{Q}_{0}^{(\ell)}(v,w)$ and $\mathbf{Q}_{2}^{(\ell)}(w)$, respectively, as follows. Moreover, when $v=w=s$, $\mathbf{Q}_{0}^{(\ell)}(v,w)=\mathbf{Q}_{0}^{(\ell)}(s)
$.

\begin{prop}\label{prop:recursivevector}
Let $(N,k,\ell)\in\mathbb{N}^3$ with $N\geq 1, k\geq 2, \ell\geq 3$. Then, there are matrices $\mathbf{B}^{(\ell)}$, $\mathbf{Q}_n^{(\ell)}$ with $n=0,1,\ldots,k+1$, $n\neq 0,2$, as given in Lemma \ref{generalresult0325}, and $\mathbf{Q}_{0}^{(\ell)}(v,w)$,$\mathbf{Q}_{2}^{(\ell)}(w)$ as given in Appendix \ref{app:scripts}, all of which, except $\mathbf{Q}_{2}^{(\ell)}(w)$, have all entries independent of $N$, so that:
    \begin{align*}
       \mathbf{\Xi}^{(0)}_{N}&[\ell;1]= \frac{k t_k (-1)^{k+1} \mathbf{B}^{(\ell)}}{2(t_1+t_2)} \left(-\frac{N}{2} \frac{\partial}{\partial t_1} + \frac{N}{2}\right)\mathbf{\Xi}_{N}^{(1)}[\ell+k-2;k]  \\
        &+\mathbf{B}^{(\ell)} \Bigg(\frac{N}{2} \frac{\partial}{\partial t_1} - \frac{N}{2} + \frac{N}{2(t_1+t_2)}\Bigg(
    \sum_{m=3}^k  \mathcal{D}_{m-1,m} \frac{\partial}{\partial t_{m-1}} + kt_k \frac{\partial}{\partial t_k} \Bigg) \Bigg) \mathbf{\Xi}_{N}^{(1)}[\ell-1;1]
    \\ & + \frac{1}{2(t_1+t_2)} \sum_{m=1}^{k-2}  \mathcal{D}_{m+1,m+2}  (-1)^m \mathbf{Q}_{m+2}^{(\ell)}   \mathbf{\Xi}_{N}^{(0)}[\ell+m;1]\\
    &+ \left( \frac{t_1}{t_1+t_2} \mathbf{Q}_1^{(\ell)} + \frac{N}{2(t_1+t_2)} {\mathbf{Q}}_0^{(\ell)}(v,w)  \right) \mathbf{\Xi}_{N}^{(0)}[\ell-1;1]
+ \frac{N}{2(t_1+t_2)} \mathbf{Q}_2^{(\ell)}(w) \mathbf{\Xi}_{N}^{(0)}[l-2;1] \\ &- \frac{\mathbf{B}^{(\ell)}}{2(t_1+t_2)} \sum_{m=0}^{k-3}\mathcal{D}_{m+2,m+3}(-1)^m (-\frac{N}{2} \frac{\partial}{\partial t_1} + \frac{N}{2}) \mathbf{\Xi}_{N}^{(1)}[\ell+m;m+2]\\ & + \frac{1}{2(t_1+t_2)}  (-1)^{k-1} k t_k \mathbf{Q}_{k+1}^{(l)}\mathbf{\Xi}_{N}^{(0)}[\ell+k-1;1]\\& +\frac{\mathbf{B}^{(\ell)} Nkt_{k}}{2(t_{1}+t_{2})}\sum_{h=2}^{k-1} (-1)^{k+h} \frac{\partial }{\partial t_h} \mathbf{\Xi}_{N}^{(1)}[\ell+k-1-h;k+1-h] \\& + \frac{\mathbf{B}^{(\ell)} N}{2(t_1+t_2)} \sum_{m=4}^k (-1)^{m-1} \mathcal{D}_{m-1,m} \sum_{h=2}^{m-2} (-1)^h \frac{\partial }{\partial t_h} \mathbf{\Xi}_{N}^{(1)}[\ell+m-2-h;m-h],
    \end{align*}
    where we denote
    \begin{equation*}
        \mathcal{D}_{p,q}= \mathcal{D}_{p,q}(t_1,\ldots,t_k):=pt_p-qt_q.
    \end{equation*}
\end{prop}
Now for the $t_{q}$-derivatives of  $\mathbf{\Xi}_{N}(v,w; t_{1},\ldots,t_{k})$, we shall expand them as a power series in $t_{q},\ldots,t_{k}$, for $q=2,\ldots,k$. The basic idea is to iteratively use Proposition \ref{generalrecursiveforshiftbypartition0325}. Before presenting the statement, let us introduce a two parameter ($v$,$w$)-version of $\mathsf{D}_{i}$ as given in \cite[Section 4.3]{AGKW}. Note that they coincide when $v=w$. Define
\begin{align*}
  &\mathsf{Z}_{i}(v,w;t_{1},\ldots, t_{k})\\
  :=&N\left\{(N+w-1) \mathbf{\Xi}_{N,\emptyset,i}-\mathbf{\Xi}^{(\omega)}_{N,\emptyset,i} +(1-2N-v-w)\mathbf{\Xi}_{N,\emptyset, i+1}\nonumber
+\mathbf{\Xi}^{(\omega)}_{N,\emptyset, i+1}+2\frac{t_1}{N}\mathbf{\Xi}_{N,\emptyset,i+1}\right\},
\end{align*}
where the parameters $(v,w)$ in the above are the same parameter that are used to define $\mathbf{\Xi}_{N,\boldsymbol{\lambda},m}$ and $\mathbf{\Xi}_{N,\boldsymbol{\lambda},m}^{(\omega)}$. Sometimes we write $\mathsf{Z}_{i}(v,w)$ for  brevity. Our goal is to study the derivatives of these functions with respect to $t_{2},\ldots,t_{k}$ evaluated at $t_{2}=\cdots=t_{k}=0$. Thus we define
\bea\label{definitionofDi}
\mathsf{Z}_{i}^{(n_{2},\ldots,n_{k})}(v,w;t_{1})
:=\frac{\partial^{n_{2}}}{\partial t_{2}^{n_{2}}}\cdots\frac{\partial^{n_{k}}}{\partial t_{k}^{n_{k}}}\mathsf{Z}_{i}(v,w;t_{1},\ldots,t_{k})\Big|_{t_{2}=\ldots=t_{k}=0},
\eea
\bea\label{definitionofpsi1}
\mathcal{S}_{N}^{(i_{2},n_{3},\ldots,n_{k})}(v,w;t_{1}):= \frac{\partial^{i_{2}}}{\partial t_{2}^{i_{2}}}\frac{\partial^{n_{3}}}{\partial t_{3}^{n_{3}}}\cdots\frac{\partial^{n_{k}}}{\partial t_{k}^{n_{k}}}\mathbf{\Xi}_{N}(v,w;t_{1},\ldots,t_{k})\Bigg|_{t_{2}=\ldots=t_{k}=0},
\eea
and
\bea\label{definitionofpsi2}
\mathcal{S}_{N,\ell,q}^{(i_{2},n_{3},\ldots,n_{k})}(v,w;t_{1}):= \frac{\partial^{i_{2}}}{\partial t_{2}^{i_{2}}}\frac{\partial^{n_{3}}}{\partial t_{3}^{n_{3}}}\cdots\frac{\partial^{n_{k}}}{\partial t_{k}^{n_{k}}}\mathbf{\Xi}_{N,\boldsymbol{\lambda}_{\ell,q}}(v,w;t_{1},\ldots,t_{k})\Bigg|_{t_{2}=\ldots=t_{k}=0}
\eea
with $\ell\geq 3$ and $1\leq q\leq 3$. When there is no risk of confusion, we use the notation
$\mathsf{Z}_{i}^{(n_{2},\ldots,n_{k})}(v,w)$, $\mathcal{S}_{N}^{(i_{2},n_{3},\ldots,n_{k})}(v,w)$ and $\mathcal{S}_{N,\ell,q}^{(i_{2},n_{3},\ldots,n_{k})}(v,w)$ as shorthands. We are now finally ready to prove Proposition \ref{mainproposition1}.
%
%
%
%
%
%
%
\begin{proof}[Proof of Proposition \ref{mainproposition1}]
After taking partial derivatives on both sides of the expression in Proposition \ref{prop:recursivevector}, identity \eqref{partialderivativewithrespecttot1}, and the expressions of $\mathbf{\Xi}_{N,\boldsymbol{\lambda}_{2,1}}$ and $\mathbf{\Xi}_{N,\boldsymbol{\lambda}_{2,2}}$ in Proposition \ref{representationshiftedby20324} with respect to $t_{2},\ldots,t_{k}$ and let $t_{2}=\cdots=t_{k}=0$, we have an analogue of formula (59) in \cite[Proposition 4.13]{AGKW}. Here the slight difference is that we set $\mathcal{L}_{N,\ell,q}^{(i_{2},n_{3},\ldots,n_{k})}(t_{1})$ to $\mathcal{S}_{N,\ell,q}^{(i_{2},n_{3},\ldots,n_{k})}(v,w;t_{1})$, as defined in (\ref{definitionofpsi2}), and set $\mathbf{Q}_{0}^{(\ell)}$ and $\mathbf{Q}_{2}^{(\ell)}$ to $\mathbf{Q}_{0}^{(\ell)}(v,w)$ and $\mathbf{Q}_{2}^{(\ell)}(w)$, respectively, as given in Appendix \ref{app:scripts}. By iteratively applying Proposition \ref{generalrecursiveforshiftbypartition0325}, we expand $\partial{\mathbf{\Xi}_{N}(v,w; t_{1},\ldots,t_{k})}/{\partial  t_{q}}$ as a power series in the variables $t_{q},\ldots,t_{k}$, which is almost identical to the statement in \cite[Lemma 4.14]{AGKW}, except that we replace $\mathbf{\Psi}_N(t_{1},\ldots,t_{k})$ and $\mathsf{D}_{i}(t_{1},\ldots,t_{k})$ with their two-parameter versions, that is, with $\mathbf{\Xi}_N(v,w;t_{1},\ldots,t_{k})$ and $\mathsf{Z}_{i}(v,w;t_{1},\ldots,t_{k})$, respectively. Hence, we obtain an analogue of the expression in \cite[Proposition 4.15]{AGKW} by setting $
\mathsf{D}_{i}^{(n_{2},\ldots,n_{k})}(t_{1})$ to $
\mathsf{Z}_{i}^{(n_{2},\ldots,n_{k})}(v,w; t_{1})$, as defined in (\ref{definitionofDi}). Then by \cite[Lemma 4.17]{AGKW} and the recursive process described in the proof of \cite[Theorem 1.9]{AGKW}, we have 
\begin{equation}\label{defofmathfrakS0325}
\mathcal{S}_{N}^{(n_2,\ldots, n_k)}(v,w;t_{1})
=\frac{1}{t_{1}^{\sum_{q=2}^{k}qn_{q}-1}}
\sum_{m=0}^{\sum_{q=2}^{k}(q-1)n_{q}}
t_1^{m-1}
P_{m}^{(N,v,w)}(t_{1})\frac{\mathrm{d}^{m}}{\mathrm{d}t_{1}^{m}}\mathbf{\Xi}_{N}(v,w;t_{1},0,\ldots,0),
\end{equation}
for polynomials $P_{m}^{(N,v,w)}$ as described in the statement of this proposition when $t_{1}>0$. Note that the left hand-side of the above equation is continuous at $t_{1}=0$, so the right-hand side has a limit when $t_{1}\rightarrow 0$. Then the above equation holds for $t_{1}\geq 0$. It is easy to check that equation (\ref{defofmathfrakS0325}) is equivalent to equation (\ref{mainstructure1}). 
Moreover, as in the previous analysis, when $v=w=s\in \mathbb{R}$, all extensions of the notation involving the new parameters $v$ and $w$ coincide with the corresponding notation used in the proof of \cite[Theorem 1.9]{AGKW}. Consequently, $P_{m}^{(N,v,w)}$ is reduced to $P_{m}^{(s,N)}$ as stated in \cite[Theorem 1.9]{AGKW}. The proof is complete.
\end{proof}

\section{Explicit forms of some matrices}
\label{app:scripts}
As mentioned previously, here we give explicit formulae for the matrices that were used in the proof of proposition \ref{mainproposition1}. The matrices 
$\mathbf{B}^{(\ell)}$, $\mathbf{Q}_n^{(\ell)}$ with $n=0,1,\ldots,k+1$, $n\neq 0,2$, and
$\widetilde{\mathbf{Q}}_{1}^{(\ell)}$,
$\widehat{\mathbf{Q}}_{2}^{(\ell)}$,
$\widetilde{\mathbf{Q}}_{2}^{(\ell)}$  are the same as those in \cite[Appendix]{AGKW}. The matrices $\mathbf{Q}_{0}^{(\ell)}:=\mathbf{Q}_{0}^{(\ell)}(v,w)$ and $\mathbf{Q}_{2}^{(\ell)}:=\mathbf{Q}_{2}^{(\ell)}(w)$ are given as follows: for $i=1,\ldots,\ell$, $j=1,\ldots,\ell-1$, we define
\begin{align*}
 \left(\mathbf{Q}_{0}^{(\ell)}(v,w)\right)_{ij}&=
\begin{cases}\label{definitionofC1}
(-1)^{i+j}\frac{j(\ell-j-2)+1-v-w}{j(j+1)}& \text{if } i\leq j\leq \ell-1;\\
\frac{j(j+v+w)-\ell+1}{j+1}& \text{if } j=i-1; \\
  0 & \text{if } j<i-1.
\end{cases}
\end{align*}
And for $i=1,\ldots,\ell$, $j=1,\ldots,\ell-2$, we define
\begin{align*}
\left(\mathbf{Q}_{2}^{(\ell)}(w)\right)_{ij}  &= \begin{cases}
(-1)^{i+j}\frac{(w-2)N+j(j+3)-(j+2)(\ell-2)+2(w-1)}{(j+1)(j+2)}  & \text{if }  i-1\leq j \leq \ell-2; \\
\frac{(j+w)(j-N)}{j+2} & \text{if } j=i-2; \\
0 & \text{if } 1\leq j <  i-2.
\end{cases}
\end{align*}



\end{document}